\documentclass[11pt]{article}
\usepackage[margin=.8in,left=.8in]{geometry}
\usepackage{amsmath}
\usepackage{amsfonts}
\usepackage{amssymb}
\usepackage{amsthm}
\usepackage{mathtools}
\usepackage{subcaption}
\usepackage{graphicx}
\usepackage{color}
\usepackage{xcolor}
\usepackage{xfrac}
\usepackage{stmaryrd}

\usepackage[safe]{tipa} % for \esh
\usepackage[utf8]{inputenc}
\usepackage{rotating}
\usepackage{hyperref}
\usepackage[all]{xy}

\usepackage{mathtools}
\usepackage{rotating}

\usepackage{stmaryrd}
\usepackage{multirow}

\usepackage{lmodern} % for  \bf and \it at the same time

\usepackage{tikz}
\usetikzlibrary{decorations.pathmorphing}
\usepackage{tikz-cd}
\usetikzlibrary{calc}
\usetikzlibrary{matrix}
\usetikzlibrary{positioning,arrows,shapes}
\usetikzlibrary{intersections,shapes.arrows, calc}
\usetikzlibrary{patterns}
\usetikzlibrary{mindmap} % LATEX and plain TEX
\usetikzlibrary[mindmap] % ConTEXt

\usepackage{lipsum}
\usepackage{adjustbox}

\definecolor{redi}{RGB}{255,38,0}
\definecolor{redii}{RGB}{200,50,30}
\definecolor{yellowi}{RGB}{255,251,0}
\definecolor{bluei}{RGB}{0,150,255}
\definecolor{blueii}{RGB}{135,247,210}
\definecolor{blueiii}{RGB}{91,205,250}
\definecolor{blueiv}{RGB}{115,244,253}
\definecolor{bluev}{RGB}{1,58,215}
\definecolor{orangei}{RGB}{220,160, 20}
\definecolor{orangeii}{RGB}{240,90, 10}
\definecolor{yellowii}{RGB}{222,247,100}
\definecolor{greeni}{RGB}{85,102,0}
\definecolor{greenii}{RGB}{20,140,10}
\definecolor{navy}{RGB}{17, 10, 102}
\definecolor{brown}{RGB}{60, 40, 0}
\definecolor{oxford}{RGB}{0, 0, 100}

\definecolor{plum}{rgb}{0.36078, 0.20784, 0.4}
\definecolor{chameleon}{rgb}{0.30588, 0.60392, 0.023529}
\definecolor{cornflower}{rgb}{0.12549, 0.29020, 0.52941}
\definecolor{scarlet}{rgb}{0.8, 0, 0}
\definecolor{brick}{rgb}{0.64314, 0, 0}
\definecolor{sunrise}{rgb}{0.80784, 0.36078, 0}
\definecolor{lightblue}{rgb}{0.15,0.35,0.75}
\definecolor{carolina}{RGB}{153, 186, 221}
\definecolor{darkblue}{rgb}{0.05,0.25,0.65}

\usepackage{multirow}

\usepackage{stmaryrd}    % for \sslash

\usepackage{enumerate} %Roman enumerate

\usepackage{helvet}   %vertical spacing in tabular

\usepackage{bigints}  %big integral

\usepackage{scalerel} %scaling math

\usepackage{mathptmx}
\usepackage{amsmath}
\usepackage{graphicx}

\usepackage{floatflt}  %Floating tables
\usepackage{array}     % specifying size of tables
\newcolumntype{L}[1]{>{\raggedright\let\newline\\\arraybackslash\hspace{0pt}}m{#1}}
\newcolumntype{C}[1]{>{\centering\let\newline\\\arraybackslash\hspace{0pt}}m{#1}}
\newcolumntype{R}[1]{>{\raggedleft\let\newline\\\arraybackslash\hspace{0pt}}m{#1}}

\newcommand{\boldpi}{\mbox{$\pi$\hspace{-6.5pt}$\pi$}}
\newcommand{\gt}{>}

\makeatletter
\newcommand{\raisemath}[1]{\mathpalette{\raisem@th{#1}}}
\newcommand{\raisem@th}[3]{\raisebox{#1}{$#2#3$}}
\makeatother

\usepackage{tabularx}   %For sizing tabular entries

\usepackage[new]{old-arrows}   %For \longhookrightarrow

\newdir{> }{{}*!/10pt/@{>}}

\makeatletter
\newif\if@sup
\newtoks\@sups
\def\append@sup#1{\edef\act{\noexpand\@sups={\the\@sups #1}}\act}%
\def\reset@sup{\@supfalse\@sups={}}%
\def\mk@scripts#1#2{\if #2/ \if@sup ^{\the\@sups}\fi \else%
  \ifx #1_ \if@sup ^{\the\@sups}\reset@sup \fi {}_{#2}%
  \else \append@sup#2 \@suptrue \fi%
  \expandafter\mk@scripts\fi}
\def\tensor#1#2{\reset@sup#1\mk@scripts#2_/}
\def\multiscripts#1#2#3{\reset@sup{}\mk@scripts#1_/#2%
  \reset@sup\mk@scripts#3_/}
\makeatother

\makeatletter
\newbox\slashbox \setbox\slashbox=\hbox{$/$}
\def\itex@pslash#1{\setbox\@tempboxa=\hbox{$#1$}
  \@tempdima=0.5\wd\slashbox \advance\@tempdima 0.5\wd\@tempboxa
  \copy\slashbox \kern-\@tempdima \box\@tempboxa}
\def\slash{\protect\itex@pslash}
\makeatother

\def\clap#1{\hbox to 0pt{\hss#1\hss}}
\def\mathllap{\mathpalette\mathllapinternal}
\def\mathrlap{\mathpalette\mathrlapinternal}
\def\mathclap{\mathpalette\mathclapinternal}
\def\mathllapinternal#1#2{\llap{$\mathsurround=0pt#1{#2}$}}
\def\mathrlapinternal#1#2{\rlap{$\mathsurround=0pt#1{#2}$}}
\def\mathclapinternal#1#2{\clap{$\mathsurround=0pt#1{#2}$}}

\let\oldroot\root
\def\root#1#2{\oldroot #1 \of{#2}}
\renewcommand{\sqrt}[2][]{\oldroot #1 \of{#2}}

\DeclareSymbolFont{symbolsC}{U}{txsyc}{m}{n}
\SetSymbolFont{symbolsC}{bold}{U}{txsyc}{bx}{n}
\DeclareFontSubstitution{U}{txsyc}{m}{n}

\DeclareSymbolFont{stmry}{U}{stmry}{m}{n}
\SetSymbolFont{stmry}{bold}{U}{stmry}{b}{n}

\DeclareFontFamily{OMX}{MnSymbolE}{}
\DeclareSymbolFont{mnomx}{OMX}{MnSymbolE}{m}{n}
\SetSymbolFont{mnomx}{bold}{OMX}{MnSymbolE}{b}{n}
\DeclareFontShape{OMX}{MnSymbolE}{m}{n}{
    <-6>  MnSymbolE5
   <6-7>  MnSymbolE6
   <7-8>  MnSymbolE7
   <8-9>  MnSymbolE8
   <9-10> MnSymbolE9
  <10-12> MnSymbolE10
  <12->   MnSymbolE12}{}

\makeatletter
\def\Decl@Mn@Delim#1#2#3#4{%
  \if\relax\noexpand#1%
    \let#1\undefined
  \fi
  \DeclareMathDelimiter{#1}{#2}{#3}{#4}{#3}{#4}}
\def\Decl@Mn@Open#1#2#3{\Decl@Mn@Delim{#1}{\mathopen}{#2}{#3}}
\def\Decl@Mn@Close#1#2#3{\Decl@Mn@Delim{#1}{\mathclose}{#2}{#3}}
\Decl@Mn@Open{\llangle}{mnomx}{'164}
\Decl@Mn@Close{\rrangle}{mnomx}{'171}
\Decl@Mn@Open{\lmoustache}{mnomx}{'245}
\Decl@Mn@Close{\rmoustache}{mnomx}{'244}
\makeatother

\makeatletter
\DeclareRobustCommand\widecheck[1]{{\mathpalette\@widecheck{#1}}}
\def\@widecheck#1#2{%
    \setbox\z@\hbox{\m@th$#1#2$}%
    \setbox\tw@\hbox{\m@th$#1%
       \widehat{%
          \vrule\@width\z@\@height\ht\z@
          \vrule\@height\z@\@width\wd\z@}$}%
    \dp\tw@-\ht\z@
    \@tempdima\ht\z@ \advance\@tempdima2\ht\tw@ \divide\@tempdima\thr@@
    \setbox\tw@\hbox{%
       \raise\@tempdima\hbox{\scalebox{1}[-1]{\lower\@tempdima\box
\tw@}}}%
    {\ooalign{\box\tw@ \cr \box\z@}}}
\makeatother

\makeatletter
\def\udots{\mathinner{\mkern2mu\raise\p@\hbox{.}
\mkern2mu\raise4\p@\hbox{.}\mkern1mu
\raise7\p@\vbox{\kern7\p@\hbox{.}}\mkern1mu}}
\makeatother

\def\1{{\bf 1}}

\def\<{\langle}
\def\>{\rangle}

\renewcommand{\(}{\begin{equation}}
\renewcommand{\)}{\end{equation}}
\newcommand{\bea}{\begin{eqnarray*}}
\newcommand{\eea}{\end{eqnarray*}}

\usepackage{cleveref}

\crefformat{section}{\S#2#1#3} % see manual of cleveref, section 8.2.1
\crefformat{subsection}{\S#2#1#3}
\crefformat{subsubsection}{\S#2#1#3}

\theoremstyle{italics}
\newtheorem{theorem}{Theorem}[section]
\newtheorem{lemma}[theorem]{Lemma}
\newtheorem{prop}[theorem]{Proposition}

\theoremstyle{definition}
\newtheorem{defn}[theorem]{Definition}

\newtheorem{example}[theorem]{Example}

\newtheorem{remark}[theorem]{Remark}
\newtheorem{note[theorem]}{Note}

\newtheorem{observation}[theorem]{\bf Observation}

\usepackage{amsfonts}
\usepackage{colortbl}

 \newcommand{\medoplus}{\ensuremath{\vcenter{\hbox{\scalebox{1.5}{$\oplus$}}}}}

\newcommand{\mapsdown}{\mbox{$\!\!\!\!$\begin{rotate}{-90}$\!\!\!\!\!\!\!\!\mapsto$\end{rotate}}}

\newcommand{\underoverset}[3]{\underset{#1}{\overset{#2}{#3}}}

\definecolor{darkblue}{rgb}{0.05,0.25,0.65}
\definecolor{darkgreen}{rgb}{0.00,0.85,0.1}
\definecolor{plum}{rgb}{0.36078, 0.20784, 0.4}

\usepackage{hyperref}

\begin{document}

\title{
  Differential Cohomotopy implies intersecting brane observables
  \\
  via configuration spaces and chord diagrams
}

\author{ Hisham Sati, Urs Schreiber}

\maketitle

\begin{abstract}
We introduce a differential refinement of Cohomotopy cohomology theory,
defined on Penrose diagram spacetimes,
whose cocycle spaces are unordered configuration spaces of points.
First we prove that brane charge quantization
in this differential 4-Cohomotopy theory implies intersecting
$p\perp(p+2)$-brane moduli given by
ordered configurations of points in the transversal 3-space.
Then we show that the higher (co-)observables on these brane moduli,
conceived as the (co-)homology of the Cohomotopy cocycle space,
are given by weight systems on horizontal chord diagrams and
reflect a multitude of effects expected in the
microscopic quantum theory
of $\mathrm{D}p\perp \mathrm{D}(p+2)$-brane intersections:
condensation to stacks of coincident branes and
their Chan-Paton factors,
fuzzy funnel states and M2-brane 3-algebras,
$\mathrm{AdS}_3$-gravity observables
and supersymmetric indices of Coulomb branches,
M2/M5-brane bound states in the BMN matrix model
and the Hanany-Witten rules,
as well as gauge/gravity duality between all these.
We discuss this in the context of the hypothesis
that the M-theory C-field
is charge-quantized in Cohomotopy theory.
\end{abstract}

\newpage

 \tableofcontents

\newpage

%%%%%%%%%%%%%%%%%%%%%%%%%%%%%%%%%%%%%%%%%%%%%%
\section{Introduction and overview}
%%%%%%%%%%%%%%%%%%%%%%%%%%%%%%%%%%%%%%%%%%%%%%

\noindent {\bf The general open problem.}
The rich physics expected on coincident and intersecting branes
(reviewed in \cite{IU12}), which geometrically engineer non-perturbative quantum gauge field theories
\cite{KatzVafa97}\cite{HananyWitten97}\cite{GaiottoWitten08}
(reviewed in \cite{Karch98}\cite{GiveonKutasov99}\cite{Fazzi17})
close to quantum chromodynamics
\cite{Witten98}\cite{SakaiSugimoto04}\cite{SakaiSugimoto05}
(reviewed in \cite{Rebhan14}\cite{Sugimoto16})
and model quantum microstates accounting for black hole entropy
\cite{SV96}\cite{CallanMaldacena96}
(reviewed in \cite{Kraus06}\cite{Sen07}),
has come to be center stage in string theory -- or rather in the
``theory formerly known as strings'' \cite{Duff96}.
Despite all that is known about D-branes
from the two limiting cases of  {\bf (a)} string perturbation theory
and {\bf (b)} worldvolume gauge theory,
an actual comprehensive
theory of non-perturbative brane physics, namely an actual formulation of
\emph{M-theory} \cite{Duff99B}, still
remains an open problem
\cite[6]{Duff96}\cite[p. 2]{HoweLambertWest97}\cite[p. 6]{Duff98}\cite[p. 2]{NicolaiHelling98}\cite[p. 330]{Duff99B}\cite[12]{Moore14}\cite[p. 2]{ChesterPerlmutter18}\cite{Witten19}\footnote{\cite{Witten19} at 21:15: ``I actually believe that string/M-theory is on the right
track toward a deeper explanation. But at a very fundamental level it's not
well understood. And I'm not even confident that we have a good concept of
what sort of thing is missing or where to find it."}\cite{Duff19}\footnote{\cite{Duff19} at 17:04: ``The problem we face is that we have a patchwork understanding of M-theory, like a quilt. We understand this corner and that corner, but what’s lacking is the overarching big picture. Directly or indirectly, my research hopes to explain what M-theory really is. We don't know what it is.''}.
The lack of such a genuine
theory of brane physics has recently surfaced in a debate
about the validity of D3-brane constructions that had dominated
the discussion in a large part of the community for the last 15 years;
see \cite{DanielssonVanRiet18}\cite[p. 14-22]{Banks19}.

\medskip

\noindent {\bf Hypothesis H.}
Based on a re-analysis of the super $p$-brane WZW terms
from the point of view of homotopy theory
\cite{FSS13}\cite{FSS15}\cite{FSS16a}\cite{FSS16b}\cite{GaugeEnhancement}
(reviewed in \cite{FSS19a}),
we have recently
formulated a concrete hypothesis about (at least part of) the mathematical
nature of M-theory:
This \hyperlink{HypothesisH}{\it Hypothesis H}
\cite[2.5]{Sati13}\cite{FSS19b} asserts that in M-theory the C-field
of 11d supergravity \cite{CJS78} is
charge-quantized \cite{Freed00}\cite{tcu}
in the non-abelian generalized cohomology theory
called \emph{J-twisted Cohomotopy theory}.
This hypothesis turns out to imply
\cite{FSS19b}\cite{FSS19c}\cite{SS19a}\cite{SS19b}
a wealth of subtle
topological effects expected in string/M-theory.
This suggests that it is a correct proposal
about the mathematics underpinning M-theory,
at least in the topological sector.

\medskip
\noindent {\bf Differential refinement.}
In this article we take a step beyond the topological sector
and investigate to which extent a geometrically (``differentially'')
refined form (cf. \cite{FSS15})
of \hyperlink{HypothesisH}{Hypothesis H} leads to
the emergence/derivation of
expected phenomena on coincident and intersecting branes.
For exposition see also \cite{Schreiber20}.

\medskip

\noindent First, our main {\bf mathematical observations} here are the following
(\cref{ChargesInCohomotopy}):

%\vspace{-.4cm}
%
%\hspace{-.9cm}
%\begin{tabular}{ll}
%
%\begin{minipage}[l]{10cm}

\begin{enumerate}[{\bf (1)}]
\vspace{-4mm}
\item A differential refinement of Cohomotopy cohomology theory
is given by \emph{un-ordered configuration spaces of points}.

\vspace{-3mm}
\item  The fiber product of such differentially refined Cohomotopy cocycle spaces describing  $\mathrm{D}6\perp \mathrm{D}8$-brane intersections
is homotopy-equivalent to the \emph{ordered configuration space of points} in the transversal space.

\vspace{-3mm}
\item The higher observables on this moduli space are equivalently
weight systems on horizontal chord diagrams.
\end{enumerate}

\vspace{-3mm}
$$
\xymatrix@C=5em@R=1.2em{
  \fbox{
    \hspace{-.35cm}
    \begin{tabular}{c}
      Cohomotopy cohomology
      \\
      theory \cref{ChargesInCohomotopy}
    \end{tabular}
    \hspace{-.35cm}
  }
  \ar@{~>}[rr]^-{
    \mbox{
      \footnotesize
      \begin{tabular}{c}
        {\color{blue}
        Differential}
        \\
        {\color{blue}
        refinement}
        \\
        Prop. \ref{DifferentialCohomotopyOnPenroseDiagrams}
      \end{tabular}
    }
  }
  \ar@{~>}[dd]_-{\!\!\!\!
    \mbox{
      \footnotesize
      \begin{tabular}{c}
        \hyperlink{HypothesisH}{Hypothesis H},
        \\
        \& Prop. \ref{FinalResult}
      \end{tabular}
    }
  }
  &&
  \fbox{
    \hspace{-.35cm}
    \begin{tabular}{c}
      Configuration spaces
      \\
      of points \cref{ConfigurationSpacesOfPoints}
    \end{tabular}
    \hspace{-.35cm}
  }
  \ar@{~>}[dd]^-{ \!\!\!\!\!\!\!
    \mbox{
      \footnotesize
      \begin{tabular}{c}
        {\color{blue}
        Cohomology}
        \\
        Prop. \ref{HigherObservablesOnIntersectingBranesAreWeightSystems}
      \end{tabular}
    }
  }
  \ar@{~>}[ddll]|-{
    \mbox{
      \footnotesize
      \begin{tabular}{c}
        {\color{blue}
        Fiber product}
        \\
        Prop. \ref{OrderedUnlabeledAsFiberProduct}
      \end{tabular}
    }
  }
  \\
  \\
  \fbox{Intersecting branes}
  \ar@{<~}[rr]_-{
    \mbox{
     \footnotesize
     {\color{blue}
     Observables
     }
     \cref{WeightSystemsAsObservablesOnIntersectingBranes}
    }
  }
  &&
  \fbox{
    \hspace{-.35cm}
    \begin{tabular}{c}
      Weight systems on
      \\
      Chord diagrams \cref{WeightSystemsOnChordDiagrams}
    \end{tabular}
    \hspace{-.35cm}
  }
}
%}
%}
%\end{tabular}
$$

\noindent Second, we make the {\bf string-theoretic observation} (\cref{WeightSystemsAsObservablesOnIntersectingBranes}) that
these weight systems on horizontal chord diagrams, when
regarded as higher observables
reflect a multitude of effects expected on brane
intersections in string theory.

\medskip
This leads to an understanding and clarification of relations among various physical concepts
and points to a unifying theme,
relying on constructions from seemingly distinct mathematical
areas which are brought together -- see \hyperlink{FigureS}{Figure 1}.

\newpage

\hspace{-1.4cm}
\fbox{\!\!\!\!\!\!\!\!\!\!\!\!\!
\hypertarget{FigureS}{}
\hspace{-1.1cm}
$
  \xymatrix@R=16pt{
    &
      \mbox{
        \color{greenii}
\begin{tabular}{c}
          M-theoretic observables
          \\
          on $\mathrm{D}6\perp\mathrm{D}8$-configurations
        \end{tabular}
      }
    \ar@{~>}[ddd]|>>>>>>>>>>>>>>>>>>>>>>>>>>>>>>>>>{
      \fbox{
        \hspace{-.5cm}
        \scalebox{.8}{
        \begin{tabular}{c}
          implications of
          \\
          \hyperlink{HypothesisH}{Hypothesis H}
        \end{tabular}
        }
        \hspace{-.3cm}
      }
    }
    &
    \overset{
      \mathclap{
     }
    }{
      \mathrm{Obs}^\bullet_{\mathrm{D}6 \perp \mathrm{D}8}
    }
    \mathrlap{
      \;\;
      \scalebox{.7}{
        Def. \ref{HigherObsOnD6D8Intersections}
      }
    }
    \ar@{}[d]|-{
      \mathllap{
      }
      \;\;\;\;
      \begin{rotate}{-90}
        $
        \mathclap{
          \!\!\!\coloneqq
        }
        $
      \end{rotate}
      \;\;\;\;\;
      \mathrlap{
        \eqref{HigherObservablesOnD6D8Intersections}
      }
    }
    \ar@{<~}[l]_-{
      \mbox{
      \hyperlink{HypothesisH}{\it Hypothesis H}
      }
    }
    &
    \\
    &
    &
    \mathclap{
    {
    \overset{
      \mathclap{
      \mbox{
        \tiny
        \color{blue}
        \begin{tabular}{c}
          Cohomology
        \end{tabular}
      }
      \;\;\;\;\;\;\;
      }
    }{
      H^\bullet
    }\;
    \Big(
      \;
      \underset{
        \mathclap{
        \mbox{
          \tiny
          \color{blue}
          \begin{tabular}{c}
            Loop spaces of
            \\
            connected components
          \end{tabular}
        }
        \!\!\!\!\!\!
        \scalebox{.7}{
          \eqref{TopologicalPhaseSpace}
        }
        }
      }{
        \underbrace{
          \underset{[c]}{\sqcup} \; \Omega_c
        }
      }
      \;
      \,
      \overset{
        \mathrlap{
        \mbox{
          \hspace{-20pt}
          \tiny
          \begin{tabular}{c}
            {\color{blue}
            Differential 4-Cohomotopy}
            \\
{\color{blue}                       cocycle space
            }
            (Prop. \ref{DifferentialCohomotopyOnPenroseDiagrams})
          \end{tabular}
        }
        }
      }{
      \overbrace{
      \boldpi^4_{\mathrm{diff}}
      }
      }
      \big(\;
        \underset{
          \mathclap{
          \mbox{
            \tiny
            \color{blue}
            \begin{tabular}{c}
              Transversal space
              \\
              to codim=3 branes
            \end{tabular}
          }
          }
        }{
        \underbrace{
        (\mathbb{R}^3)^{\mathrm{cpt}}
        \wedge
        (\mathbb{R}^1)_+
        }
        }
        \,
        \underset{
          \mathclap{
          \raisebox{-13pt}{
          \scalebox{.7}{
            \eqref{PenroseDiagramSpaces}
          }
          }
          }
        }{
          \cup
        }
        \,
        \underset{
          \mathclap{
          \mbox{
            \tiny
            \color{blue}
            \begin{tabular}{c}
              Transversal space
              \\
              to codim=1 branes
            \end{tabular}
          }
          }
        }{
        \underbrace{
        (\mathbb{R}^3)_+
        \wedge
        (\mathbb{R}^1)^{\mathrm{cpt}}
        }
        }
      \; \big)
    \Big)
    }
    }
    \ar@{}[d]|-{
      \begin{rotate}{-90}
        $
        \mathclap{
          \!\!\!\!\!\xleftrightarrow{\;\;\simeq\;\;}
        }
        $
      \end{rotate}
      \mathrlap{
        \;\;\;\;\;\;\; \scalebox{.7}{Prop. \ref{FinalResult}}
      }
    }
    \\
    &
    &
    \overset{
      \mathclap{
      \mbox{
        \tiny
        \color{blue}
        \begin{tabular}{c}
          Ordinary
          \\
          cohomology
        \end{tabular}
      }
      \;\;\;\;\;\;\;\;\;
      }
    }{
    H^\bullet
    }
    \Big(\;
      \underset{
        \mathclap{
        \mbox{
          \tiny
          \color{blue}
          %\begin{tabular}{c}
          %  loop spaces of
          %  connected components
          %\end{tabular}
        }
        }
      }{
      {
      \underset{
        N_{\mathrm{f}} \in \mathbb{N}
      }{\sqcup}
      \Omega
      }
      }
      \;\;
      \overset{
        \mathclap{
        \mbox{
          \tiny
          \begin{tabular}{c}
            $\phantom{a}$
            \\
            \color{blue}
            Configuration space
            \\
            \color{blue}
            of $N_{\mathrm{f}}$ ordered points
          \end{tabular}
        }
        \mathrlap{
          \!\!\!\!
          \raisebox{-3pt}{
          \scalebox{.7}{(Def. \ref{ConfigurationSpaces})}
          }
        }
        }
      }{
      \overbrace{
      \underset{
        {}^{\{1,\cdots, N_{\mathrm{f}}\}}
      }{
        \mathrm{Conf}(\mathbb{R}^3)
      }
      }
      } \;
    \Big)
    \ar@{}[d]|-{
      \begin{rotate}{-90}
        $
        \mathclap{
          \!\!\!\!\!\!\xleftrightarrow{\;\;\simeq\;\;}
        }
        $
      \end{rotate}
      \mathrlap{
        \;\;\;\;\;\;
        \scalebox{.7}{
          Prop. \ref{HigherObservablesOnIntersectingBranesAreWeightSystems}
        }
      }
    }
    \;
    \ar@{}[r]|-{
      \overset{
        \xleftrightarrow{\;(-)^\ast\;}
      }{
        \mbox{
          \tiny
          \begin{tabular}{c}
            \color{greenii} duality
            \\
            \cref{DualityIsDuality}
          \end{tabular}%
        }
        \;\;\;\;\;\;\;
        }
      }{
        \big(\mathcal{A}^{{}^{\mathrm{t}}}\big)_\bullet
      }
  }
$
  \hspace{-7mm}
}

\vspace{.2cm}

\noindent { \footnotesize
{\bf Figure 1 -- Emergence of intersecting brane observables}
from a differentially refined version (\cref{ChargesInCohomotopy})
of \hyperlink{HypothesisH}{\it Hypothesis H}.
}

\medskip
\medskip
\noindent {\bf Top-down M-theory.} We highlight that,
assuming \hyperlink{HypothesisH}{Hypothesis H},
the analysis shown in
\hyperlink{Figure1}{Figure 1} is completely \emph{top-down}:
knowledge about gauge field theory and
perturbative string theory is not used in deriving
the algebras of observables of M-theory, but only to
interpret them. See also Observation \ref{Dualities} on dualities.

\medskip
While we suggest that the rich system of
expected effects emerging in \hyperlink{Figure1}{Figure 1},
further supports the proposal that \hyperlink{HypothesisH}{Hypothesis H}
is a correct proposal about the mathematical nature of M-theory,
there must of course be more to M-theory than seen in \hyperlink{Figure1}{Figure 1}.
But it is also clear that the differential refinement of Cohomotopy
cohomology theory discussed here (in \cref{ChargesInCohomotopy} below)
is to be further refined, notably by enhancing it with
super-differential flux form structure as in \cite{FSS15}\cite{FSS16a},
with ADE-equivariant structure as in \cite{ADE},
and with fiberwise stabilization as in \cite{GaugeEnhancement}.
This is to be discussed elsewhere.

\newpage

\noindent {\bf Gauge/Gravity duality.}
Collecting the observables and states emerging in
\hyperlink{Figure1}{Figure 1},
we observe that the {\it mathematical duality}
(as an instance of the concept described in \cite{PorstTholen91}\cite{Corfield17})
between higher (co-)observables \eqref{HigherObservables}
on $\mathrm{D}6\perp\mathrm{D}8$-branes
by \hyperlink{HypothesisH}{\it Hypothesis H}
\begin{equation}
  \label{DualityIsDuality}
  \raisebox{28pt}{
  \xymatrix@R=2pt{
    \mathllap{
      \mbox{
        \tiny
        \color{blue}
        Higher observables
      }
    }
    \;\;
    \overset{
      \mathllap{
      \mbox{
        \tiny
        \color{blue}
\begin{tabular}{c}
          Cohomology
          \\
          $\phantom{a}$
        \end{tabular}
      }
      }
    }{
    H^\bullet
    }
    \big(
      \underbrace{
      \underset{N_{\mathrm{f}} \in \mathbb{N}}{\sqcup}
      \Omega
      \underset{
        {}^{\{1,\cdots, N_{\mathrm{f}}\}}
      }{\mathrm{Conf}}
      (\mathbb{R}^3)
      }
    \big)
    & \simeq &
    \overset{
      \mathclap{
      \mbox{
        \tiny
        \color{blue}
        \begin{tabular}{c}
        {\color{greenii}  weight systems}
          \\
          $\phantom{a}$
        \end{tabular}
      }
      }
    }{
      \mathcal{W}
    }
    \ar@{<-->}[ddrr]|-{ }
    & = &
    \big( \mathcal{A} \big)^\ast
    \mathrlap{
      \;\,
      \mbox{
        \tiny
        \color{blue}
        \begin{tabular}{c}
          Higher co-states
        \end{tabular}
      }
    }
    \ar@{<-->}[ddll]|-{
      \;\;\;\;\;\;\;\;
      \mathrlap{
        \raisebox{4pt}{
          \scalebox{.6}{
            {\color{blue} dualization}
            \hspace{2.3cm}
            \cref{QuantumStatesOnChordDiagrams}
          }
        }
      }
    }
    \\
    \\
    \mathllap{
      \overset{
        \raisebox{23pt}{
          \scalebox{.6}{
Prop. \ref{HigherObservablesEquivalentToCohomologyOfLoopedConfigurationSpace}          }
        }
      }{
      \mbox{
        \tiny
        \color{blue}
        Higher {\bf co}-observables
      }
      }
    }
    \;\;
    \underset{
      \mathllap{
      \mbox{
        \tiny
        \color{blue}
        \begin{tabular}{c}
          $\phantom{a}$
          \\
          Homology
        \end{tabular}
      }
      }
    }{
      H_\bullet
    }
    \big(
      \overset{
        \raisebox{10pt}{
          \tiny
          \color{blue}
          \begin{tabular}{c}
            phase space
          \end{tabular}
        }
      }{
      \overbrace{
      \underset{N_{\mathrm{f}} \in \mathbb{N}}{\sqcup}
      \Omega
      \underset{
        {}^{\{1,\cdots, N_{\mathrm{f}}\}}
      }{\mathrm{Conf}}
      (\mathbb{R}^3)
      }
      }
    \big)
    &
    \overset{
      \raisebox{21pt}{
        \scalebox{.6}{
          Prop. \ref{HigherObservablesOnIntersectingBranesAreWeightSystems}
        }
      }
    }{
      \simeq
    }
    &
    \underset{
      \mathclap{
      \mbox{
        \tiny
        \color{blue}
        \begin{tabular}{c}
          $\phantom{a}$
          \\
        {\color{orangeii}  chord diagrams}
        \end{tabular}
      }
      }
    }{
      \mathcal{A}
    }
    & = &
    \big( \mathcal{W} \big)^\ast
    \mathrlap{
      \;\,
      \mbox{
        \hspace{.1cm}
        \tiny
        \color{blue}
        \begin{tabular}{c}
          Higher states
        \end{tabular}
      }
    }
  }
  }
\end{equation}
reflects the {\it gauge/gravity duality} (e.g. \cite{DHMB15}) between
observables/states of
gauge theories and gravity theories on branes
found in \cref{WeightSystemsAsObservablesOnIntersectingBranes} --
see \hyperlink{Figure2}{Figure 2}:

\medskip
\medskip
\medskip
{
\hypertarget{Figure2}{}
\hspace{-1.8cm}
%
  }
\end{tabular}
}

\vspace{1mm}

\noindent { \footnotesize
 {\bf Figure 2 --
   Emergence of gauge/gravity duality}
   from a differentially refined version (\cref{ChargesInCohomotopy})
   of \hyperlink{HypothesisH}{\it Hypothesis H}.
}
\newpage

\medskip
\noindent {\bf Configuration spaces of intersecting branes seen in Cohomotopy.}
The brane intersections arising this way from
\hyperlink{HypothesisH}{\it Hypothesis H} \& Prop. \ref{FinalResult}
are transversal $p \perp (p+2)$-brane
intersections, specifically for $p = 6$ (Remark \ref{GeometricEngineeringOfMonopoles} below),
where $N_{\mathrm{f}}$ $(p+2)$-branes are arranged along
an axis and $N_{\mathrm{c}}={\sum}_{i =1}^{N_{\mathrm{f}}} N_{\mathrm{c},i}$
semi-infinite $p$-branes transversally intersecting them,
with $N_{\mathrm{c},i}$ of them coincident and ending on the $i$th $(p+2)$-brane.
The $p$-branes move along the
$\mathbb{R}^3$ inside the
$(p+2)$-branes which is normal to their intersection locus.
\begin{equation}
\label{ConfigurationSpaceOfBraneIntersections}
\left\{
\raisebox{-77pt}{

}
\right\}
\;\;\;
  \underset{
    \mbox{
      \tiny
      \begin{tabular}{c}
        homotopy
        \\
        equivalence
      \end{tabular}
    }
  }{
    \simeq
  }
  \;\;\;
  \underset{
    \mbox{
      \tiny
      \color{blue}
      \begin{tabular}{c}
        $\phantom{a}$
        \\
        Ordered
        configuration space
        \\
        of $N_{\mathrm{f}}$ points
        in $\mathbb{R}^3$
      \end{tabular}
    }
  }{
    \underset{
      {}^{\{1,\cdots, N_{\mathrm{f}}\}}
    }{
      \mathrm{Conf}
    }
    (\mathbb{R}^3)
  }
\end{equation}\

%Here we observe the following interesting phenomena:
%
%\begin{enumerate}[{\bf (a)}]
%\vspace{-3mm}
%\item {\bf Stacks of coincident branes emerge. }
%The numbers $N_{c,i}$ of coincident branes in the
%$i$th stack
%are not encoded in the configuration space itself.
%Since they are coincident with each other and non-coincident with
%those in other stacks they move as one point in the configuration
%of non-coincident points.
%Instead the numbers turn out to be encoded in the
%\emph{observables} on these configurations; see %\cref{OnStacksOfCoincidentStrands}.
%A similar removal of (further) coincidences from the
%(Coulomb branch) moduli space
%has also been considered in \cite{BullimoreDimofteGaiotto15}\cite{AC17}.
%
%\vspace{-2mm}
%\item {\bf The $\mathrm{D}6\perp\mathrm{D}8$-intersections
%are singled out.}
%The fact that the degree of the Cohomotopy theory $\pi^4$ involved
%in \hyperlink{HypothesisH}{Hypothesis H} is 4
%(this being the co-dimension of $\mathrm{M2}$-$\mathrm{M5}\perp %\mathrm{MK6}$ branes \cite{FSS19b}\cite{FSS19c}\cite{SS19a})
%singles out the case $p =6$ for these $p \perp (p+2)$-brane intersections
%(Remark \ref{SingleOutP6} below).
%This case stands out in that exactly the D8-branes do not
%appear from naive double dimensional reduction of M-theory
%(see also \cite{GaugeEnhancement})
%and that it is precisely the $\mathrm{D}p\perp \mathrm{D}(p+2)$-brane
%intersections for $p = 6$ that correspond to actual monopoles
%in quantum chromodynamics (Remark \ref{GeometricEngineeringOfMonopoles}
%below.)
%\end{enumerate}

\vspace{-2mm}
\noindent {\bf $\mathrm{D}p\perp\mathrm{D}(p+2)$-brane intersections as seen from nonabelian DBI theory.}
Transversal $\mathrm{D}p\perp\mathrm{D}(p+2)$-brane
intersections have been discussed in the literature using the
nonabelian DBI field theory \cite{Tseytlin97}\cite{Myers99}
which is expected on the worldvolume of coincident D-branes;
see \cite{Karch98}\cite{GiveonKutasov99}\cite{Fazzi17} for review.
The following table summarizes the main results of the traditional analysis
(on the left) and indication of emergence from Cohomotopy (on the right):
%The main results of this traditional analysis are shown on the left of the following table:

\medskip
\medskip

\hspace{-.6cm}
\begin{tabular}{|c|c|c||c|}
  \hline
 \!\!\! {\bf
    \begin{tabular}{c}
      Expected physics of
      \\
      $\mathrm{D}p\perp \mathrm{D}(p+2)$-brane
      \\
      intersections
    \end{tabular}
  }
  \!\!\!
  &
  {\bf
    Statement
  }
  &
  \!\!\!
  {
  \bf
  \begin{tabular}{c}
    Derivation from
    \\
    non-abelian DBI
  \end{tabular}
  }
  \!\!\!
  &
 \!\!\!\!\!\!\! {
  \bf
  \begin{tabular}{c}
    Emergence from
    \\
    \hyperlink{Hypothewsis H}{Hypothesis H}
  \end{tabular}
  }
  \!\!\!\!\!\!
  \\
  \hline
  \hline
  {\it Fuzzy funnel geometry}
  &
  \small
  \begin{minipage}[l]{6cm}
  $\mathclap{\phantom{\vert^{\vert}}}$
   A single
   $N_{\mathrm{c},i}$ $\mathrm{D}p \!\perp \! \mathrm{D}(p+2)$-intersection
   is described by non-commutative fuzzy funnel geometry
   (\hyperlink{su2RepToFuzzyFunnel}{Figure 3}),
   identified, via Nahm's equations,
   with Yang-Mills monopoles in the
   $\mathrm{D}(p+2)$-brane worldvolume.
   $\mathclap{\phantom{\vert_{\vert_{\vert}}}}$
  \end{minipage}
  &
  \begin{minipage}[l]{3cm}
    $\mathclap{\phantom{\vert^{\vert}}}$
    \newline \cite{Diaconescu97}\cite{CMT99}
    \cite{HananyZaffaroni99}\cite{Myers01}
    \cite{ConstableLambert02}\cite{BarrettBowcock04}
    \cite{RST04}\cite{McNamara06}
    \cite{MPRS06}
    $\mathclap{\phantom{\vert_{\vert_\vert}}}$
  \end{minipage}
  &
  \cref{su2WeightSystemsAreFuzzyFunnelObservables}
  \\
  \hline
  {\it Hanany-Witten rules}
  &
  \small
  \begin{minipage}[l]{6cm}
  $\mathclap{\phantom{\vert^{\vert^{\vert}}}}$
  The collection of all $N_{\mathrm{f}}$
  $\mathrm{D}p \!\perp\! \mathrm{D}(p+2)$-brane
  intersections is subject to combinatorial rules,
  such as the \emph{s-rule} and the \emph{ordering constraint}.
  $\mathclap{\phantom{\vert_{\vert_{\vert}}}}$
  \end{minipage}
  &
  \begin{minipage}[l]{3cm}
    \cite{HananyWitten97}\cite{BGS97}
    \cite{BachasGreen98}\cite{HoriOoguriOz98}
    \cite{GiveonKutasov99}\cite{GaiottoWitten08}
  \end{minipage}
  &
  \cref{HananyWittenTheory}
  \\
  \hline
\end{tabular}

\vspace{4mm}
\noindent {\bf The open problem of the Non-abelian DBI action.}
While these traditional discussions undoubtedly yield a compelling picture,
it is worth recalling that (in contrast to the abelian case of non-coincident branes)
there is to date \emph{no derivation} from perturbative string theory
of the nonabelian DBI-action for coincident D-branes, as highlighted
in \cite[p. 1]{TaylorRaamsdonk00}\cite[p. 2]{Schwarz01}\cite[p. 5]{Chemissany04}.
The commonly used
{symmetrized trace prescription} of \cite{Tseytlin97}\cite{Myers99},
is somewhat ad-hoc;
and it is known not to be correct
at higher orders \cite{HashimotoTaylor97}\cite{BBdRS01}.
Some correction terms have been proposed in \cite{TaylorRaamsdonk00},
and different proposal for going about the non-abelian DBI-action
has recently been made in \cite{BFS19}.
In contrast, here we find
key expected consequences
of non-abelian DBI-Lagrangians
for intersecting brane physics
emerge from \hyperlink{HypothesisH}{Hypothesis H}
in a \emph{non-Lagrangian} way altogether.

\medskip

\medskip

\noindent {\bf Outline:} The paper is outlined as follows:

\noindent In
\cref{ChargesInCohomotopy}
we introduce the \emph{differential} \hyperlink{HypothesisH}{\it Hypothesis H}
and show that it implies weight systems as higher observables.

\noindent In \cref{WeightSystemsOnChordDiagrams} we recall weight systems on
chord diagrams, streamlined towards our applications.

\noindent In
\cref{WeightSystemsAsObservablesOnIntersectingBranes}
we observe that weight system observables
reflect a variety of effects in intersecting brane physics.

\newpage

%%%%%%%%%%%%%%%%%%%%%%%%%%%%%%%%%%%%%%%%%%%%%%%%%%%%%%%%%%%%%%%%%%%%%
\section{Intersecting brane charges in differential Cohomotopy}
\label{ChargesInCohomotopy}
%%%%%%%%%%%%%%%%%%%%%%%%%%%%%%%%%%%%%%%%%%%%%%%%%%%%%%%%%%%%%%%%%%%%%%

\vspace{-1mm}
\noindent {\bf The open problem of formulating a genuine theory of brane physics.}
As indicated in the Introduction, despite all the discussion of (well-supported but conjectural) aspects of
intersecting brane physics,
an actual formulation of a non-perturbative quantum theory of branes,
namely of \emph{M-theory} \cite{Duff99B}, has remained an open problem
\cite[6]{Duff96}\cite[p. 2]{HoweLambertWest97}\cite[p. 6]{Duff98}\cite[p. 2]{NicolaiHelling98}\cite[p. 330
]{Duff99B}\cite[12]{Moore14}\cite[p. 2]{ChesterPerlmutter18}\cite
{Witten19}\cite{Duff19}.
The need for an identification of the non-perturbative
theory has recently become manifest
with the community no longer able to agree on the
validity of brane constructions that have been discussed for many years
\cite{DanielssonVanRiet18}\cite[p. 14-22]{Banks19}.
Even the very ingredients of such a theory have remained open.

\medskip

\noindent {\bf Charge quantization in generalized cohomology theory.}
On the other hand, the low energy limit of M-theory is supposed
to be $D = 11$ supergravity \cite{CJS78}, whose only ingredient,
besides the field of gravity, is the \emph{C-field},
the higher analog of the B-field in string theory, which in turn
is the higher analog of the ``A-field'' in particle physics,
namely of the Maxwell field, i.e., of the abelian Yang-Mills field.
But a famous insight going back to Dirac (see \cite{Heras18})
says that in its non-perturbative quantum theory, the Maxwell field
becomes subject to a refinement known as
\emph{Dirac charge quantization} (see \cite{Freed00} for a general treatment).
In modern formulation this means that the flux density of the field (the Faraday tensor),
which a priori seems to be just a differential 2-form, is promoted to
a cocycle in differential ordinary 2-cohomology theory.
Later, a directly analogous topological constraint has been argued
to apply to the B-field in string theory, where up to some fine print,
what naively looks like the flux density 3-form of the B-field
is argued to really be regarded as being charge-quantized in
differential ordinary 3-cohomology theory (see \cite{Bry}).
One might suspect an evident pattern here, which would seem to
continue with the suggestion that the M-theory C-field
needs to be regarded as charge-quantized
in differential ordinary 3-cohomology theory, up to some fine print
(\cite{DFM03}\cite{HopkinsSinger05}\cite{tcu}\cite{FSS14a}).
On the other hand, and in contrast to the C-field in M-theory,
the B-field in string theory does not exist in isolation;
instead, it couples to the RR-field. The combination of the B-field
and the RR-field has famously and widely been argued to be
charge quantized in a differential \emph{generalized} cohomology theory,
namely in some version of twisted K-theory
(see \cite{GS} and also \cite[2]{GaugeEnhancement} for pointers and discussion in our context).

\medskip

\noindent {\bf Generalized cohomology theory for C-field charge quantization in M-theory.}
All this rich structure in string theory is -- somehow --
supposed to lift to just the metric field and the C-field in M-theory.
This suggests that the M-theory C-field itself must
be regarded as being charge-quantized in some rich generalized
cohomology theory \cite{Sa1}\cite{Sa2}\cite{Sa3}\cite{tcu}
such as Cohomotopy cohomology theory \cite[2.5]{Sati13}.
Based on a systematic analysis in super rational homotopy theory
of the $\kappa$-symmetry super $p$-brane WZW terms
\cite{FSS13}\cite{FSS15}\cite{FSS16a}\cite{FSS16b}
\cite{GaugeEnhancement}
(see \cite{FSS19a} for review),
a concrete hypothesis for
this generalized cohomological charge quantization of the C-field
was formulated in \cite{FSS19b}:

\vspace{0mm}
\begin{center}
\hypertarget{HypothesisH}{}
\fbox{{\bf Hypothesis H.} {\it The M-theory C-field is charge-quantized in J-twisted Cohomotopy theory.}
}
\end{center}

\vspace{0mm}
\noindent In a series of articles
\cite{FSS19b}\cite{FSS19c}\cite{SS19a}\cite{SS19b}
various implications of
this {\it Hypothesis H}
have been checked
to agree with various expected aspects of M-theory
in the topological sector, i.e.,
in the approximation where only the homotopy type of spacetime
is taken into account.

\medskip

\noindent {\bf Differential Cohomotopy and intersecting branes.}
Here we consider a partial refinement of Cohomotopy cohomology theory to
a \emph{differential} cohomology theory, which is sensitive at least to
the homeomorphism type of spacetime (Prop. \ref{DifferentialCohomotopyOnPenroseDiagrams} below).
Then we prove (Prop. \ref{FinalResult} below)  that this charge quantization of the C-field
in differential Cohomotopy theory implies that the cocycle space of intersecting  D6-D8-brane
charges is the ordered configuration space of points
as in \eqref{ConfigurationSpaceOfBraneIntersections}.
This means that:

\begin{enumerate}[{\bf (1)}]
\vspace{-2mm}
\item The higher observables \eqref{HigherObservablesOnD6D8Intersections} in \cref{WeightSystemsOnChordDiagrams}
and hence, by \eqref{ObservablesAreWeightSystems},
the weight systems on chord diagrams in \cref{WeightSystemsOnChordDiagrams}
are the quantum observables on intersecting brane moduli that
are implied by \hyperlink{HypothesisH}{Hypothesis H}.

\vspace{-2.5mm}
\item Therefore, also the aspects of intersecting brane physics
that are reflected in weight systems on chord diagrams according
to the discussion in \cref{WeightSystemsAsObservablesOnIntersectingBranes}
are implications of \hyperlink{HypothesisH}{Hypothesis H}.
\end{enumerate}

\newpage

%%%%%%%%%%%%%%%%%%%%%%%%%%%%%%%%%%%%%%%%%%%%%%%
\subsection{Charges vanishing at infinity}
\label{VanishingAtInfinity}
%%%%%%%%%%%%%%%%%%%%%%%%%%%%%%%%%%%%%%%%%%%%%%%

\noindent {\bf Points at infinity.}
For the following definitions applied to physics,
we are to think of all boundaries and  base points
as representing ``points at infinity''.
We write $\mathbb{D}^n$ for the \emph{closed} $n$-disk
with boundary $\partial \mathbb{D}^n \simeq S^{n-1}$
and interior $\mathrm{Int}(\mathbb{D}^n) \simeq \mathbb{R}^n$.
We write $(-)^{\mathrm{cpt}}$ for the one-point compactification
of a topological space, so that
\begin{equation}
  \label{OnePointCompactificationOfEuclideanSpace}
  (\mathbb{R}^n)^{\mathrm{cpt}}
  \;\simeq\;
  \mathbb{D}^n / \partial \mathbb{D}^n
  \;\simeq\;
  S^n
\end{equation}
and we write
$$
  \infty \;\in\; (\mathbb{R}^n)^{\mathrm{cpt}}
$$
for the extra point. This is literally the \emph{point at infinity},
and under the above equivalences, all points on the boundary
of $\mathbb{D}^n$ get identified with it:
$$
  \xymatrix@R=1em@C=3em{
    &
    \overset{
      \mathclap{
      \mbox{
        \tiny
        \color{blue}
\begin{tabular}{c}
          Euclidean
          \\
          $n$-space
          \\
          $\phantom{a}$
        \end{tabular}
      }
      }
    }{
      \mathbb{R}^n
    }
  \;  \ar@{^{(}->}[d]
   \;  \ar@{^{(}->}[rr]^{
      \mathclap{
      \mbox{
        \tiny
        \color{blue}
        \begin{tabular}{c}
          interior
          \\
          $\phantom{a}$
        \end{tabular}
      }
      }
    }
    &&
    \overset{
      \mathclap{
      \mbox{
        \tiny
        \color{blue}
        \begin{tabular}{c}
          Closed
          \\
          $n$-disk
          \\
          $\phantom{a}$
        \end{tabular}
      }
      }
    }{
      \mathbb{D}^n
    }
    \ar@{<-^{)}}[rr]^{
      \mathclap{
      \mbox{
        \tiny
        \color{blue}
        \begin{tabular}{c}
          boundary
          \\
          $\phantom{a}$
        \end{tabular}
      }
      }
    }
    \ar@{->>}[d]
    &&
    \overset{
      \mathclap{
      \mbox{
        \tiny
        \color{blue}
        \begin{tabular}{c}
          $(n-1)$-sphere
          \\
          $\phantom{a}$
        \end{tabular}
      }
      }
    }{
      S^{n-1}
    }
    \ar@{->>}[d]
    \\
    \underset{
      \mathclap{
      \mbox{
        \tiny
        \color{blue}
        \begin{tabular}{c}
          $\phantom{a}$
          \\
          $n$-sphere
        \end{tabular}
      }
      }
    }{
      S^n
    }
    \ar@{}[r]|-{\simeq }
    &
    \underset{
      \mathclap{
      \mbox{
        \tiny
        \color{blue}
        \begin{tabular}{c}
          $\phantom{a}$
          \\
          One-point
          \\
          compactification
        \end{tabular}
      }
      }
    }{
      \big(
        \mathbb{R}^n
      \big)^{\mathrm{cpt}}
    }
    \ar[rr]_-{ \simeq }
    &&
    \underset{
      \mathclap{
      \mbox{
        \tiny
        \color{blue}
        \begin{tabular}{c}
          $\phantom{a}$
          \\
          Boundary collapsed
          \\
          to one point
        \end{tabular}
      }
      }
    }{
      \mathbb{D}^n/\partial \mathbb{D}^n
    }
    \ar@{<-^{)}}[rr]
    &&
    \underset{
      \mathclap{
      \mbox{
        \tiny
        \color{blue}
        \begin{tabular}{c}
          $\phantom{a}$
          \\
          Point at
          \\
          infinity
        \end{tabular}
      }
      }
    }{
      \{\infty\}
    }
  }
$$
We will thus regard one-point compactifications $(-)^{\mathrm{cpt}}$
as pointed topological spaces with the base point denoted ``$\infty$''.

\medskip
If for classifying spaces we instead denote the base point by ``$0$'',
then pointed maps express exactly the idea of
\emph{cocycles vanishing at infinity}:
$$
  \xymatrix@R=1em@C=3em{
    \mbox{
      \tiny
      \color{blue}
      Pointed spaces
    }
    &
    \overset{
      \mathclap{
      \mbox{
        \tiny
        \color{blue}
        \begin{tabular}{c}
          Space-time
          \\
          $\phantom{a}$
        \end{tabular}
      }
      }
    }{
      X
    }
    \ar[rrr]^-{
      \mbox{
        \tiny
        \color{blue}
        \begin{tabular}{c}
          cocycle /
          \\
          field configuration
        \end{tabular}
      }
    }
    &&&
    \overset{
      \mathclap{
      \mbox{
        \tiny
        \color{blue}
        \begin{tabular}{c}
          Classifying
          space
          \\
          $\phantom{a}$
        \end{tabular}
      }
      }
    }{
      A
    }
    \\
    \mbox{
      \tiny
      \color{blue}
      Base points
    }
    &
    \underset{
      \mathclap{
      \mbox{
        \tiny
        \color{blue}
        \begin{tabular}{c}
          $\phantom{a}$
          \\
          point at infinity
        \end{tabular}
      }
      }
    }{
      \{\infty\}
    }
    \ar[rrr]_-{
      \mbox{
        \tiny
        \color{blue}
        vanishing at infinity
      }
    }
    \ar@{^{(}->}[u]
    &&&
    \underset{
      \mathclap{
      \mbox{
        \tiny
        \color{blue}
        \begin{tabular}{c}
          $\phantom{a}$
          \\
          vanishing
          field value
        \end{tabular}
      }
      }
    }{
      \{0\}
    }
    \ar@{^{(}->}[u]
  }
$$

If we wish to consider $\mathbb{R}^d$ explicitly
without the requirement that cocycles on it vanish at infinity, we
instead add the ``point at infinity'' as a disjoint point
$$
  (\mathbb{R}^d)_+
  \;:=\;
  \mathbb{R}^d
  \sqcup \{\infty\}
  \,.
$$
In summary:
$$
  \xymatrix@R=1em{
    \mbox{
      \tiny
      \color{blue}
      \begin{tabular}{c}
        Un-pointed space
      \end{tabular}
    }
    &
    &
    \overset{
      \mathclap{
      \mbox{
        \tiny
        \color{blue}
        \begin{tabular}{c}
          Euclidean space...
        \end{tabular}
      }
      }
    }{
      \mathbb{R}^d
    }
    \!\!\!\!\!\!\!\!\!\!\!\!\!\! \ar@{_{(}->}[dl]
    \;\;\;\;\; \ar@{^{(}->}[dr]
    \\
    \mbox{
      \tiny
      \color{blue}
      Pointed spaces
    }
    &
    \underset{
      \mathclap{
      \mbox{
        \tiny
        \color{blue}
        \begin{tabular}{c}
          $\phantom{a}$
          \\
          ... such that cocycles
          \\
          vanish at infinity.
        \end{tabular}
      }
      }
    }{
      (\mathbb{R}^d)^{\mathrm{cpt}}
    }
    &&
    \underset{
      \mathclap{
      \mbox{
        \tiny
        \color{blue}
        \begin{tabular}{c}
          $\phantom{a}$
          \\
          ...such that cocycles
          \\
          are unconstrained.
          \\
        \end{tabular}
      }
      }
    }{
      (\mathbb{R}^d)_+
    }
  }
$$

\medskip

Forming the smash product of these pointed spaces then
yields Euclidean spaces on which cocycles have to vanish
at infinity in some directions, but not necessarily in others:

\newpage

$\,$

\vspace{-1cm}

\begin{equation}
\label{PenroseDiagramSpaces}
\scalebox{.9}{
%
  }
  \\
  \hline
\end{tabular}
}
\end{equation}

\vspace{-.2cm}

%%%%%%%%%%%%%%%%%%%%%%%%%%%%%%%%%%%%%%%%%%%%%%
\subsection{Configuration spaces of points}
\label{ConfigurationSpacesOfPoints}
%%%%%%%%%%%%%%%%%%%%%%%%%%%%%%%%%%%%%%%%%%%%%%

We now first recall, in Def. \ref{ConfigurationSpaces}, the relevant definitions of configuration spaces of points
(see e.g. \cite[1]{Boedigheimer87}).
Then we observe,
in Prop. \ref{OrderedUnlabeledAsFiberProduct},
a certain
relation between un-ordered and ordered configuration spaces of points.
This is the key to relating differential Cohomotopy to intersecting branes
in \cref{IntersectingD6D8BraneChargesInDifferentialCohomotopy}.

\medskip
\begin{defn}[Configuration spaces of points]
\label{ConfigurationSpaces}
Let $\Sigma^D$ be a smooth manifold
with (a possibly empty) boundary $\partial \Sigma^D \hookrightarrow \Sigma^D$.
For $k \in \mathbb{N}$, with
$\mathbb{D}^k$ denoting the closed $k$-disk, $\Delta$ the diagonal, and
${\rm Sym}_n$ the symmetric group of order $n$,
we consider the following topological configuration spaces of points
in $\Sigma^D$, possibly with labels in $\mathbb{D}^k$:
\begin{center}

\end{center}

\end{defn}

Here is an illustration of a labelled and un-ordered configuration
of points:
\begin{center}
%

\vspace{.2cm}

\begin{minipage}[l]{17cm}
\noindent {\footnotesize \bf An element of the
unordered $\mathbb{D}^1$-labeled  configuration space
$\mathrm{Conf}\big( \mathbb{R}^3, \mathbb{D}^1 \big)$  }
{\footnotesize
according to Def. \ref{ConfigurationSpaces},
is a set of points in $\mathbb{R}^3 \times \mathbb{R}^1$
with distinct projections to $\mathbb{R}^3 \times \{0\}$. The topology
is such that points moving to infinity along
$\mathbb{R}^1$ (i.e., to the boundary of $\mathbb{D}^1$) disappear.
}
\end{minipage}
\end{center}

\medskip
In order to study all possible configurations, we introduce the following useful notion.

\begin{defn}[Category of Penrose diagrams]
\label{ContravariantlyFunctorialConfigurationSpaces}
For $p \in \mathbb{N}$ we write
\begin{equation}
  \label{CategoryOfPenroseDiagrams}
  \mathrm{PenroseDiag}_p
  \;\;:=\;\;
  \left\{\!\!\!\!\!
  \mbox{
\begin{tabular}{l}
      Penrose-diagram spaces of dimension $p$
      \\
      with continuous functions between them
      \\
      which are injections away from infinity
    \end{tabular}
  }
  \!\!\!\!\! \right\}
\end{equation}
for the category whose objects are the pointed topological spaces
$(\mathbb{R}^p)^{\mathrm{cpt}} \wedge (\mathbb{R}^{p-d})_+$
from \eqref{PenroseDiagramSpaces}, for $0 \leq d \leq p$,
and whose morphisms are the
continuous functions between these that (co-)restrict to embeddings
after removal of basepoints, as shown on the left of the following
diagrams:
\begin{equation}
  \label{ConfFunctorForPositivedPrime}
  \hspace{-2cm}
  \raisebox{32pt}{
  \xymatrix{
    \overset{
      \mathclap{
      \mbox{
        \tiny
        \color{blue}
        \begin{tabular}{c}
          any
          \\
          Penrose diagram
          \\
          $\phantom{a}$
        \end{tabular}
      }
      }
    }{
      (\mathbb{R}^d)^{\mathrm{cpt}}
      \wedge
      (\mathbb{R}^{p-d})_+
    }
    \ar@{}[dr]|-{
      \mbox{
        \tiny
        \color{blue}
        continuous injections
      }
    }
    \ar@{<-^{)}}[r]
    &
    ~ \overset{
      \mathclap{
      \mbox{
        \tiny
        \color{blue}
        \begin{tabular}{c}
          its subspace
          \\
          away from infinity
          \\
          $\phantom{a}$
        \end{tabular}
      }
      }
    }{
      \mathbb{R}^p
    }
    &
    \hspace{-5mm}
    \mapsto
    \hspace{-5mm}
    &
    \mathbb{R}^d \times \mathbb{D}^{p-d}/\mathrm{bdry}
    \ar[d]_{i^\ast}^-{
      \mbox{
        \tiny
        $
        {\begin{array}{cc}
          i(x) & y \mathrlap{\,\notin \mathrm{Im}(i)}
          \\
          \mapsdown & \mapsdown
          \\
          x & \infty
        \end{array}}
        $
      }
    }
    &
    \hspace{-5mm}
    \mapsto
    \hspace{-5mm}
    &
    \overset{
      \mathclap{
      \mbox{
        \tiny
        \color{blue}
        \begin{tabular}{c}
          configuration space of points
          \\
          in the Penrose diagram space
          \\
          which may vanish towards infinity
          \\
          along directions {\bf not} compactified
          \\
          $\phantom{a}$
        \end{tabular}
      }
      }
    }{
      \mathrm{Conf}
      \big(
        \mathbb{R}^d,
        \mathbb{D}^{p-d}
      \big)
    }
    \ar[d]_-{ (i^\ast)_\ast }
    &
   \hspace{-1.2cm} \mathrlap{
    \mbox{for $d \geq 1$}
    }
    \\
    (\mathbb{R}^{d'})^{\mathrm{cpt}}
    \wedge
    (\mathbb{R}^{p-{d'}})_+
   ~ \ar@{<-^{)}}[r]
    \ar@{^{(}->}[u]
    &
    ~
    \mathbb{R}^p
    \ar@{^{(}->}[u]^-i
    &
    \hspace{-5mm}
    \mapsto
    \hspace{-5mm}
    &
    \mathbb{R}^{d'} \times \mathbb{D}^{p-{d'}}/\mathrm{bdry}
    &
    \hspace{-5mm}
    \mapsto
    \hspace{-5mm}
    &
    \mathrm{Conf}
    \big(
      \mathbb{R}^{d'},
      \mathbb{D}^{p-{d'}}
    \big)
    & \hspace{-1.2cm}
    \mathrlap{
    \mbox{for $d' \geq 1$}
    }
  }
  }
\end{equation}
In the special case that the domain of the map is the Penrose
diagram with no compactified dimensions, we set:
\begin{equation}
  \label{ConfFunctorForVanishingdPrime}
  \hspace{-2cm}
  \raisebox{32pt}{
  \xymatrix{
    (\mathbb{R}^d)^{\mathrm{cpt}}
    \wedge
    (\mathbb{R}^{p-d})_+
    \ar@{<-^{)}}[r]
    &
    ~
    \mathbb{R}^p
    &
    \hspace{-5mm}
    \mapsto
    \hspace{-5mm}
    &
    \mathbb{R}^d \times \mathbb{D}^{p-d}/\mathrm{bdry}
    \ar[d]_{i^\ast}^-{
      \mbox{
        \tiny
        $
        {\begin{array}{cc}
          i(x) & y \mathrlap{\,\notin \mathrm{Im}(i)}
          \\
          \mapsdown & \mapsdown
          \\
          x & \infty
        \end{array}}
        $
      }
    }
    &
    \hspace{-5mm}
    \mapsto
    \hspace{-5mm}
    &
    \mathrm{Conf}
    \big(
      \mathbb{R}^d,
      \mathbb{D}^{p-d}
    \big)
    \ar[d]_-{ (i^\ast)_\ast }
    & \hspace{-1.2cm}
    \mathrlap{
      \mbox{for $d \geq 1$}
    }
    \\
    \underset{
      \mathclap{
      \mbox{
        \tiny
        \color{blue}
        \begin{tabular}{c}
          $\phantom{a}$
          \\
          in the special case that
          \\
          no infinite directions are compactified...
        \end{tabular}
      }
      }
    }{
      \underset{
        = (\mathbb{R}^p)_+
      }{
      \underbrace{
        (\mathbb{R}^0)^{\mathrm{cpt}}
        \wedge
        (\mathbb{R}^{p})_+
      }
      }
    }
    \ar@{<-^{)}}[r]
    \ar@{^{(}->}[u]
    &
    ~
    \mathbb{R}^p
    \ar@{^{(}->}[u]^-i
    &
    \hspace{-5mm}
    \mapsto
    \hspace{-5mm}
    &
    \mathbb{D}^p/\mathrm{bdry}
    &
    \hspace{-5mm}
    \mapsto
    \hspace{-5mm}
    &
    \underset{
      \mathclap{
      \mbox{
        \tiny
        \color{blue}
        \begin{tabular}{c}
          $\phantom{a}$
          \\
          ...we assign this configuration space of any number of
          \\
          points that may vanish towards infinity in any direction
          \\
          (instead of $\mathrm{Conf}(\mathbb{R}^0,\mathbb{D}^p)$,
          whose configurations have at most one point)
        \end{tabular}
      }
      }
    }{
      \mathrm{Conf}
      \big(
        \mathbb{D}^{p}
      \big)
    }
    & \hspace{-1.2cm}
    \mathrlap{
      \mbox{(for $d' = 0$)}
    }
  }
  }
\end{equation}
To this category \eqref{CategoryOfPenroseDiagrams}
we extend the construction of configuration spaces
from Def. \ref{ConfigurationSpaces}
as a contravariant functor with values in pointed topological spaces,
\begin{equation}
  \label{ConfigurationSpaceFunctor}
  \mathrm{Conf}
  \;\;:\;\;
  \mathrm{PenroseDiag}_p^{\mathrm{op}}
  \xymatrix{\ar[rr]&&}
  \underset{
    \mathclap{
    \mbox{
      \tiny
      \color{blue}
      \begin{tabular}{c}
        $\phantom{a}$
        \\
        1-category of actual
        \\
        pointed topological spaces
        \\
        (not up to homotopy)
      \end{tabular}
    }
    }
  }{
    \mathrm{Top}^{\ast/}
  }
  \end{equation}
by defining its action on morphisms as shown on the right of
the above diagrams \eqref{ConfFunctorForPositivedPrime} and
\eqref{ConfFunctorForVanishingdPrime}.
\end{defn}

\begin{example}[Maps of configuration spaces for ordered fiber product]
\label{MapsOfConfigurationSpacesForOrderedFiberProduct}
We are going to be interested in the following
pairs of maps of Penrose diagram spaces \eqref{PenroseDiagramSpaces}
and their induced maps of configuration spaces, according to
Def. \ref{ContravariantlyFunctorialConfigurationSpaces}:
$$
\hspace{-.5cm}
\scalebox{.88}{
$
\begin{array}{ccccc}
  \left(
  \!\!\!\!\!
  \raisebox{20pt}{
  \xymatrix@R=9pt@C=-20pt{
    (\mathbb{R}^D)^{\mathrm{cpt}}
    \wedge
    (\mathbb{R}^1)_+
    &&
    (\mathbb{R}^1)^{\mathrm{cpt}}
    \wedge
    (\mathbb{R}^D)_+
    \\
    \\
    &
    (\mathbb{R}^{D+1})_+
    \ar@{^{(}->}[uul]^{i_L}
    \ar@{_{(}->}[uur]_{i_R}
  }
  }
  \!\!\!\!
  \right)
  &
  \hspace{-2mm}
  \mapsto
    \hspace{-2mm}
  &
  \left(
  \!\!\!\!\!
  \raisebox{20pt}{
  \xymatrix@C=-20pt@R=9pt{
    \mathbb{R}^D \times \mathbb{D}^1/\mathrm{bdry}
    &&
    \mathbb{R}^1 \times \mathbb{D}^D/\mathrm{bdry}
    \\
    \\
    &
    \mathbb{D}^D \times \mathbb{D}^1 /\mathrm{bdry}
    \ar@{<-^{)}}[uul]^{ (i_L)^\ast }
    \ar@{<-_{)}}[uur]_{(i_R)^\ast}
  }
  }
  \!\!\!\!
  \right)
  &
    \hspace{-2mm}
  \mapsto
    \hspace{-2mm}
  &
  \left(
  \!\!\!\!\!\!
  \raisebox{20pt}{
  \xymatrix@C=-25pt@R=9pt{
    \mathrm{Conf}
    \big(
      \mathbb{R}^D, \mathbb{D}^1
    \big)
    &&
    \mathrm{Conf}
    \big(
      \mathbb{R}^1, \mathbb{D}^D
    \big)
    \\
    \\
    &
    \mathrm{Conf}
    \big(
      \mathbb{D}^{D+1}
    \big)
    \ar@{<-^{)}}[uul]^{ ((i_L)^\ast)_\ast }
    \ar@{<-_{)}}[uur]_{ ((i_R)^\ast)_\ast  }
  }
  }
  \!\!\!\!
  \right)
  \\
\end{array}
$
}
$$

\end{example}

\medskip

\begin{prop}[Ordered unlabeled configurations as fiber product of unordered labeled configurations]
  \label{OrderedUnlabeledAsFiberProduct}
  For $D \in \mathbb{N}$, there is a homotopy equivalence
  between
the disjoint union of ordered unlabelled configuration spaces
  in $\mathbb{R}^D$
  and the fiber product of unordered but labelled configuration spaces
  (Def. \ref{ConfigurationSpaces})
  as follows:
  \begin{equation}
    \label{TheFiberProductAndItsHomotopyEquivalence}
    \underset{
      \mathclap{
      \mbox{
        \tiny
        \color{blue}
        \begin{tabular}{c}
          $\phantom{a}$
          \\
          Ordered configurations
          \\
          of points in $\mathbb{R}^D$
        \end{tabular}
      }
      }
    }{
      \underset{
        n \in \mathbb{N}
      }{\sqcup}
      \underset{
        {}^{\{1,\cdots, n\}}
      }{
        \mathrm{Conf}
      }
      \big(
        \mathbb{R}^D
      \big)
    }
    \;\;\;\;
    \underset{
      \mathclap{
      \mbox{
        \tiny
        \begin{tabular}{c}
          hmtpy
        \end{tabular}
      }
      }
    }{
      \simeq
    }
    \;\;\;\;
    \overset{
      \mathclap{
      \mbox{
        \tiny
        \color{blue}
        \begin{tabular}{c}
          Un-ordered configurations
          \\
          of points in $\mathbb{R}^D$
          \\
          with labels in $\mathbb{D}^1$
          \\
          $\phantom{a}$
        \end{tabular}
      }
      }
    }{
      \mathrm{Conf}
      \big(
        \mathbb{R}^D,
        \mathbb{D}^1
      \big)
    }
    \underset{
      \;\;\;
      \mathclap{
        \underset{
          \mbox{
            \tiny
            \color{blue}
            \begin{tabular}{c}
              $\phantom{a}$
              \\
              Un-ordered configurations
              \\
              of points in $\mathbb{D}^{D+1}$
            \end{tabular}
          }
        }{
          \mathrm{Conf}
          (
            \mathbb{D}^{D+1}
          )
        }
      }
      \;\;\;
    }
    {\;\;\;\;\;\;
      \times
    \;\;\;\;\;\;}
    \overset{
      \mbox{
        \tiny
        \color{blue}
        \begin{tabular}{c}
          Un-ordered configurations
          \\
          of points in $\mathbb{R}^1$
          \\
          with labels in $\mathbb{R}^D$
          \\
          $\phantom{a}$
        \end{tabular}
      }
    }{
      \mathrm{Conf}\big(\mathbb{R}^1, \mathbb{D}^D \big)
    }
    \,,
  \end{equation}
  where the fiber product on the right is that induced
  from the maps in Example \ref{MapsOfConfigurationSpacesForOrderedFiberProduct}.
\end{prop}
\proof
  We compute as follows (where all topologies are the evident ones) --
  see \hyperlink{FigureO}{Figure O} for illustration of the
  logic behind the argument:
  $$
    \begin{aligned}
        \mathrm{Conf}\big(\mathbb{R}^D, \mathbb{D}^1 \big)
&    \underset{
      \;\;\;
      \mathclap{
        \mathrm{Conf}
        (
          \mathbb{D}^{D+1}
        )
      }
      \;\;\;
    }{\;\;\;\;\;\;
      \times
    \;\;\;\;\;}
    \mathrm{Conf}\big(\mathbb{R}^1, \mathbb{D}^D \big)
    \\
    &
    \;\;\;
    \underset{
      \mathclap{
      \mbox{
        \tiny
        homeo
      }
      }
    }{\simeq}
    \;\;\;
    \underset{n \in \mathbb{N}}{\sqcup}
    \Big\{ \!\!\!
      \left.
      \big\{
        (\vec x_i, y_i ) \in \mathbb{R}^D \times \mathbb{R}^1
      \big\}_{i = 1}^n
      \,
      \right\vert
      \,
      \underset{i \neq j}{\forall}
      \big(
        \vec x_i \neq \vec x_j
        \;\;\mbox{and}\;\;
        y_i \neq y_j
      \big)
    \Big\}\big/\mathrm{Sym}(n)
    \\
    &
    \;\;\;
    \underset{
      \mathclap{
      \mbox{
        \tiny
        homeo
      }
      }
    }{\simeq}
    \;\;\;
    \underset{n \in \mathbb{N}}{\sqcup}
    \Big\{ \!\!\!
      \left.
      \big(
        \big\{\vec x_i \in \mathbb{R}^D \big\}_{i = 1}^n,
        \sigma \in \mathrm{Sym}(n),
        (d_0, d_1, \cdots, d_{n-1}) \in \mathbb{R}^1 \times (\mathbb{R}^1_+)^{n-1}
      \big)
      \,
      \right\vert
      \,
      \underset{i \neq j}{\forall}
      \big(
        \vec x_i \neq \vec x_j
      \big) \!
    \Big\}\big/\mathrm{Sym}(n)
    \\
    &
    \;\;\;
    \underset{
      \mathclap{
      \mbox{
        \tiny
        hmtpy
      }
      }
    }{\simeq}
    \;\;\;
    \underset{n \in \mathbb{N}}{\sqcup}
    \Big\{ \!\!\!
      \left.
      \big(
        \big\{\vec x_i \mathbb{R}^D \big\}_{i = 1}^n,
        \sigma \in \mathrm{Sym}(n)
      \big)
      \,
      \right\vert
      \,
      \underset{i \neq j}{\forall}
      \big(
        \vec x_i \neq \vec x_j
      \big)
    \Big\}\big/\mathrm{Sym}(n)
    \\
    &
    \;\;\;
    \underset{
      \mathclap{
      \mbox{
        \tiny
        homeo
      }
      }
    }{\simeq}
    \;\;\;
    \underset{n \in \mathbb{N}}{\sqcup}
    \Big\{ \!\!\!
      \left.
      \big\{\vec x_i \in \mathbb{R}^D \big\}_{i = 1}^n
      \,
      \right\vert
      \,
      \underset{i \neq j}{\forall}
      \big(
        \vec x_i \neq \vec x_j
      \big)
    \Big\}
    \\
    &
    \;\;\;
    =
    \;\;\;
    \underset{n \in \mathbb{N}}{\sqcup}
    \;
    \underset{
      {}^{\{1,\cdots,n\}}
    }{\mathrm{Conf}}(\mathbb{R}^D)\;.
    \end{aligned}
  $$
  Here the first step just unwinds the definition of the fiber product.
  In the second step we encode an $n$-tuple of pairwise distinct real numbers
  $(y_1, y_2, \cdots, y_n)$ equivalently as a pair consisting
  of the permutation $\sigma$ that puts them into linear order
  and the tuple $(d_0, d_1, \cdots, d_{n-1})$ of their relative positive distances:
  $$
    \xymatrix@R=10pt{
      y_{\sigma_1}
      \ar@{=}[d]
      \ar@{}[r]|-{<}
      &
      y_{\sigma_2}
      \ar@{=}[d]
      \ar@{}[r]|-{<}
      &
      y_{\sigma_3}
      \ar@{=}[d]
      \ar@{}[r]|-{<}
      &
      \cdots
      \ar@{}[r]|-{<}
      &
      y_{\sigma_n}
      \ar@{=}[d]
      \\
      d_0
      \ar@{}[r]|-{<}
      &
      d_0 + d_1
      \ar@{}[r]|-{<}
      &
      d_0 + d_1 + d_2
      \ar@{}[r]|-{<}
      &
      \cdots
      \ar@{}[r]|-{<}
      &
      \underoverset{i = 0}{n-1}{\sum} d_i
    }
  $$
  In the third step we use that the space of these
  relative distances is, clearly, homotopy equivalent to the
  point: $\mathbb{R}^1 \times (\mathbb{R}^1)^{n-1}
  \underset{\mathrm{hmtpy}}{\simeq} \ast$.
  In the fourth step we use that
  $
    (X \times G) /_{\!{}_{\mathrm{diag}}} G
    \underset{\mathrm{homeo}}{\simeq}
    X
  $
  for any $G$-space $X$. The last step
  recognizes the ordered configuration space
  according to Def. \ref{ConfigurationSpaces}.
\hspace{7cm} \endproof

The content of Prop. \ref{OrderedUnlabeledAsFiberProduct} is illustrated
by the following graphics:

\begin{center}
\hypertarget{FigureO}{}


\begin{minipage}[l]{17cm}
\noindent {\footnotesize \bf
Figure O --
The ordered unlabeled configuration space is a fiber product of
unordered labeled configuration spaces
}
{\footnotesize
according to Prop. \ref{OrderedUnlabeledAsFiberProduct}:
A linearly ordered configuration of points in
$\mathbb{R}^3$ is the same as {\bf (a)} an unordered configuration in
$\mathbb{R}^3 \times \mathbb{R}^1$ which projects to {\bf (b)}
 an unordered $\mathbb{D}^1$-labelled configuration in $\mathbb{R}^3$
as well as to {\bf (c)} an unordered $\mathbb{D}^3$-labelled configuration in
$\mathbb{R}^1$.
Condition {\bf (c)} equips the configuration from condition {\bf (b)} with
a linear ordering.
}
\end{minipage}

\end{center}

%%%%%%%%%%%%%%%%%%%%%%%%%%%%%%%%%%%%%%%%%%%%%%%%%%%%%%%%%
\subsection{Differential Cohomotopy cocycle spaces}
\label{CohomotopyChargeMap}
%%%%%%%%%%%%%%%%%%%%%%%%%%%%%%%%%%%%%%%%%%%%%%%%%%%%%%%%%

For the following, we take $X$
to be a locally compact pointed topological space
of the homotopy type of a CW-complex, for example
one of the Penrose diagram spaces \eqref{PenroseDiagramSpaces}
discussed in \cref{VanishingAtInfinity}.

\medskip

\noindent {\bf Plain Cohomotopy cohomology theory.}
For $p \in \mathbb{N}$
a degree, the \emph{cocycle space of $p$-Cohomotopy theory on $X$}
is the pointed mapping space from $X$ to the $p$-sphere:
\begin{equation}
  \label{CohmotopyCocycleSpace}
  \boldpi^p
  (
    X
  )
  \;:=\;
  \mathrm{Maps}^{\ast/\!\!}
  \big(
    X,
    S^p
  \big)
  \,.
\end{equation}
The set of connected component of this space is the
actual \emph{$p$-Cohomotopy set} of $X$:
\begin{equation}
  \label{CohmotopySet}
 \pi^p
  (
    X
  )
  \;:=\;
  \pi_0
  \Big(
  \mathrm{Maps}^{\ast/\!\!}
  \big(
    X,
    S^p
  \big)
  \Big)
  \,.
\end{equation}
This implies that the homotopy type of $\boldpi^n(X)$,
and so in particular the isomorphism class of $\pi^n(X)$,
depend only on the homotopy type of $X$.
The resulting (contravariant) functorial assignment
\vspace{-2mm}
\begin{equation}
  \label{PlainCohomotopyFunctor}
  \overset{
    \mathclap{
    \mbox{
      \tiny
      \color{blue}
\begin{tabular}{c}
        Space(-time)
      \end{tabular}
    }
    }
  }{
    \overbrace{
      X
    }
  }
   \xymatrix{\ar[rr]&&}
  \overset{
    \mathclap{
    \mbox{
      \tiny
      \color{blue}
      \begin{tabular}{c}
        Cocycle space of
        \\
        $p$-Cohomotopy cohomology theory
        \\
        evaluated on $X$
      \end{tabular}
    }
    }
  }{
    \overbrace{
      \boldpi^p(X)
    }
  }
\end{equation}
embodies \emph{$p$-Cohomotopy theory}
as non-abelian (unstable) generalized cohomology theory.

\medskip

\noindent {\bf Differential cohomology.}
A refinement of Cohomotopy cohomology theory to a
\emph{differential} (or \emph{geometric}) non-abelian generalized cohomology theory
is an assignment
\vspace{-2mm}
\begin{equation}
  \label{GeneralFormOfDifferentialCohomotopy}
  \overset{
    \mathclap{
    \mbox{
      \tiny
      \color{blue}
      \begin{tabular}{c}
        Space(-time)
      \end{tabular}
    }
    }
  }{
    \overbrace{
      X
    }
  }
     \xymatrix{\ar[rr]&&}
  \overset{
    \mathclap{
    \mbox{
      \tiny
      \color{blue}
      \begin{tabular}{c}
        {\color{magenta}Geometric}
        cocycle space
        \\
        of
        {\color{magenta}differential} $p$-Cohomotopy theory
        \\
        evaluated on $X$
      \end{tabular}
    }
    }
  }{
    \overbrace{
      \boldpi^p_{\mathrm{diff}}(X)
    }
  }
\end{equation}
of a \emph{geometric cocycle space} $\boldpi^p_{\mathrm{diff}}(X)$, of sorts,
which may depend on geometric data carried by $X$, but which
is such that the underlying homotopy type
$\mbox{\textesh}\boldpi^p(X)$ of the geometric cocycle space
is homotopy equivalent to that of the bare cocycle space \eqref{CohmotopyCocycleSpace}:
\begin{equation}
  \label{ConditionOnDifferentialCohomotopy}
  \underset{
    \mathclap{
    \mbox{
      \tiny
      \color{blue}
      \begin{tabular}{c}
        Underlyng plain homotopy type
        \\
        of geometric cocycle space
      \end{tabular}
    }
    }
  }{
  \underbrace{
  \mbox{\textesh}
  \big(\;
    \overset{
      \mathclap{
      \mbox{
        \tiny
        \color{blue}
        \begin{tabular}{c}
          Geometric
          \\
          cocycle space
        \end{tabular}
      }
      }
    }{
      \overbrace{
        \boldpi^p_{\mathrm{diff}}(X)
      }
    }
  \; \big)
  }
  }
  \;\;\;\;\;\;\;
  \underset{
    \mathclap{
    \mbox{
      \tiny
      hmtpy
    }
    }
  }{\simeq}
  \;\;\;\;\;\;\;
  \underset{
    \mathclap{
    \mbox{
      \tiny
      \color{blue}
      \begin{tabular}{c}
        Plain Cohomotopy
        \\
        cocycle space
      \end{tabular}
    }
    }
  }{
    \underbrace{
      \pi^p(X)
    }
  }\;.
\end{equation}
In full generality, $\boldpi^p(X)$ may be a \emph{cohesive $\infty$-stack}, but for
our purpose here it is sufficient to allow $\boldpi^p(X)$ to be a manifold,
or even just a topological space (understood up to homeomorphism,
instead of up to homotopy equivalence), which is a special
simple example of cohesive $\infty$-stacks. In this simple case
the operation $\mbox{\textesh}(-)$ of computing underlying homotopy types
is just the usual way of regarding a topological space as
a representative of its homotopy type, and hence we will not
further display it.

\medskip

\noindent
{\bf Configuration spaces as differential Cohomotopy cocycle spaces.}
The following statement provides a solution to the
constraint \eqref{ConditionOnDifferentialCohomotopy} on
a differential refinement of Cohomotopy cohomology theory,
in the case when $X$ is a Penrose diagram space \eqref{PenroseDiagramSpaces}.
Applying the results from \cite[2.7]{May72}\cite[3]{Segal73} in our setting
leads us to the following.

\begin{prop}[Labelled configuration spaces via Cohomotopy cocycles]
 \label{MaySegalTheoorem}
 For any natural numbers $d < p \in \mathbb{N}$, the
 un-ordered configuration space
 $\mathrm{Conf}\big( \mathbb{R}^d, \mathbb{D}^{p-d}\big)$
 of points in $\mathbb{R}^d$
 with labels in $\mathbb{D}^{p-d}$ (Def. \ref{ConfigurationSpaces})
 has the homotopy type of the plain $p$-Cohomotopy cocycle space
 \eqref{CohmotopyCocycleSpace} of the one-point compactified
 $d$-dimensional Euclidean
 space $(\mathbb{R}^d)^{\mathrm{cpt}}$ \eqref{OnePointCompactificationOfEuclideanSpace}:
\begin{equation}
  \label{CohomotopyChargeMap}
  \hspace{-2cm}
  \xymatrix@C=6em{
    \underset{
      \mbox{
        \tiny
        \color{blue}
        \begin{tabular}{c}
          $\phantom{a}$
          \\
          Un-ordered configuration space
          \\
          of points in $\mathbb{R}^d$
          \\
          with labels in $\mathbb{D}^{p-d}$
        \end{tabular}
      }
    }{
      \mathrm{Conf}\big( \mathbb{R}^d, \mathbb{D}^{p-d}  \big)
    }
    \ar[rr]^-{
      \mathclap{
      \mbox{
        \tiny
        \color{blue}
        \begin{tabular}{c}
          send configuration of points
          \\
          to their Cohomotopy charge
          \\
          $\phantom{a}$
        \end{tabular}
      }
      }
    }_-{
      \underset
      {
        \mathclap{
        \mbox{
          \tiny
          hmtpy
        }
        }
      }
      {
        \simeq
      }
    }
  &&
    \underset{
      \mathclap{
      \mbox{
        \tiny
        \color{blue}
        \begin{tabular}{c}
          $\phantom{a}$
          \\
          Cocycle space of
          \\
          $p$-Cohomotopy cohomology theory
          \\
          on the one-point compactification
          \\
          of $d$-dim Euclidean space
        \end{tabular}
      }
      }
    }{
      \boldpi^{p}
      \big(
        (\mathbb{R}^{d})^{\mathrm{cpt}}
      \big)
    }
  }
  \mathrlap{
    \phantom{AAAA}
    \mbox{for $d < p$}.
  }
\end{equation}
\end{prop}
%\begin{proof}
%This follows by applying \cite[2.7]{May72}\cite[3]{Segal73} in our setting.
%\end{proof}
\begin{remark}[Cohomotopy charge map]
  The \emph{Cohomotopy charge map} \eqref{CohomotopyChargeMap}
  is described in detail in \cite{SS19a}, with many illustrations, and generalized to
  equivariant Cohomotopy of flat orbifolds.  Notice that this map has originally
  been called the \emph{electric field map} \cite{Segal73},  in an attempt to think
  of it as assigning a physical field sourced by a configuration of charged points.
 While this physics interpretation seems to superficially make sense for representative
 maps, it is incompatible with the passage to homotopy classes on the right side of
 \eqref{CohomotopyChargeMap} (which does not reflect the passage to  gauge equivalence
 classes of electric fields). Instead, the claim of \hyperlink{HypothesisH}{Hypothesis H}
  is that  the actual physics interpretation of the Cohomotopy charge map
  \eqref{CohomotopyChargeMap} is as assigning brane charge in M-theory.
\end{remark}
\begin{example}[Unlabeled from labeled]
  The special case of Prop. \ref{MaySegalTheoorem} with
  $d = 0$ is evident:
  $$
  \xymatrix{
    \boldpi^{p}
    \big(\,
      \underset{
        = S^0
      }{
      \underbrace{
        (\mathbb{R}^{0})^{\mathrm{cpt}}
      }
      }\,
    \big)
    \ar[rr]_{
      \underset
      {
        \mathclap{
        \mbox{
          \tiny
          hmtpy
        }
        }
      }
      {
        \simeq
      }
    }
  &&
  \mathrm{Conf}\big( \mathbb{R}^0, \mathbb{D}^{p}  \big)
  }
  $$
  since now the
  left hand side is the space of maps from a
  single point to $S^p$,
  while right hand side is the space of labels in
  $S^p$ carried by a single point. Both of these spaces
  are canonically homeomorphic to $S^p$ itself.
\end{example}
But there is an alternative equivalence pertaining to this degenerate case,
which is again non-trivial.
  Applying  \cite[p. 95]{McDuff75}\cite[Example 11]{Boedigheimer87} to our setting we get
  the following.
\begin{prop}[Configurations vanishing at the boundary]
  \label{ConfigurationsOnRdVanishingAtTheBoundary}
  There is a homotopy equivalence
  \vspace{-2mm}
  $$
    \xymatrix{
    \boldpi^{p}
    \big(
      (\mathbb{R}^{0})^{\mathrm{cpt}}
    \big)
      \ar[rr]^-{\rm
                 hmtpy
          }_
          {\simeq}
            &&
      \mathrm{Conf}
      \big(
        \mathbb{D}^p, \mathbb{D}^0
      \big)
    }.
  $$
\end{prop}
%\proof
%  This follows from \cite[p. 95]{McDuff75}\cite[Example 11]{Boedigheimer87}.
%\endproof

\vspace{1mm}
Hence in the degenerate case of $d = 0$, the combination of
Prop. \ref{MaySegalTheoorem}
%\ref{CohomotopyChargeMap}
and Prop. \ref{ConfigurationsOnRdVanishingAtTheBoundary} is the
statement that we have a diagram of homotopy equivalences as follows:

\vspace{-2mm}
\begin{equation}
\hspace{-2cm}
  \xymatrix@C=4em@R=-.7em{
    &
    \overset{
      \mbox{
        \tiny
        \color{blue}
        \begin{tabular}{c}
          Cohomotopy cocycle space
          \\
          of the point
          \\
          $\phantom{a}$
        \end{tabular}
      }
    }{
      \boldpi^p
      \big(
        (\mathbb{R}^0)^{\mathrm{cpt}}
      \big)
    }
    \ar[dl]_-{ \underset{ \mathclap{\mbox{\tiny hmtpy}} }{\simeq} }
    \ar[dr]^-{ \underset{ \mathclap{\mbox{\tiny hmtpy}} }{\simeq} }
    \\
    \underset{
      \mathclap{
      \mbox{
        \tiny
        \color{blue}
        \begin{tabular}{c}
          $\phantom{a}$
          \\
          Configuration space of
          \\
          un-ordered points in $\mathbb{R}^0$
          \\
          (which can be at most one point)
          \\
          carrying a label in $S^p$
        \end{tabular}
      }
      }
    }{
      \mathrm{Conf}( \mathbb{D}^0, \mathbb{D}^p )
    }
    \ar[rr]_-{
      \underset{ \mathclap{\mbox{\tiny hmtpy}} }{\simeq}
    }
    &&
    \underset{
      \mathclap{
      \mbox{
        \tiny
        \color{blue}
        \begin{tabular}{c}
          $\phantom{a}$
          \\
          Configuration space of
          \\
          unordered points in $\mathbb{D}^p$
          \\
          carrying no label
          \\
          but vanishing when at $\infty \in \mathbb{D}^p/\partial \mathbb{D}^0$
        \end{tabular}
      }
      }
    }{
      \mathrm{Conf}( \mathbb{D}^p, \mathbb{D}^0)
      \mathrlap{
        \;
        =:
        \mathrm{Conf}(\mathbb{D}^p)\;.
      }
    }
  }
\end{equation}

With this we may finally state the main concept of this
section, and prove its consistency:

\begin{prop}[Differential Cohomotopy on Penrose diagrams via configuration spaces]
  \label{DifferentialCohomotopyOnPenroseDiagrams}
  For any $p \in \mathbb{N}$,
  and for spacetimes in the category
  \eqref{CategoryOfPenroseDiagrams}
  of Penrose diagrams \eqref{PenroseDiagramSpaces},
  a consistent enhancement of plain $p$-Cohomotopy cohomology theory
  \eqref{PlainCohomotopyFunctor} to
  a geometric/differential cohomology theory \eqref{GeneralFormOfDifferentialCohomotopy}, hence satisfying
  the condition \eqref{ConditionOnDifferentialCohomotopy},
  is given by the configuration space functor \eqref{ConfigurationSpaceFunctor}:
  $$
    \underset{
      \mathclap{
      \mbox{
        \tiny
        \color{blue}
        \begin{tabular}{c}
          Geometric cocycle spaces
          \\
          of
          differential
          Cohomotopy
        \end{tabular}
      }
      }
    }{
      \boldpi^p_{\mathrm{diff}}
    }
    \hspace{8mm} := \hspace{5mm}
    \underset{
      \mathclap{
      \mbox{
        \tiny
        \color{blue}
        \begin{tabular}{c}
          on Penrose diagram
          \\
          space(-times)
        \end{tabular}
      }
      }
    }{
      \mathrm{PenroseDiag}_p^{\mathrm{op}}
    }
           \xymatrix{\ar[rrr]^-{\mathrm{Conf}} &&&}
    \underset{
      \mathclap{
      \mbox{
        \tiny
        \color{blue}
        \begin{tabular}{c}
          $\phantom{aaaaaaaaaa}$
          \\
          are configuration spaces of points
          \\
          regarded as actual topological spaces
        \end{tabular}
      }
      }
    }
    {
      \mathrm{Top}^{\ast/}
    }.
  $$
\end{prop}
\begin{proof}
 The assignment $X \mapsto \boldpi^p(X)$ of
  homotopy types of plain Cohomotopy cocycle spaces \eqref{CohmotopyCocycleSpace}
  is homotopy invariant in $X$. Hence the uncompactified
  factors $(\mathbb{R}^{p-d})_+$ in the Penrose diagrams
  \eqref{PenroseDiagramSpaces}, being homotopy-contractible,
  do not contribute to the homotopy type
  of the plain Cohomotopy cocycle spaces:
  $$
    \boldpi^p
    \big(
      (\mathbb{R}^d)^{\mathrm{cpt}}
      \wedge
      (\mathbb{R}^{p-d})_+
    \big)
    \;\;\;
    \underset{
      \mathclap{
      \mbox{
        \tiny
        hmtpy
      }
      }
    }{\simeq}
    \;\;\;
    \boldpi^p
    \big(
      (\mathbb{R}^d)^{\mathrm{cpt}}
    \big)
    \,.
  $$
  With this, it follows that  Prop. \ref{MaySegalTheoorem}
  implies that condition \eqref{ConditionOnDifferentialCohomotopy}
  is satisfied for $d \geq 1$
  \vspace{-2mm}
$$
 \hspace{3.1cm}
  \xymatrix@C=3em{
    \mathllap{
      \mbox{
        \tiny
        \color{blue}
        \begin{tabular}{c}
          Cohomotopy charge...
        \end{tabular}
      }
      \;
    }
    \boldpi^{p}
    \big(
      (
        \overset{
          \mathclap{
          \mbox{
            \tiny
            \color{blue}
            \begin{tabular}{c}
              ...vanishing at
              infinity
              \\
              along
              these directions...
            \end{tabular}
          }
        }
      }{
        \overbrace{
          \mathbb{R}^{d}
        }
      }
      )^{\mathrm{cpt}}
      \wedge
      (
      \underset{
        \mathclap{
        \mbox{
          \tiny
          \color{blue}
          \begin{tabular}{c}
            ...but not necessarily
            \\
            along these directions...
          \end{tabular}
        }
        }
      }{
        \underbrace{
          \mathbb{R}^{p-d}
        }
      }
    )_+
    \big)
    \ar[rrr]^-{
      \mbox{
        \tiny
        \color{blue}
        ...is equivalent to...
      }
    }
    _-{
      \underset
      {
        \mathclap{
        \mbox{
          \tiny
          hmtpy
        }
        }
      }
      {
        \simeq
      }
    }
  &&&
  \overset{
    \mathclap{
    \mbox{
      \tiny
      \color{blue}
      \begin{tabular}{c}
        ...configurations of
        charged points in $\mathbb{R}^p$,
        \\
        which may disappear to infinity...
      \end{tabular}
    }
    }
  }{
    \mathrm{Conf}
    \big(
      \mathbb{R}^d,
      \underset{
        \mathclap{
        \mbox{
          \tiny
          \color{blue}
          \begin{tabular}{c}
            ... only along
            \\
            the latter directions.
          \end{tabular}
        }
        }
      }{
        \underbrace{
          \mathbb{D}^{p-d}
        }
      }
    \big)
  }
  }
  \phantom{AAAA}
  \mbox{for $d \geq 1$}
$$
and Prop. \ref{ConfigurationsOnRdVanishingAtTheBoundary}
implies that condition \eqref{ConditionOnDifferentialCohomotopy}
holds for $d = 0$:
  $$
    \mathllap{
      \mbox{
        \tiny
        \color{blue}
        Cohomotopy charge...
      }
      \;\;
    }
    \xymatrix@C=3em{
    \boldpi^{p}
    \big(
      \overset{
        \mathclap{
        \mbox{
          \tiny
          \color{blue}
          \begin{tabular}{c}
            .. not necessarily
            vanishing at infinity
            \\ in any direction...
          \end{tabular}
        }
        }
      }{
        \underset{
          = (\mathbb{R}^p)_+
        }{
        \underbrace{
          (\mathbb{R}^{0})^{\mathrm{cpt}}
          \wedge
          (\mathbb{R}^p)_+
        }
        }
      }
    \big)
      \ar[rrr]^-{
        \mbox{
          \tiny
          \color{blue}
          ...is equivalent to...
        }
      }_-{
        \underset{
          \mathclap{
          \mbox{
            \tiny
            hmtpy
          }
          }
        }{\simeq}
      }
      &&&
      \overset{
        \mathclap{
        \mbox{
          \tiny
          \color{blue}
          \begin{tabular}{c}
            ...configurations of charged points in $\mathbb{R}^p$,
            \\
            which may disappear to infinity...
          \end{tabular}
        }
        }
      }{
      \mathrm{Conf}
      \big(
        \underset{
          \mathclap{
          \mbox{
            \tiny
            \color{blue}
            ...in any direction.
          }
          }
        }{
          \underbrace{
            \mathbb{D}^p
          }
        }
        \big)
      }
      }.
  $$
 In summary:
 $$
  \boldpi^p_{\mathrm{diff}}
  \;:\;
  \left\{
  \raisebox{25pt}{
  \xymatrix@R=-2pt{
    (\mathbb{R}^d)^{\mathrm{cpt}}
    \wedge
    (\mathbb{R}^{p-d})_+
    \; \ar@{|->}[rr]
    &
    &
    \mathrm{Conf}
    \big(
      \mathbb{R}^d
      ,
      \mathbb{D}^{p-d}
    \big)
    \ar@{}[r]|-{
    \underset{
      \mbox{
        \tiny
        hmtpy
      }
    }{\simeq}
    }
    &
    \boldpi^p
    \big(
      (\mathbb{R}^d)^{\mathrm{cpt}}
      \wedge
      (\mathbb{R}^{p-d})_+
    \big)
    &
    \mbox{for $d \geq 1$}
    \\
    \underset{
      = (\mathbb{R}^p)_+
    }{
    \underbrace{
      (\mathbb{R}^0)^{\mathrm{cpt}}
      \wedge
      (\mathbb{R}^{p})_+
    }
    }
    \;\ar@{|->}[rr]
    &
    &
    \mathrm{Conf}
    \big(
      \mathbb{D}^{p}
    \big)
    \ar@{}[r]|-{
    \underset{
      \mbox{
        \tiny
        hmtpy
      }
    }{\simeq}
    }
    &
    \boldpi^p
    \big(
      (\mathbb{R}^0)^{\mathrm{cpt}}
      \wedge
      (\mathbb{R}^{p})_+
    \big)
    &
    \mbox{for $d = 0$}
  }
  }
  \right.
$$
and hence condition \eqref{ConditionOnDifferentialCohomotopy}
is verified.
\hspace{11.2cm} \end{proof}

\medskip

%%%%%%%%%%%%%%%%%%%%%%%%%%%%%%%%%%%%%%%%%%%%%%%%%%%%%%%%%%%%%%%%%%%%%%%%
\subsection{Intersecting brane charges in differential Cohomotopy }
\label{IntersectingD6D8BraneChargesInDifferentialCohomotopy}
%%%%%%%%%%%%%%%%%%%%%%%%%%%%%%%%%%%%%%%%%%%%%%%%%%%%%%%%%%%%%%%%%%%%%%%%

With \hyperlink{HypothesisH}{Hypothesis H}, we now assume that the differential
4-Cohomotopy theory of Prop. \ref{DifferentialCohomotopyOnPenroseDiagrams}
reflects brane charges in string/M-theory on Penrose diagram spaces
\eqref{PenroseDiagramSpaces}, and explore the consequences.
By the discussion of charges vanishing at infinity in \cref{VanishingAtInfinity},
we expect that the differential 4-Cohomotopy on the Penrose diagram space
$(\mathbb{R}^d)^{\mathrm{cpt}} \wedge (\mathbb{R}^{4-d})_+$
reflects charges of branes of codimension $d$.
Indeed, for $d = 4$ we found an accurate picture of
$\mathrm{MK6}$-charges from Cohomotopy in \cite{SS19a}.
Now to speak about intersecting branes means to consider the Cohomotopy charge
of unions of Penrose diagram spaces, which makes sense in the topological presheaf
topos over the site of Penrose diagram 4-spaces from
Def. \ref{ContravariantlyFunctorialConfigurationSpaces}.
\begin{defn}[Union of Penrose diagram spaces]
  \label{UnionOfPenroseDiagramSpaces}
  For $ 0 \leq d \leq 4$ write
  \begin{equation}
    \label{UnionOfPenroseDiagramSpaces}
      (\mathbb{R}^{d})^{\mathrm{cpt}}
      \wedge
      (\mathbb{R}^{4-d})_+
      \;\;\cup\;\;
      (\mathbb{R}^d)_+
      \wedge
      (\mathbb{R}^{4-d})^{\mathrm{cpt}}
    \;\;\in\;\;
    \mathrm{Sh}\big( \mathrm{PenroseDiag}_4, \mathrm{Top}^{\!\ast/}\big)
  \end{equation}
  (see the left half of \eqref{DifferentialCohomotopyOfIntersections} below)
  for the union,
  with respect to the canonical inclusion maps of
  Example \ref{MapsOfConfigurationSpacesForOrderedFiberProduct},
  of Penrose diagram spaces \eqref{PenroseDiagramSpaces},
  regarded as representables
  in the topological presheaf topos over the site \eqref{CategoryOfPenroseDiagrams}.
\end{defn}
By the discussion in \cref{VanishingAtInfinity},
the generalized space \eqref{UnionOfPenroseDiagramSpaces}
may be regarded as the transversal space to
the intersection of charged objects of
codimension-$d$ with those of codimension-$(4-d)$. Indeed, we establish the following.
%We then have the following statement:

\begin{prop}[Differential Cohomotopy and configuration spaces]
 \label{FinalResult}
The geometric cocycle space \eqref{GeneralFormOfDifferentialCohomotopy}
that the differential 4-Cohomotopy  theory
from Prop. \ref{DifferentialCohomotopyOnPenroseDiagrams}
assigns to the transversal space \eqref{UnionOfPenroseDiagramSpaces}
for $d =3$
has the homotopy type of the ordered configuration space
of points in $\mathbb{R}^3$ (Def. \ref{ConfigurationSpaces}):

 \vspace{-.8cm}
\begin{equation}
  \label{DifferentialCohomotopyOfIntersections}
  \xymatrix@C=-4pt@R=-15pt{
      \overset{
        \mathclap{
        \mbox{
          \tiny
          \color{blue}

  }
  }
  }
  \!\!\!\!\!\!\!\!\!\!\!\!\!\!\!
  \!\!\!\!\!\!\!\!\!\!\!\!\!\!\!
  \!\!\!\!\!\!\!\!
  \right\}
  }
\end{equation}
\end{prop}
\begin{proof}
  Being given by a contravariant functor \eqref{ConfigurationSpaceFunctor},
  the assignment $\boldpi^4_{\mathrm{diff}}$
  takes the union (cofiber coproduct) of representable presheaves on the
  left to the intersection (fiber product) of its values
  on the cofactors. This fiber product is just the one
  appearing on the right of \eqref{TheFiberProductAndItsHomotopyEquivalence}.
  Hence the statement follows by
  Prop. \ref{OrderedUnlabeledAsFiberProduct}.
\hspace{2.5cm} \end{proof}

In conclusion, the following
Remarks \ref{SingleOutP6}, \ref{MassiveTypeIPrime},
\ref{GeometricEngineeringOfMonopoles} highlight how, in the
above discussion, the dimensions conspire, starting with the
degree 4 of 4-Cohomotopy due to \hyperlink{HypothesisH}{\it Hypothesis H}:

\begin{remark}[Distinguished system]
  \label{SingleOutP6}
  The case $d = 3$ (equivalently $d =1$) in Def. \ref{UnionOfPenroseDiagramSpaces},
  hence $p = 6$,
  is singled out as being the mathematically exceptional one:
  For $d \in \{0,2,4\}$ the corresponding analog of Prop.
  \ref{FinalResult} produces a fiber product of unordered configuration
  spaces with fairly uninteresting cohomology. It is
  only in the case of codimensions
  $1 = 4-3$
  that, via Prop. \ref{OrderedUnlabeledAsFiberProduct},
  a linear ordering on the points is induced,
  thus of Chan-Paton labels on the corresponding
  branes, leading to the rich observables found in
  \cref{WeightSystemsAsObservablesOnIntersectingBranes}.
\end{remark}

\begin{remark}[Massive Type I']
  \label{MassiveTypeIPrime}
  Following the discussion of \hyperlink{HypothesisH}{Hypothesis H}
  in \cite{FSS19d}\cite{SS19a}, we are to think of
  Prop. \ref{FinalResult} as applying to non-perturbative massive
  type I' string theory, hence to heterotic M-theory.
  With no equivariance considered here, the Ho{\v r}ava-Witten interval
  becomes invisible in homotopy
  theory and the codimensions 3 \& 1 in Prop. \ref{FinalResult}
  are those of $\mathrm{D}6\perp\mathrm{D}8$ brane
  intersections in massive type I', as shown.
\end{remark}

\begin{remark}[Geometric engineering of monopoles]
\label{GeometricEngineeringOfMonopoles}
For any $p \in \{0,1, \cdots, 6\}$ (at least)
transversal $\mathrm{D}p \perp \mathrm{D}(p+2)$-brane
intersections geometrically engineer Yang-Mills monopoles
(i.e. Donaldson-Atiyah-Hitchin-style monopoles \cite{AtiyahHitchin88}\cite{Donaldson84}
characterized by Nahm's equation) in the
worldvolume theory of the $\mathrm{D}(p+2)$-brane.

\medskip

   };

 \draw
   (-5.3,-3.2)
   node
   {
     \tiny
     $
     1
     $
   };

 \draw
   (-5.3+.2,-3.2)
   node
   {
     \tiny
     $
     2
     $
   };

 \draw
   (-5.3+.4,-3.2)
   node
   {
     \tiny
     $
     3
     $
   };

 \draw
   (-5.3+.6,-3.2)
   node
   {
     \tiny
     $
     4
     $
   };

 \draw
   (-5.3+.8,-3.2)
   node
   {
     \tiny
     $
     5
     $
   };

 \draw
   (-5.3+1,-3.2)
   node
   {
     \tiny
     $
     6
     $
   };

  \draw (-5.9,-1.5)
    node
    {
      \tiny
      \color{blue}
      monopole
    };

  \end{scope}

  \draw (-5.9+2,-1.5)
    node
    {
      \tiny
      \color{blue}
      D6
    };

  \draw (-5.9+1.4,-1.5-1.4)
    node
    {
      \tiny
      \color{blue}
      D8s
    };

\end{tikzpicture}
}
\end{tabular}

\medskip
\noindent In this case of $p = 6$, \cite{HLPY08} observe that monopoles engineered as
$\mathrm{D6} \perp \mathrm{D8}$-intersections yield
the actual 4d monopoles of nuclear physics, through the
Sakai-Sugimoto model for QCD
\cite{Witten98}\cite{SakaiSugimoto04}\cite{SakaiSugimoto05}
(for review see \cite{Rebhan14}\cite{Sugimoto16}):

\medskip

\begin{center}
\begin{tabular}{rrcccccccccc}
  \multicolumn{2}{c}{
  \raisebox{20pt}{
  \small
  \fbox{
  \begin{tabular}{c}
    \bf D-brane configuration
    \\
    \bf geometrically engineering
    \\
    \bf quantum chromodynamics
    \\
    (Witten-Sakai-Sugimoto model)
  \end{tabular}
  }
  }
  }
  &
  $
    \overset{
      \mathclap{
      \mbox{
        \tiny
        \color{blue}
        \begin{tabular}{c}
          Small
          \\
          extra
          \\
          dimension
          \\
          $\phantom{a}$
        \end{tabular}
      }
      }
    }{
      \mbox{\color{brown}$S^1$}
    }
  $
  &
  $\times$
  &
  $
    \overset{
      \mathclap{
      \mbox{
        \tiny
        \color{blue}
        \begin{tabular}{c}
          space
          \\
          $\phantom{a}$
        \end{tabular}
      }
      }
    }{
      \Sigma^3
    }
  $
  &
  $\times$
  &
   $
    \overset{
      \mathclap{
      \mbox{
        \tiny
        \color{blue}
        \begin{tabular}{c}
          time
          \\
          $\phantom{a}$
        \end{tabular}
      }
      }
    }{
      \mathbb{R}^{0,1}
    }
  $
 &
  $\times$
  &
   $
    \overset{
      \mathclap{
      \mbox{
        \tiny
        \color{blue}
        \begin{tabular}{c}
          radial
          \\
          $\phantom{a}$
        \end{tabular}
      }
      }
    }{
      \mathbb{R}^1_{\geq 1}
    }
  $
  &
  $
    \overset{
      \mathclap{
      \mbox{
      \tiny
      \color{blue}
      \hspace{-.23cm}
      \begin{tabular}{c}
        Large extra dimensions
        \\
        $\phantom{a}$
        \\
        $\phantom{a}$
        \\
        $\phantom{a}$
      \end{tabular}
      }
      }
    }{
      \times
    }
  $
  &
   $
    \overset{
      \mathclap{
      \mbox{
        \tiny
        \color{blue}
        \begin{tabular}{c}
          angular
          \\
          $\phantom{a}$
        \end{tabular}
      }
      }
    }{
      S^4
    }
  $
\\
 $N_{c}$ $\phantom{\mbox{monopole}}\mathllap{\mbox{color}}$ branes
 &
  ${\mathrm{D4}}_{\mathrlap{\mathrm{col}}\phantom{\mathrm{mon}}}$
    &
    \multicolumn{5}{l}{
      \colorbox{gray}{
      \hspace{-.3cm}
      ------------------------------
      \hspace{-.3cm}
      }
    }
  \\
  $N_f$ $\phantom{\mbox{monopole}}\mathllap{\mbox{flavor}}$ branes
  &
  ${\color{darkblue}\mathrm{D8}}_{\mathrlap{\mathrm{fla}}\phantom{\mathrm{mon}}}$
  &
  &&
  \multicolumn{7}{l}{
         \colorbox{darkblue}{
         \hspace{-.3cm}
         --------------------------------------------------
         \hspace{-.3cm}
         }
     }
  \\
  \multirow{2}{*}{meson fields}
  &
  ${\color{darkblue}\mathrm{CS}5}_{\mathrlap{\mathrm{fla}}\phantom{\mathrm{mon}}}$
  &
  & &
  \multicolumn{5}{l}{
    \colorbox{darkblue}{
    \hspace{-.3cm}
      ------------------------------------
    \hspace{-.3cm}
    }
  }
  &
  \\
  &
  ${\color{darkblue}\mathrm{WZ}4}_{\mathrlap{\mathrm{fla}}\phantom{\mathrm{mon}}}$
  &
  & &
  \multicolumn{3}{l}{
    \colorbox{darkblue}{
    \hspace{-.3cm}
      --------------------
    \hspace{-.3cm}
    }
  }
  &
  \\
  $N_b$ $\phantom{\mbox{monopole}}\mathllap{\mbox{baryon}}$ branes
  &
  ${\color{cyan}\mathrm{D4}}_{\mathrlap{\mathrm{bar}}\phantom{\mathrm{mon}}}$
  &&&&&  \!\!\!\!
         \colorbox{cyan}{
         \hspace{-.3cm}
         ------
         \hspace{-.3cm}
         }
  &&&&
         \colorbox{cyan}{
         \hspace{-.3cm}
         ---
         \hspace{-.3cm}
         }
  \\
 $N_{m}$ monopole branes
  &
  ${\color{orangeii}\mathrm{D}6}_{\mathrm{mon}}$
  &
    \colorbox{orangeii}{
    \hspace{-.3cm}
      ---
    \hspace{-.3cm}
    }
  &&&&
  \multicolumn{5}{l}{
    \colorbox{orangeii}{
    \hspace{-.3cm}
      -------------------------------------
    \hspace{-.3cm}
    }
  }
  \\

  &
  ${\color{green!65!red}\mathrm{NS}5}_{\phantom{\mathrm{mon}}}$
  &
  &&&&
  \multicolumn{5}{l}{
    \colorbox{green!65!red}{
    \hspace{-.3cm}
      -------------------------------------
    \hspace{-.3cm}
    }
  }
\end{tabular}
\end{center}
Under this identification and via Prop. \ref{FinalResult},
the statements about fuzzy funnel
observables in \cref{su2WeightSystemsAreFuzzyFunnelObservables}
translate to statements about QCD monopoles.
\end{remark}

%%%%%%%%%%%%%%%%%%%%%%%%%%%%%%%%%%%%%%%%%%%%%%%%%%%%%%%%%%%%%%
\subsection{Higher observables on intersecting brane configurations}
\label{HigherObservablesOnIntersectingBraneModuli}
%%%%%%%%%%%%%%%%%%%%%%%%%%%%%%%%%%%%%%%%%%%%%%%%%%%%%%%%%%%%%%

\noindent
{\bf Topological covariant phase spaces.}
We consider the following setting:
\begin{enumerate}[{\bf (i)}]
\vspace{-2mm}
\item  Any assignment $\mathbf{Fields}$  of spaces of
field configurations, such as the
cohomotopically charge-quantized C-field
$\mathbf{Fields} := \boldpi^4_{\mathrm{diff}}$ of Prop. \ref{DifferentialCohomotopyOnPenroseDiagrams}.

\vspace{-2mm}
\item
$X$ a spatial slice of spacetime, hence with
$\mathbf{Fields}(X)$ its field configuration space.

\vspace{-2mm}
\item $c_{\mathrm{in}}, c_{\mathrm{out}} \in \mathbf{Fields}(X)$
two field configurations in the same connected component.
\end{enumerate}
\vspace{-2mm}
\noindent Then we may think of the the based path space
\begin{equation}
  \label{TopologicalPhaseSpaceElement}
  \hspace{0mm}
  \overset{
    \mathclap{
    \mbox{
      \tiny
      \color{blue}
      \begin{tabular}{c}
        Based path space in
        \\
        field configurations
        \\
        $\phantom{a}$
      \end{tabular}
    }
    }
  }{
    P_{c_{\mathrm{in}}}^{c_{\mathrm{out}}} \mathbf{Fields}(X)
  }
  \;:=
  \Big\{ \!\!\!
    \left.
    c
    \in
    \mathrm{Maps}
    \big(
      [0,1],
      \,
      \mathbf{Fields}(X)
    \big)
    \right\vert
    c(0) = c_{\mathrm{in}},
    c(1) = c_{\mathrm{out}}
  \! \Big\}
    \underset{
    \mbox{
      \tiny
      hmtpy
    }
  }{\simeq}
    P_{c_{\mathrm{in}}}^{c_{\mathrm{in}}} \mathbf{Fields}(X)
  \;
  =:
  \;
  \overset{
    \mathclap{
    \mbox{
      \tiny
      \color{blue}
      \begin{tabular}{c}
        Based loop space in
        \\
        field configurations
        \\
        $\phantom{a}$
      \end{tabular}
    }
    }
  }{
    \Omega_{c_{\mathrm{in}}} \mathbf{Fields}(X)
  }
\end{equation}
as an element of the
\emph{covariant phase space},  each of which
represents a field history evolving from $c_{\mathrm{in}}$
to $c_{\mathrm{out}}$.
Any fixed choice of such field
history induces (by evolving back along it)
a homotopy equivalence to the based loop space
of the cocycle space, as shown on the right in
\eqref{TopologicalPhaseSpace}.
This, in turn
is independent, up to homotopy, from the choice of basepoint.
Therefore we may regard the disjoint union of the construction
\eqref{TopologicalPhaseSpaceElement}
over the connected components
$[c] \in \pi_0\big( \mathbf{Fields}(X)  \big)$ of field configurations
as the \emph{topological covariant phase space}
\begin{equation}
  \label{TopologicalPhaseSpace}
  \overset{
    \mathclap{
    \mbox{
      \tiny
      \color{blue}
      \begin{tabular}{c}
        Topological
        \\
        covariant phase space
        \\
        $\phantom{a}$
      \end{tabular}
    }
    }
  }{
    \mathbf{Phase}(X)
  }
 \;\; \;:=\; \;\;
  \underset{
    \mathclap{
    \mbox{
      \tiny
      \color{blue}
      \begin{tabular}{c}
        Disjoint union over
        \\
        connected components
      \end{tabular}
    }
    }
  }{
    \underset{
      [c]
    }{\sqcup}
  }
  \;\;
  \overset{
    \mathclap{
    \mbox{
      \tiny
      \color{blue}
      \begin{tabular}{c}
        Based loop space of
        \\
        field configuration space
      \end{tabular}
    }
    }
  }{
  \overbrace{
    \Omega_c
    \mathbf{Fields}(X)
  }
  }\;.
\end{equation}
Without further equations of motion imposed on the
field histories this would be the \emph{off-shell}
phase space; but for our purposes here all
\emph{topological} constraints on the fields,
such as the ``integral equation of motion''
on the C-field \cite{DMW00a}\cite{DMW00b},
are enforced \cite{FSS19b}
by the cohomological charge quantization
in the cohomology theory $\mathbf{Fields} = \boldpi^4_{\mathrm{diff}}$,
and therefore we do regard \eqref{TopologicalPhaseSpace}
as the topological sector of the full covariant phase space.

\medskip

\noindent {\bf Higher order observables.}
The \emph{observables} of a physical theory are traditionally
taken to be $\mathbb{F}$-valued functions on the covariant phase space,
hence functions with values in the given ground field.
But to do justice to the homotopy-theoretic nature of
fields charge-quantized in generalized cohomology theories,
following \cite{SS17},
we here take \emph{higher observables} to be
$H \mathbb{F}$-valued functions on the topological
covariant phase space \eqref{TopologicalPhaseSpace},
i.e., taking values in the Eilenberg-MacLane spectrum $H\mathbb{F}$
and its suspensions. After passage to gauge equivalence classes,
these higher observables hence form the cohomology ring of the
topological phase space \eqref{TopologicalPhaseSpace}:
\begin{equation}
  \label{HigherObservables}
  \begin{array}{ccccc}
  \mathllap{
    \mbox{
       \tiny
       \color{blue}
       \begin{tabular}{c}
         Higher observables
         \\
         $\phantom{a}$
       \end{tabular}
     }
     \;\;
   }
  \mathrm{Obs}^\bullet(X)
  &:=&
  H^\bullet
  \big(
    \mathbf{Phase}(X)
  \big)
  &:=&
  H^\bullet
  \Big(
    \underset{
      [c]
    }{\sqcup}
    \Omega_{c}
    \mathbf{Fields}(X)
  \Big)
  \\
  \mathllap{
    \mbox{
       \tiny
       \color{blue}
       \begin{tabular}{c}
         Higher co-observables
         \\
         $\phantom{a}$
       \end{tabular}
     }
     \;\;
   }
  \mathrm{Obs}_\bullet(X)
  &:=&
  H_\bullet
  \big(
    \mathbf{Phase}(X)
  \big)
  &:=&
  H_\bullet
  \Big(
    \underset{
      [c]
    }{\sqcup}
    \Omega_{c}
    \mathbf{Fields}(X)
  \Big)
  \end{array}
  \,.
\end{equation}

\noindent {\bf Higher observables on $\mathrm{D}6\perp \mathrm{D}8$-brane configurations.}
Specifying  the higher observables \eqref{HigherObservables} to the case where we consider,
with \hyperlink{HypothesisH}{\it Hypothesis H}:

\begin{enumerate}[{\bf (i)}]
\vspace{-2mm}
\item  $\mathbf{Fields} := \boldpi^4_{\mathrm{diff}}$  to be the C-field moduli of Prop. \ref{DifferentialCohomotopyOnPenroseDiagrams};

\vspace{-2mm}
\item $X :=
      (\mathbb{R}^{3})^{\mathrm{cpt}}
      \wedge
      (\mathbb{R}^{1})_+
      \;\;\cup\;\;
      (\mathbb{R}^3)_+
      \wedge
      (\mathbb{R}^{1})^{\mathrm{cpt}}
$
  to be the transversal space
      of $\mathrm{D}6 \perp \mathrm{D}8$-brane intersections
      according to Def. \ref{UnionOfPenroseDiagramSpaces};
\end{enumerate}
\vspace{-2mm}
\noindent
we are led to the following notion:
\begin{defn}
\label{HigherObsOnD6D8Intersections}
We take the algebra of
\emph{higher observables on configurations of
$\mathrm{D}6\perp\mathrm{D}8$-brane intersections}
(Remark \ref{GeometricEngineeringOfMonopoles})
to be the ordinary cohomology ring \eqref{HigherObservables}
of the componentwise based loop space \eqref{TopologicalPhaseSpace}
of the
differential 4-Cohomotopy cocycle space
(Prop. \ref{DifferentialCohomotopyOnPenroseDiagrams})
that is assigned
to the transversal space for codim=3/codim=1 brane intersections
(Def. \ref{UnionOfPenroseDiagramSpaces}):
\vspace{-4mm}
 \begin{equation}
  \label{HigherObservablesOnD6D8Intersections}
 \hspace{1.5cm}
  \overset{
    \mathclap{
    \mbox{
      \tiny
      \color{blue}
      \begin{tabular}{c}
        Higher observables on
        \\
        $\mathrm{D}6\perp \mathrm{D}8$-configurations
        \\
        by \hyperlink{HypoothesisH}{Hypothesis H}
        \\
        $\phantom{a}$
      \end{tabular}
    }
    }
  }{
    \mathrm{Obs}^\bullet_{\mathrm{D}6 \perp \mathrm{D}8}
  }
  \;\; :=\;\;
  H^\bullet
  \Big(
    \boldpi^4_{\mathrm{diff}}
    \big(
      (\mathbb{R}^{3})^{\mathrm{cpt}}
      \wedge
      (\mathbb{R}^{1})_+
      \;\;\cup\;\;
      (\mathbb{R}^3)_+
      \wedge
      (\mathbb{R}^{1})^{\mathrm{cpt}}
    \big)
  \Big).
\end{equation}
\end{defn}
\noindent With the results from \cref{ChargesInCohomotopy} we may
characterize these higher observables more concretely:
\begin{prop}[Higher observables as cohomology of looped configuration space]
  \label{HigherObservablesEquivalentToCohomologyOfLoopedConfigurationSpace}
  The algebra of higher observables on
  $\mathrm{D}6\perp \mathrm{D}8$-configurations
  \eqref{HigherObservablesOnD6D8Intersections}
  is isomorphic to the direct sum,
  over the number $N_{\mathrm{f}}$ of points,
  of the cohomology
  rings of the based loop spaces of configuration spaces
  (Def. \ref{ConfigurationSpaces}) of $N_{\mathrm{f}}$ points
  in Euclidean 3-space:
  \vspace{0mm}
  \begin{equation}
    \label{HigherObservablesAreCohomologyOfLoopedConfigurationSpace}
    \mathrm{Obs}^\bullet_{\mathrm{D}6 \perp \mathrm{D}8}
    \;\simeq\;
    \underset{
      N_{\mathrm{f}}
      \in \mathbb{N}
    }{\medoplus}
    H^\bullet
    \Big(
      \Omega
      \underset{
        {}^{\{1,\cdots, N_{\mathrm{f}}\}}
      }{\mathrm{Conf}}
      (\mathbb{R}^3)
    \Big)\;.
  \end{equation}
\end{prop}
\begin{proof}
  Using the fact that ordinary cohomology is invariant
  under homotopy equivalences, this follows with
  Prop. \ref{FinalResult}.
\hspace{16.25cm}\end{proof}

\begin{remark}
Via Prop. \ref{HigherObservablesEquivalentToCohomologyOfLoopedConfigurationSpace}
the higher co-observables \eqref{HigherObservables}
are identified with the
higher order OPEs of extended field theories
as considered in \cite{BBBDN18}.
\end{remark}

\medskip

\noindent {\bf Higher observables on $\mathrm{D}6\perp \mathrm{D}8$ are weight systems on chord diagrams.}
Remarkably, there is a combinatorial model for the cohomology ring
\eqref{HigherObservablesOnD6D8Intersections} of higher observables,
namely in terms of \emph{weight systems on chord diagrams}.
The definitions of these are reviewed in detail in \cref{WeightSystemsOnChordDiagrams} below. The reader may wish
to come back to
the following Prop. \ref{HigherObservablesOnIntersectingBranesAreWeightSystems} after
looking through \cref{WeightSystemsOnChordDiagrams}.

\begin{prop}[Cohomology of looped configuration space is horizontal weight systems]
\label{HigherObservablesOnIntersectingBranesAreWeightSystems}
For any natural number $N_{\mathrm{f}} \in \mathbb{N}$ we have\footnote{ This holds over any ground field $\mathbb{F}$ (such as the complex numbers),
and in fact more generally
over any commutative ring (such as the integers).}
for the based loop space of the ordered configuration
space
$
    \underset{
      {}^{\{1,\cdots,N_{\mathrm{f}}\}}
    }{
      \mathrm{Conf}
    }
    \!
    (\mathbb{R}^3)
$
of $N_{\mathrm{f}}$ points in $\mathbb{R}^3$
(Def. \ref{ConfigurationSpaces}) that:

\vspace{-2mm}
\item {\bf (i)}  Its homology Pontrjagin ring is isomorphic,
as a graded Hopf algebra (see \cite{Halperin92}), to the algebra
$\mathcal{A}^{{}^{\mathrm{pb}}}_{N_{\mathrm{f}}}$
\eqref{HorizontalChordDiagramsModulo2TAnd4TRelations}
of horizontal chord diagrams \eqref{HorizontalChordDiagrams}
with $N_{\mathrm{f}}$ strands
modulo the 2T-relations \eqref{2TRelationsOnHorizontalChordDiagrams}
and 4T relations \eqref{4TRelationOnHorizontalChordDiagrams}:
\begin{equation}
  \mbox{
    \tiny
    \color{blue}
\begin{tabular}{c}
      Homology ring of
      \\
      based loop space of
      \\
      configuration space of
      \\
      $N_{\mathrm{f}}$ ordered points in $\mathbb{R}^3$
    \end{tabular}
  }
  \phantom{AA}
  H_\bullet
  \Big(
    \Omega
    \underset{
      {}^{\{1,\cdots,N_{\mathrm{f}}\}}
    }{
      \mathrm{Conf}
    }
    \!
    (\mathbb{R}^3)
  \Big)
  \;\;\;\simeq\;\;\;
  \big(
    \mathcal{A}^{{}^{\mathrm{pb}}}_{N_{\mathrm{f}}}
  \big)_\bullet
  \phantom{AA}
  \mbox{
    \tiny
    \color{blue}
    \begin{tabular}{c}
      Concatenation algebra of
      \\
      horizontal chord diagrams
      \\
      with $N_{\mathrm{f}}$ strands
      \\
      modulo 2T- \& 4T-relations
    \end{tabular}
  }
\end{equation}
\vspace{-4mm}
\item
{\bf (ii)}
Its cohomology is isomorphic, as a graded vector space,
to the space
$\big( \mathcal{W}^{{}^{\mathrm{pb}}}_{N_{\mathrm{f}}}\big)^\bullet$
\eqref{SpaceOfHorizontalWeightSystems}
of weight systems (Def. \ref{HorizontalWeightSystems})
on horizontal chord diagrams with $N_{\mathrm{f}}$ strands:
\begin{equation}
  \mbox{
    \tiny
    \color{blue}
    \begin{tabular}{c}
      Cohomology group of
      \\
      based loop space of
      \\
      configuration space of
      \\
      $N_{\mathrm{f}}$ ordered points in $\mathbb{R}^3$
    \end{tabular}
  }
  \phantom{AA}
  H^\bullet
  \Big(
    \Omega
    \underset{
      {}^{\{1,\cdots,N_{\mathrm{f}}\}}
    }{
      \mathrm{Conf}
    }
    \!
    (\mathbb{R}^3)
  \Big)
  \;\;\;\simeq\;\;\;
  \big(
    \mathcal{W}^{{}^{\mathrm{pb}}}_{N_{\mathrm{f}}}
  \big)^\bullet
  \phantom{AA}
  \mbox{
    \tiny
    \color{blue}
    \begin{tabular}{c}
      Weight systems on
      \\
      horizontal chord diagrams
      \\
      with $N_{\mathrm{f}}$ strands
    \end{tabular}
  }
\end{equation}

\vspace{-2mm}
\item
{\bf (iii)} Hence the higher co-observables \eqref{HigherObservablesOnD6D8Intersections}
are identified with
horizontal chord diagrams of any number of strands
\begin{equation}
  \label{CoObservablesAreChordDiagrams}
  \underset{
    \mbox{
      \tiny
      \color{blue}
      \begin{tabular}{c}
 %       $\phantom{a}$
 %       \\
        Higher order co-observables on
        \\
        $\mathrm{D}6\perp \mathrm{D}8$-brane configurations
      \end{tabular}
    }
  }{
    \mathrm{Obs}_\bullet^{\mathrm{D}6 \perp \mathrm{D}8}
  }
    \;\;\simeq\;\;
  \underset{
    \mbox{
      \tiny
      \color{blue}
      \begin{tabular}{c}
%        $\phantom{a}$
 %       \\
        Horizontal chord diagrams
        \\
        with any number of strands
      \end{tabular}
    }
  }{
    \big(
      \mathcal{A}^{{}^{\mathrm{pb}}}
    \big)_\bullet
  }
  \;:=\;
  \underset{
    N_{\mathrm{f}}
    \in
    \mathbb{N}
  }{\medoplus}
  \big(
    \mathcal{A}^{{}^{\mathrm{pb}}}_{N_{\mathrm{f}}}
  \big)_\bullet
\end{equation}
\vspace{-2mm}
and the higher observables \eqref{HigherObservablesOnD6D8Intersections}
with the weight systems on these:
\begin{equation}
  \label{ObservablesAreWeightSystems}
  \underset{
    \mbox{
      \tiny
      \color{blue}
      \begin{tabular}{c}
        $\phantom{a}$
        \\
        Higher order observables on
        \\
        $\mathrm{D}6\perp \mathrm{D}8$-brane configurations
      \end{tabular}
    }
  }{
    \mathrm{Obs}^\bullet_{\mathrm{D}6 \perp \mathrm{D}8}
  }
    \;\;\simeq\;\;
  \underset{
    \mbox{
      \tiny
      \color{blue}
      \begin{tabular}{c}
   %     $\phantom{a}$
    %    \\
        Weight systems on
        \\
        horizontal chord diagrams
        \\
        with any number of strands
      \end{tabular}
    }
  }{
    \big(
      \mathcal{W}^{{}^{\mathrm{pb}}}
    \big)^\bullet
  }
  \;:=\;
  \underset{
    N_{\mathrm{f}}
    \in
    \mathbb{N}
  }{\medoplus}
  \big(
    \mathcal{W}^{{}^{\mathrm{pb}}}_{N_{\mathrm{f}}}
  \big)^\bullet
\end{equation}

\end{prop}
\begin{proof}
  By
  \cite[Thm. 2.2]{FadellHusseini01}
  (also \cite[Thm. 4.1]{CohenGitler01}\cite[Thm. 2.3]{CohenGitler02})
  we have an isomorphism
  \begin{equation}
    \label{HomologyRingIsEnvelopingOfInfinitesimalBraids}
    \mbox{
      \tiny
      \color{blue}
      \begin{tabular}{c}
        Homology ring of
        \\
        based loop space of
        \\
        configuration space of points
      \end{tabular}
    }
    \phantom{AA}
    H_\bullet
    \Big(
      \Omega
      \underset{
        {}^{\{1,\cdots,N_{\mathrm{f}}\}}
      }{
        \mathrm{Conf}
      }
      \!
      (\mathbb{R}^3)
    \Big)
    \;\;
    \simeq
    \;\;
    \mathcal{U}
    \big(
      \mathcal{L}_{N_\mathrm{f}}(1)
    \big)
    \phantom{AA}
    \mbox{
      \tiny
      \color{blue}
      \begin{tabular}{c}
        Universal enveloping algebra of
        \\
        infinitesimal braid Lie algebra
        \\
        generated in degree 1
      \end{tabular}
    }
  \end{equation}
  identifying the homology ring of the looped
  configuration space with the universal enveloping
  algebra of the \emph{infinitesimal braid Lie algebra}
  on $N_{\mathrm{f}}$ strands with generators in degree 1,
  hence of the Lie algebra freely defined by the
  infinitesimal braid relations \eqref{InfinitesimalBraidRelations}.
  Using that these relations are equivalently the
  2T-relations \eqref{2TRelationsOnHorizontalChordDiagrams}
  and 4T-relations \eqref{4TRelationOnHorizontalChordDiagrams}
  on horizontal chord diagrams, direct inspection reveals that
  this universal enveloping algebra
  is canonically isomorphic, as a graded associative algebra,
  to the concatenation algebra
  of horizontal chord diagrams \eqref{HorizontalChordDiagramsModulo2TAnd4TRelations}:
  \begin{equation}
    \label{AlgebraOfHorizontalChordDiagramsIsUniversalEnvelopingOfInfinitesimalBraids}
    \mbox{
      \tiny
      \color{blue}
      \begin{tabular}{c}
        Universal enveloping algebra of
        \\
        infinitesimal braid Lie algebra
      \end{tabular}
    }
    \phantom{AA}
    \mathcal{U}
    \big(
      \mathcal{L}_{N_\mathrm{f}}(1)
    \big)
    \;\;
    \simeq
    \;\;
    \mathcal{A}^{{}^{\mathrm{pb}}}_{N_{\mathrm{f}}}
    \phantom{AA}
    \mbox{
      \tiny
      \color{blue}
      \begin{tabular}{c}
        Concatenation algebra of
        \\
        horizontal chord diagrams
      \end{tabular}
    }
  \end{equation}

  \vspace{1mm}
  \noindent The combination of \eqref{HomologyRingIsEnvelopingOfInfinitesimalBraids}
  with \eqref{AlgebraOfHorizontalChordDiagramsIsUniversalEnvelopingOfInfinitesimalBraids}
  yields the first statement.  The second statement then follows by direct dualization,
  using the universal coefficient theorem --
  see also the statement of \cite[Thm. 4.1]{Kohno02}.
    With this, the third statement follows by Prop. \ref{HigherObservablesEquivalentToCohomologyOfLoopedConfigurationSpace}.
\hspace{12cm} \end{proof}

\begin{remark}[Quantum algebra structure on higher co-observables]
  \label{QuantumAlgebraStructure}
 {\bf (i)}
  The product operation on the homological Hopf algebras
  $H_\bullet
    \big(
      \Omega
      \underset{
        {}^{\{1,\cdots, N_{\mathrm{f}}\}}
      }{\mathrm{Conf}}
      (\mathbb{R}^3)
    \big) \simeq \mathcal{A}^{{}^{\mathrm{pb}}}_{N_{\mathrm{f}}}$
  in Prop. \ref{HigherObservablesOnIntersectingBranesAreWeightSystems}
  is non-commutative (manifestly so from \eqref{HorizontalChordDiagramsModulo2TAnd4TRelations})
  while its co-product is graded co-commutative (as it comes from the
  diagonal map on the space
  $\Omega
      \underset{
        {}^{\{1,\cdots, N_{\mathrm{f}}\}}
      }{\mathrm{Conf}}
      (\mathbb{R}^3)$.

\vspace{-2mm}
\item  {\bf (ii)} Accordingly, for the cohomological Hopf algebras
  $H^\bullet
    \big(
      \Omega
      \underset{
        {}^{\{1,\cdots, N_{\mathrm{f}}\}}
      }{\mathrm{Conf}}
      (\mathbb{R}^3)
    \big) \simeq \mathcal{W}^{{}^{\mathrm{pb}}}_{N_{\mathrm{f}}}$
  in Prop. \ref{HigherObservablesOnIntersectingBranesAreWeightSystems}
  it is the other way around: Here the product operation is
  graded-commutative (being the cup product on cohomology).

\vspace{-2mm}
\item  {\bf (iii)} In this sense,  when regarded as graded algebras of (co-)observables,
  weight systems $\mathcal{W}$ form an algebra of
  \emph{classical observables},
  while chord diagrams $\mathcal{A}$ form an algebra of
  \emph{quantum observables}.
\end{remark}

\newpage

%%%%%%%%%%%%%%%%%%%%%%%%%%%%%%%%%%%%%%%%%%%%%%%
\section{Weight systems on chord diagrams}
\label{WeightSystemsOnChordDiagrams}
%%%%%%%%%%%%%%%%%%%%%%%%%%%%%%%%%%%%%%%%%%%%%%%%

Here we lay out the key definitions and facts regarding
weight systems on chord diagrams, streamlined towards
our applications in \cref{WeightSystemsAsObservablesOnIntersectingBranes}.
For round chord diagrams we follow \cite{BarNatan95},
which has made it into textbook literature
\cite[4-6]{CDM11}\cite[11,13-14]{JacksonMoffat19}.
For weight systems on horizontal chord diagrams,
which we find to be of deeper relevance
(see Prop. \ref{FundamentalTheremOfWeightSystems} and
its interpretations in \cref{HorizontalChordDiagramsEncodesStringTopologyOperations},
\cref{MatrixModelObservables}, and
\cref{HananyWittenTheory})
we follow \cite{BarNatan96}, which seems not to have
found as much attention yet.

%%%%%%%%%%%%%%%%%%%%%%%%%%%%%%%%%%%%%%%%%%%%%
\subsection{Horizontal chord diagrams}
%%%%%%%%%%%%%%%%%%%%%%%%%%%%%%%%%%%%%%%%%%%%%

\begin{equation}
\label{HorizontalChordDiagrams}
\hspace{-.6cm}
\mbox{

}
\end{equation}

\vspace{.2cm}

\noindent The linear span of the set \eqref{HorizontalChordDiagrams}
of horizontal chord diagrams is
canonically a graded associative algebra under concatenation $\circ$ of strands:

\vspace{-.7cm}

\begin{equation}
  \label{GradedAlgebraOfHorizontalChordDiagrams}
\underset{
  \mathclap{
  \mbox{
    \tiny
    \color{blue}
    %
}}
\!\!
\right]
\end{equation}

\noindent Hence if, for any $i < j \in \{1, \cdots, N_{\mathrm{f}}\}$,  we write
\begin{equation}
  \label{GeneratorsOfHorizontalChordDiagrams}
  \underset{
    \mathclap{
    \mbox{
      \tiny
      \color{blue}
      %
}
\right]
\;\;\; \in \;\;
\mathrm{Span}\big( \mathcal{D}^{{}^{\mathrm{pb}}}_{N_{\mathrm{f}}} \big)_1
\end{equation}
for the horizontal chord diagram
%on $N_{\mathrm{f}}$ strands
with exactly one chord, which goes between the $i$th and the $j$th strand,
then the algebra of horizontal chord diagrams is just the
\emph{free associative algebra}
on these generators $t_{i j}$ of degree 1.

\newpage

\noindent On this free algebra consider the following relations:

\noindent {\bf (i)} the {\it 2T relations}:

\vspace{-2mm}
\begin{equation}
\label{2TRelationsOnHorizontalChordDiagrams}
\scalebox{.9}{

}
\end{equation}

\vspace{-1mm}
\noindent Expressed in terms of the algebra generators \eqref{GeneratorsOfHorizontalChordDiagrams}, these are
equivalently the \emph{infinitesimal braid relations}
\cite[(1.1.4)]{Kohno87}:
\begin{equation}
  \label{InfinitesimalBraidRelations}
  \left.
  \begin{array}{ccl}
  (\mathrm{2T}) & \big[t_{i j},\, t_{j k}\big] & = 0
  \\
  (\mathrm{4T}) &  \big[t_{i k} + t_{j k} ,\, t_{i j}\big] & = 0
  \end{array}
  \;\;
  \right\}
  \;\;\;\;
  \mbox{
    for all pairwise distinct
    $i,j,k,l \in \{1, \cdots, N_{\mathrm{f}}\}$
  }.
\end{equation}

\vspace{1mm}
\noindent Now, the quotient of the graded algebra \eqref{GradedAlgebraOfHorizontalChordDiagrams}
of linear combinations of horizontal chord diagrams
by these relations \eqref{InfinitesimalBraidRelations}
is a graded associative algebra
denoted
\begin{equation}
  \label{HorizontalChordDiagramsModulo2TAnd4TRelations}
  \begin{aligned}
  \mathcal{A}^{{}^{\mathrm{pb}}}_{N_{\mathrm{f}}}
  & :=\;
  \mathrm{Span}
  \big(
    \mathcal{D}^{{}^{\mathrm{pb}}}_{N_{\mathrm{f}}}
  \big)
  \big/
  (\mathrm{2T}, \mathrm{4T})
  \\
  &=\;
  \mathrm{GradedAssoc}
  \Big(
    \big\{
      \underset{
        \mbox{\tiny $\mathrm{deg} = 1$ }
      }{t_{i j} = - t_{j i}}
      \vert i < j \in \{1,\cdots N_{\mathrm{f}}\}\big\}
  \Big)
  \big/ (\mathrm{2T}, \mathrm{4T})\,.
  \end{aligned}
  \end{equation}
Hence:
$$
  \hspace{0cm}
  \mathcal{A}^{{}^{\mathrm{pb}}}_{N_{\mathrm{f}}}
 \! := \!
  \mathrm{Span}
  \! \!\left(\!\!\!
    \overset{
      \mbox{
       \tiny
       \color{blue}
       Horizontal chord diagrams
      }
    }{
      \left\{
        \mbox{
          \raisebox{-55pt}{
          \scalebox{.5}{

    \end{tabular}
       }
      }
  \!\!\!\!\!\! \right)
$$

\vspace{2mm}
\begin{defn}
\label{HorizontalWeightSystems}
A \emph{weight system on horizontal chord diagrams} is a linear function
on the span of horizontal chord diagrams \eqref{GradedAlgebraOfHorizontalChordDiagrams}
 modulo 2T- and 4T-relations
\eqref{HorizontalChordDiagramsModulo2TAnd4TRelations}.
Hence the space of all
weight systems is the graded linear dual space to
the quotient space \eqref{HorizontalChordDiagramsModulo2TAnd4TRelations},
to be denoted

\vspace{0cm}

\begin{equation}
  \label{SpaceOfHorizontalWeightSystems}
  \hspace{1cm}
  \overset{
    \mathclap{
    \mbox{
      \tiny
      \color{blue}
      \begin{tabular}{c}
                Space of
        weight systems
        \\
        on horizontal
        chord diagrams
        \\
        with $N_{\mathrm{f}}$ strands
        \\
        $\phantom{a}$
      \end{tabular}
    }
    }
  }{
    \big(
      \mathcal{W}^{{}^{\mathrm{pb}}}_{N_{\mathrm{f}}}
    \big)^\bullet
  }
  \;\;\;\;\;\;\;\; :=\;\;\;\;\;\;\;
  \overset{
    \mathclap{
    \mbox{
       \tiny
       \color{blue}
       \begin{tabular}{c}
         Graded linear dual to
         span of
         \\
         horizontal chord diagrams
         \\
         modulo 2T- and 4T relations
         \\
         $\phantom{a}$
       \end{tabular}
    }
    }
  }{
    \big(
     \big(
       \mathcal{A}^{{}^{\mathrm{pb}}}_{N_{\mathrm{f}}}
     \big)_\bullet
    \big)^\ast
  }.
\end{equation}
\end{defn}

\medskip

%%%%%%%%%%%%%%%%%%%%%%%%%%%%%%%%%%%%%%%%%%
\subsection{Round chord diagrams}
%%%%%%%%%%%%%%%%%%%%%%%%%%%%%%%%%%%%%%%%%%

\noindent {\bf Closing up horizontal chord diagrams.}
Given any permutation $\sigma \in \mathrm{Sym}(N_{\mathrm{f}})$
of $N_{\mathrm{f}}$ elements, there is an evident
way to close a horizontal chord diagram \eqref{HorizontalChordDiagrams}
to a round chord diagram. For example, for $\sigma = (312)$
a cyclic permutation of three elements, we have:

\vspace{-.5cm}

\begin{equation}
\label{AHorizontalChordDiagramTraced}
\raisebox{-90pt}{

}
\end{equation}

\vspace{-.0cm}

\noindent An evident generalization of round chord
diagrams, needed below, is obtained by allowing internal vertices:

\vspace{-.4cm}

\begin{equation}
\label{JacobiDiagrams}
\mbox{
\hspace{-1cm}
%
}
\;\;\;\,
,
\;
\scalebox{1.5}{$\cdots$}
\right\}
}
$
\end{tabular}
}
\end{equation}

\vspace{-.1cm}

The closing operation as in \eqref{AHorizontalChordDiagramTraced}
on the set
of horizontal chord diagrams \eqref{HorizontalChordDiagrams},
together with the understanding of
round chord diagrams \eqref{RoundChordDiagrams}
as special cases
of Jacobi diagrams \eqref{JacobiDiagrams} gives functions of
sets of our three types of diagrams, as follows:

\vspace{-.4cm}

\begin{equation}
  \label{FunctionsOfSetsOfDiagrams}
  \xymatrix{
    \underset{
      \mathclap{
      \mbox{
        \tiny
        \color{blue}
\begin{tabular}{c}
          $\phantom{a}$
          \\
          Set of
          \\
          horizontal chord diagrams
        \end{tabular}
      }
      }
    }{
      \mathcal{D}^{{}^{\mathrm{pb}}}_{N_{\mathrm{f}}}
    }
   \;\; \ar@{->>}[rr]^-{ \mathrm{close}_{{}_{(N_{\mathrm{f}}12\cdots)}} }
    && \;\;\;
    \underset{
      \mathclap{
      \mbox{
        \tiny
        \color{blue}
        \begin{tabular}{c}
          $\phantom{a}$
          \\
          Set of
          \\
          round chord diagrams
        \end{tabular}
      }
      }
    }{
      \mathcal{D}^{{}^{\mathrm{c}}}
    }
    \;\;\; \ar@{^{(}->}[rr]^-{ i }
    && \;\;\;
    \underset{
      \mathclap{
      \mbox{
        \tiny
        \color{blue}
        \begin{tabular}{c}
          $\phantom{a}$
          \\
          Set of
          \\
          Jacobi diagrams
        \end{tabular}
      }
      }
    }{
      \mathcal{D}^{{}^{\mathrm{t}}}
    }.
  }
\end{equation}

\vspace{-.3cm}

\noindent {\bf Round closure of the 4T relations.}
Under the closing map \eqref{FunctionsOfSetsOfDiagrams}, the four types of horizontal
chord diagrams \eqref{HorizontalChordDiagrams} that appear in the horizontal 4T
relation \eqref{4TRelationOnHorizontalChordDiagrams} give the following four types
of round chord diagrams \eqref{RoundChordDiagrams}:

\vspace{4mm}
\hspace{-.3cm}
\scalebox{.85}{
%
}

\vspace{3mm}
\noindent This means that in order to make the closing
operation on the left of \eqref{FunctionsOfSetsOfDiagrams}
pass to the quotient space $\mathcal{A}^{{}^{\mathrm{pb}}}$
(see \eqref{HorizontalChordDiagramsModulo2TAnd4TRelations}),
we have to quotient the span of round chord diagrams
by the following {\it round 4T relations} on $\mathrm{Span}\big( \mathcal{D}^c\big)$:

\vspace{-.2cm}

\begin{equation}
\label{Round4TRelations}
\hspace{-.5cm}
\scalebox{.85}{
$
\left[
\!\!\!\!\!\!
\mbox{
\raisebox{-50pt}{
%
}
}
\!
\right]
$}
\end{equation}

\noindent Hence, in direct analogy to
Def. \ref{HorizontalChordDiagramsModulo2TAnd4TRelations}, we have:
\begin{defn}
  \label{WeightSystemsOnRoundChordDiagrams}
  Write

  \vspace{-.9cm}

  $$
    \mathcal{A}^{{}^{\mathrm{c}}}
    \;:=\;
    \mathrm{Span}\big(\mathcal{D}^{{}^{\mathrm{c}}}\big)
    /
    (\mathrm{4T})
  $$
  for the graded quotient vector space of the span of round chord diagrams \eqref{RoundChordDiagrams}
  by the round 4T relations \eqref{Round4TRelations}
  (see first line of \eqref{RoundChordDiagramsModulo4TIsJavcobidDiagramsModuloSTU} below).
  A \emph{weight system on round chord diagrams} \footnote{Beware that some authors
  call these \emph{framed} weight systems, since we do not impose the 1T relation.}
  is a
  linear function on this space:

\vspace{-.1cm}

\begin{equation}
  \label{SpaceOfRoundWeightSystems}
  \underset{
    \mathclap{
    \mbox{
      \tiny
      \color{blue}
      \begin{tabular}{c}
        $\phantom{a}$
        \\
        Space of
        weight systems
        \\
        on round
        chord diagrams
      \end{tabular}
    }
    }
  }{
    \big(
      \mathcal{W}^{{}^{\mathrm{c}}}
    \big)^\bullet
  }
\;\;\;  \;\;\;\;\;\;:=\;\;\;\;\;\;\;\;\;
  \underset{
    \mathclap{
    \mbox{
       \tiny
       \color{blue}
       \begin{tabular}{c}
         Graded linear dual to
         span of
         \\
         round chord diagrams
         \\
         modulo round 4T relations
         \\
         $\phantom{a}$
       \end{tabular}
    }
    }
  }{
    \big(
      \big(
        \mathcal{A}^{{}^{\mathrm{c}}}
      \big)_\bullet
    \big)^\ast
  }
\end{equation}

\vspace{-.5cm}

\end{defn}

\vspace{0cm}

\noindent {\bf Resolution of round 4T- to STU-relations.}
We would like that also the injection $i$
of round chord diagrams into Jacobi diagrams,
on the right of \eqref{FunctionsOfSetsOfDiagrams}, to
pass to these quotients. For that we consider, moreover,
the following {relations on the linear span of Jacobi diagrams,
called the {\it STU relations} on
$\mathrm{Span}\big( \mathcal{D}^{{}^{\mathrm{t}}}\big)$:

\begin{equation}
\label{STURelations}
\left[
\mbox{
\raisebox{-30pt}{

}
}
\!
\right]
\end{equation}

\vspace{3mm}
\noindent The \emph{reason} behind these STU-relations is that
on Jacobi diagrams they \emph{resolve} the round 4T relations
\eqref{Round4TRelations}:

\newpage

$\,$

\vspace{-1.1cm}

\begin{equation}
\hspace{-.3cm}
\scalebox{.7}{
%
}
\end{equation}

\noindent {\bf Equivalence of chord diagrams and Jacobi diagrams.}
Using this factorization of the round 4T relation by the
STU-relations, one proves that
the linear span of round chord diagrams modulo the round 4T-relations
is equivalent to that of Jacobi diagrams modulo the
STU-relations (due to \cite[Thm. 6]{BarNatan95}; see \cite[5.3]{CDM11}):

\vspace{-.5cm}

\begin{equation}
  \label{RoundChordDiagramsModulo4TIsJavcobidDiagramsModuloSTU}
  \scalebox{.95}{
  $
  \hspace{-.3cm}
  %
  $
  }
\end{equation}

\vspace{-.5cm}

\noindent Hence:

\vspace{-.2cm}

\begin{prop}[Relating weight systems]
\label{RelationOfWeightSystems}
The maps \eqref{FunctionsOfSetsOfDiagrams} of sets of chord diagrams
dualize to a linear bijection of weight systems on Jacobi diagrams
(i.e., the graded linear dual of $\mathcal{A}^{{}^{\mathrm{t}}}$)
with weight systems on round chord diagrams \eqref{SpaceOfRoundWeightSystems},
followed by a linear injection of the latter into the
space of weight systems on horizontal chord diagrams \eqref{SpaceOfHorizontalWeightSystems}:
\
\vspace{-.7cm}

\begin{equation}
  \label{SequenceOfMapsOfWeightSystems}
  \xymatrix@C=4em{
    \underset{
      \mathclap{
      \mbox{
        \tiny
        \color{blue}
        \begin{tabular}{c}
          $\phantom{a}$
          \\
          Weight systems on
          \\
          horizontal chord diagrams
        \end{tabular}
      }
      }
    }{
      \mathcal{W}^{{}^{\mathrm{pb}}}_{N_{\mathrm{f}}}
    }
   \;\;\;\;\;  \ar@{<-^{)}}[rr]^-{
      (
        \mathrm{close}_{(N_{\mathrm{f}}12\cdots)}
      )^\ast
    }_-{
      \mbox{
        \tiny
        \color{blue}
        \begin{tabular}{c}
          injection
        \end{tabular}
      }
    }
    && \;\;\;\;\;
    \underset{
      \mathclap{
      \mbox{
        \tiny
        \color{blue}
        \begin{tabular}{c}
          $\phantom{a}$
          \\
          Weight systems on
          \\
          round chord diagrams
        \end{tabular}
      }
      }
    }{
      \mathcal{W}^{{}^{\mathrm{c}}}
    }
    \ar@{<-}[rr]^-{ i^\ast }_-{\simeq
      \mbox{
        \tiny
        \color{blue}
        \begin{tabular}{c}
          bijection
        \end{tabular}
      }
    }
    && \;\;\;
    \underset{
      \mathclap{
      \mbox{
        \tiny
        \color{blue}
        \begin{tabular}{c}
          $\phantom{a}$
          \\
          Weight systems on
          \\
          Jacobi diagrams
        \end{tabular}
      }
      }
    }{
      \mathcal{W}^{{}^{\mathrm{t}}}
    }
  }.
\end{equation}
\end{prop}

\newpage

%%%%%%%%%%%%%%%%%%%%%%%%%%%%%%%%%%%%%%%%%%%%%%
\subsection{Lie algebra weight systems}
\label{LieAlgebraWeightSystems}
%%%%%%%%%%%%%%%%%%%%%%%%%%%%%%%%%%%%%%%%%%%%%%

\noindent {\bf Metric Lie algebras appear.}
The equivalence \eqref{RoundChordDiagramsModulo4TIsJavcobidDiagramsModuloSTU}
reveals that weight systems secretly encode Lie theoretic data.
Indeed, the STU-relation \eqref{STURelations}
is manifestly the Jacobi identity, or more generally the
Lie action property. This is expressed in \emph{Penrose diagram notation}
(reviewed in \cite[appendix, p. 424-434]{PenroseRindler84})
also called \emph{string diagram calculus} (reviewed in \cite{Selinger09});
see the big table on page \pageref{StringDiagramTable} for the translation. Diagrammatically:
\begin{equation}
\label{ActionProperty}
  \mbox{
  \raisebox{-40pt}{

  }
$$

$$
  \;
  \rho(f(x,y),z)
  \;\;\;\;\;\;\;\;\;\;\;\;\;
  =
  \;\;\;\;\;\;\;\;\;\;\;\;\;
  \rho(y,\rho(x,z))
  \;\;\;\;\;\;\;\;\;\;\;\;\;\;
  -
  \;\;\;\;\;\;\;\;\;\;\;\;\;\;
  \rho(x,\rho(y,z))
\end{equation}
\noindent
With $f(x,y) = [x,y]$ a Lie bracket, this is the Lie action property on $\rho$.
Moreover, with $\rho(x,z) = [x,z]$ the adjoint action, this is
the Jacobi identity.

\medskip

This means that metric Lie representations of metric Lie algebras internal to
tensor categories induce weight systems (Def. \ref{HorizontalWeightSystems})
on chord diagrams.  For ordinary Lie algebras this is due to \cite[Sec. 2.4]{BarNatan95}, while
the general statement is made explicit in \cite[Sec. 3]{RobertsWillerton06},
following observations in \cite{Vaintrob94}\cite{Vogel11}. We capture this as:

\vspace{-2mm}
\begin{equation}
\label{LieAlgebraWeightSystem}
\mbox{

   %}
   (1.2,1);

% \begin{scope}[shift={(.6,1.15)}]
% \begin{scope}[scale=(.6)]
% \draw[fill=white]
%     (-.2,-.2)
%     rectangle node{\tiny $g$}
%     (.2,.2);
% \end{scope}
% \end{scope}

% \begin{scope}[shift={(.6+.3,1.15-1+.05)}]
% \begin{scope}[scale=(.6)]
% \draw[fill=white]
%     (-.2,-.2)
%     rectangle node{\tiny $g$}
%     (.2,.2);
% \end{scope}
% \end{scope}

 \draw[ultra thick, darkblue] (0,2) to (0,-1);
 \draw[ultra thick, darkblue] (1.2,2) to (1.2,-1);
 \draw[ultra thick, darkblue] (2.4,2) to (2.4,-1);

 \draw (-.6,1.5) node {$\cdots$};
 \draw (0.6,1.5) node {$\cdots$};
 \draw (1.2+0.6,1.5) node {$\cdots$};
 \draw (2.4+0.6,1.5) node {$\cdots$};

 \draw (-.6,-.5) node {$\cdots$};
 \draw (0.6,-.5) node {$\cdots$};
 \draw (1.2+0.6,-.5) node {$\cdots$};
 \draw (2.4+0.6,-.5) node {$\cdots$};

 \draw (0,-1.4) node {$V$};
 \draw (1.2,-1.4) node {$V$};
 \draw (2.4,-1.4) node {$V$};

 \draw (0,-1.4+3.8) node {$V$};
 \draw (1.2,-1.4+3.8) node {$V$};
 \draw (2.4,-1.4+3.8) node {$V$};

 \begin{scope}[shift={(0,0)}]
   \draw[draw=orangeii, fill=orangeii] (0,0) circle (.07);
   \draw[draw=darkblue, thick, fill=white]
     (-.15,-.13)
     rectangle node{\tiny $\rho$}
     (.15,.17);
 \end{scope}
 \begin{scope}[shift={(2.4,0)}]
   \draw[draw=orangeii, fill=orangeii] (0,0) circle (.07);
   \draw[draw=darkblue, thick, fill=white]
     (-.15,-.13)
     rectangle node{\tiny $\rho$}
     (.15,.17);
 \end{scope}

 \begin{scope}[shift={(0,1)}]
   \draw[draw=orangeii, fill=orangeii] (0,0) circle (.07);
   \draw[draw=darkblue, thick, fill=white]
     (-.15,-.13)
     rectangle node{\tiny $\rho$}
     (.15,.17);
 \end{scope}
 \begin{scope}[shift={(1.2,1)}]
   \draw[draw=orangeii, fill=orangeii] (0,0) circle (.07);
   \draw[draw=darkblue, thick, fill=white]
     (-.15,-.13)
     rectangle node{\tiny $\rho$}
     (.15,.17);
 \end{scope}

\end{tikzpicture}
}
\\
$\phantom{a}$
\\
{\tiny \color{blue} Chord/Jacobi diagram}
&
\hspace{.4cm}{\tiny \color{blue} evaluates to}
&
\hspace{-1.1cm}
{\tiny \color{blue} element of ground field $k = \mathrm{End}(\mathbf{1})$:}
\\
\\
\raisebox{-60pt}{

  }
  &
  \scalebox{1}{
    $
      \rho_a{}^l{}_i \, k_{l j}
      =
      \rho_a{}^l{}_j \, k_{l i}
    $
  }
  \\
  \hline
\end{tabular}
}

\newpage

\noindent {\bf Metric super Lie algebras appear.}
The relevance of tensor categories in \eqref{LieAlgebraWeightSystem}
more general than that of plain vector spaces, is that by considering the tensor
category of super vector spaces (e.g., \cite[3.1]{Varadarajan04}), it immediately
follows that metric representations of \emph{super Lie algebras} \cite{Kac77}
or rather of \emph{metric super Lie algebras} (as in \cite[3.3]{dMFMR09})
are a source of weight systems on chord diagrams
\cite{Vaintrob94}\cite{FKV97}\cite{Vogel11}; see \cite[6.4]{CDM11}.
Moreover, we observe that Deligne's theorem
\cite{Deligne02} (see \cite{Ostrik04})
says that \emph{all} reasonable tensor categories
(satisfying just a mild set-theoretic size bound) are
representation categories of algebraic super-groups,
whence all reasonable Lie algebra weight systems
on chord diagrams are induced by metric super Lie algebras,
in general equivariant with respect to some super
symmetry group.
This means that the theory of weight systems on chord diagrams
largely overlaps with that of metric representations
of metric super Lie algebras.  However, interestingly, weight systems
see even one further datum, as we describe next.

%%%%%%%%%%%%%%%%%%%%%%%%%%%%%%%%%%%%%%%%%%%%%%%%
\subsection{On stacks of coincident strands}
\label{OnStacksOfCoincidentStrands}
%%%%%%%%%%%%%%%%%%%%%%%%%%%%%%%%%%%%%%%%%%%%%%%%

\noindent {\bf Stacks of coincident strands.}
We now consider horizontal chord diagrams
$D \in \mathcal{D}^{{}^{\mathrm{pb}}}$
that superficially have $N_{\mathrm{f}}$ strands
as in \eqref{HorizontalChordDiagrams},
but where, on closer inspection, the $i$th strand is seen/resolved to
consist of a stack of $N_{\mathrm{c},i}$ ``coincident strands'',
for some tuple of natural numbers:

\vspace{-.4cm}

\begin{equation}
  \label{Tuples}
  \vec N_{\mathrm{c}}
    =
  (N_{\mathrm{c},1}, \cdots, N_{\mathrm{c},N_{\mathrm{f}}} )
  \;\in\;
  \mathbb{N}^{N_{\mathrm{f}}}
  \;\;\; \mbox{with}\;\;\;
  N_{\mathrm{c}}
  :=
  \underoverset{i = 1}{N_{\mathrm{f}}}{\sum} N_{\mathrm{c},i}\;.
\end{equation}
The following operation $\Delta$ (see \cite[2.2]{BarNatan96})
may be seen to make this idea precise:

\begin{equation}
  \label{Delta}
  \hspace{-3.5cm}
  \xymatrix@R=-10pt@C=2.5em{
    \mathcal{D}^{{}^{\mathrm{pb}}}
    \times
    \big(
      \underset{\mathbb{N}}{\oplus} \mathbb{N}
    \big)
    \ar[rrrr]^-{ \Delta }
    &&&& \;\;
    \mathcal{A}^{{}^{\mathrm{pb}}}
  \\
  \big( D, \; \vec N_{\mathrm{c}}  \big)
    \ar@{}[rrrr]|<<<<<{\longmapsto}
  &&&&
  \mathclap{
  \underset{
    \mathclap{
      \mbox{
        \tiny
        \color{blue}
        \begin{tabular}{c}
          $\phantom{a}$
          \\
          Chord diagram
          like $D$
          \\
          but with stacks of
          \\
          $\vec N_{\mathrm{c}}$ coincident strands
        \end{tabular}
      }
    }
  }{
    \Delta^{\vec N_{\mathrm{c}}}
    (D)
  }
  \;:=\;
  \left[
   \!\!
  \mbox{
    \small
    \begin{tabular}{l}
      Sum of
      horizontal chord diagrams
      \\
      with $N_{\mathrm{c}} = N_{\mathrm{c}_1} + \cdots + N_{\mathrm{c},N_{\mathrm{f}}}$ strands
      \\
      whose chords are the chords $t_{i j}$ of $D$
      \\
      but re-attached in all
      $N_{\mathrm{c},i}\cdot N_{\mathrm{c},j}$ ways
      \\
      to the $i$th and the $j$th stack of chords
    \end{tabular}
  }
  \!\!
  \right]
  }
  }
\end{equation}

\vspace{-.1cm}

\begin{equation}
  \label{DirectSumcalA}
  \mathllap{
    \mbox{where now}
    \phantom{AAAAAAAA\!}
  }
  \mathcal{A}^{{}^{\mathrm{pb}}}
  \;:=\;
  \underset
  {
    N_{\mathrm{f}}
    \in
    \mathbb{N}
  }
  {\medoplus}
  \mathcal{A}^{{}^{\mathrm{pb}}}_{N_{\mathrm{f}}}
  ,\;\;\;\;\;\;\;\;\;
  \xymatrix{
    \mathcal{A}^{{}^{\mathrm{pb}}}
    \ar@{->>}[rr]^-{ p_{{}_{N_{\mathrm{f}}}} }
    && \;\;
    \mathcal{A}^{{}^{\mathrm{pb}}}_{N_{\mathrm{f}}}
   \;\; \ar@{^{(}->}[rr]^-{ i_{{}_{N_{\mathrm{f}}}} }
    && \;\;
    \mathcal{A}^{{}^{\mathrm{pb}}}
  }
\end{equation}

\vspace{-.1cm}

\noindent
denotes the direct sum of all spaces of horizontal chord diagrams
\eqref{HorizontalChordDiagramsModulo2TAnd4TRelations}
over the number $N_{\mathrm{f}}$ of strands.
$$
\begin{array}{rccl}
  \mathllap{
    \mbox{For example:}
    \phantom{AAAAA}
  }
  \Delta^{(2,2)}
  \left(
      \raisebox{-11pt}{
        \scalebox{.8}{
 \begin{tikzpicture}

 \begin{scope}
   \clip (-1,-.1) rectangle (-1+.1,+.1);
   \draw[draw=orangeii, fill=orangeii] (-1,0) circle (.07);
 \end{scope}

 \begin{scope}
   \clip (+1,+.1) rectangle (+1-.1,-.1);
   \draw[draw=orangeii, fill=orangeii] (+1,0) circle (.07);
 \end{scope}

 \draw[thick, orangeii] (-1,0) to (1,0);

 \draw[darkblue, ultra thick] (-1,.6) to (-1,-.6);

 \draw[darkblue, ultra thick] (+1,.6) to (+1,-.6);

 \end{tikzpicture}
        }
      }
  \right)
  & = &  &
  \left[
    \raisebox{-11pt}{
      \scalebox{.8}{
 \begin{tikzpicture}

 \begin{scope}
   \clip (-1,-.1) rectangle (-1+.1,+.1);
   \draw[draw=orangeii, fill=orangeii] (-1,0) circle (.07);
 \end{scope}

 \begin{scope}
   \clip (+1,+.1) rectangle (+1-.1,-.1);
   \draw[draw=orangeii, fill=orangeii] (+1,0) circle (.07);
 \end{scope}

 \draw[thick, orangeii] (-1,0) to (1,0);
 \draw[thick, white] (-1+.3-.05,0) to (-1+.3+.05,0);

 \draw[darkblue, ultra thick] (-1,.6) to (-1,-.6);
 \draw[darkblue, ultra thick] (-1+.3,.6) to (-1+.3,-.6);

 \draw[darkblue, ultra thick] (+1,.6) to (+1,-.6);
 \draw[darkblue, ultra thick] (+1+.3,.6) to (+1+.3,-.6);

 \end{tikzpicture}
      }
    }
  \right]
  \;+\;
  \left[
    \raisebox{-11pt}{
      \scalebox{.8}{
 \begin{tikzpicture}

 \begin{scope}
   \clip (-1,-.1) rectangle (-1+.1,+.1);
   \draw[draw=orangeii, fill=orangeii] (-1,0) circle (.07);
 \end{scope}

 \begin{scope}[shift={(.3,0)}]
   \clip (+1,+.1) rectangle (+1-.1,-.1);
   \draw[draw=orangeii, fill=orangeii] (+1,0) circle (.07);
 \end{scope}

 \draw[thick, orangeii] (-1,0) to (1+.3,0);
 \draw[thick, white] (-1+.3-.05,0) to (-1+.3+.05,0);
 \draw[thick, white] (+1-.05,0) to (+1+.05,0);

 \draw[darkblue, ultra thick] (-1,.6) to (-1,-.6);
 \draw[darkblue, ultra thick] (-1+.3,.6) to (-1+.3,-.6);

 \draw[darkblue, ultra thick] (+1,.6) to (+1,-.6);
 \draw[darkblue, ultra thick] (+1+.3,.6) to (+1+.3,-.6);

 \end{tikzpicture}
      }
    }
  \right]
  \\
\end{array}%

\newpage

\noindent {\bf Lie algebra weight systems on horizontal chord diagrams with stacks
of coincident strands.}
The construction $\Delta$ \eqref{Delta} of horizontal chord diagram with stacks of coincident
strands passes from the plain set $\mathcal{D}^{{}^{\mathrm{pb}}}$ of horizontal chord
diagrams to a linear map on the vector space $\mathcal{A}^{{}^{\mathrm{pb}}}$ in
\eqref{DirectSumcalA}. This means that we obtain further weight systems \eqref{SpaceOfHorizontalWeightSystems} on horizontal chord diagrams by applying
Lie algebra weight systems $w_{(V,\rho)}$ from \eqref{LieAlgebraWeightSystem}
to a horizontal chord diagram $D$ after ``zooming in'' to $\Delta^{\vec N_{\mathrm{c}}}(\Delta)$,
resolving their stacks of coincident strands:

\vspace{-.4cm}

\begin{equation}
  \label{LieAlgebraWeightSystemOnHoizontalChordDiagramsWithStacksOfCoincidentStrands}
  \hspace{-3cm}
  \xymatrix@R=1pt@C=2pt{
    \mathrm{Span}
    \Big( \!\!\!\!\!\!\!
    &
    \underset{
      \mbox{
      }
    }{
    \big(
    \mathrm{MetLieMod}_{/\sim}
    \big)
    }
    \ar@{}[r]|-{\times}
    &
    \big(
      \underset{\mathbb{N}}{\oplus} \mathbb{N}
    \big)
    \ar@{}[r]|-{
       \raisebox{-12pt}{
         \scalebox{.74}{
         $\underset{\mathbb{N}}{\times} $
         }
       }
    }
    &
    \big(
    \underset{N_{\mathrm{f}} \in \mathbb{N}}{\sqcup}
    \mathrm{Sym}(N_{\mathrm{f}})
    \big)
    & \!\!\!\!\!\!\!\!\!\!
    \Big)
    \ar[rr]^-{
      \overset{
        \mbox{
          \tiny
          \color{blue}
\begin{tabular}{c}
            Lie algebra weight systems
            \\
            on horizontal chord diagrams
            \\
            with stacks of coincident strands
            \\
            $\phantom{a}$
          \end{tabular}
        }
      }{
        \mathrm{tr}_{(-)}
        \;\circ\;
        w_{(-,-)}
        \;\circ\;
        \Delta^{(-)}(-)
      }
    }
    &
    {\phantom{AAAAAAAAAAAAAAAAA}}
    &
    \;\;\;
    \mathcal{W}^{{}^{\mathrm{pb}}}
    \\
    \;\;\;\;\;\;\;\;
    \mathrlap{
     \;\;\big(
    }
    &
    \!\!\!\!\!\!\!\!\!\!\!\!\!\!\!\!\!\!
      \overset{
        \mathclap{
        \raisebox{-21pt}{
          \tiny
          \color{blue}
          \begin{tabular}{c}
            Metric Lie
            \\
            representation
            \\
            $\phantom{AA}$
          \end{tabular}
        }
        }
      }{\;\;\;\;
        \mathfrak{g}\otimes V \overset{\rho}{\to} V
      }
      \,,
      &
      \hspace{-2mm}
      \overset{
        \mathclap{
        \raisebox{-21pt}{
          \tiny
          \color{blue}
          \begin{tabular}{c}
            Stacks of
            \\
            coincident
            strands
            \\
            $\phantom{a}$
          \end{tabular}
        }
        }
      }{
        \vec N_{\mathrm{c}}
      }
      \,,
      &
      \overset{
        \mbox{
          \tiny
          \color{blue}
          \begin{tabular}{c}
            Winding
            \\
            monodromy
            \\
            $\phantom{a}$
          \end{tabular}
        }
      }{
        \sigma
      }
      &
   \!\!\!\!\!\!\!\!\!\!\!\!\!\!\!\!\!\!\!\!\!\!\!\!\!\!\!\!\! \!\!\!\!\!\big)
    \ar@{}[rr]|<<<<<{\longmapsto}
    &&
    \mathclap{
    \!\!\!\!\!\!\!\!\!
    \!\!\!\!\!\!\!\!\!
    \Big( \hspace{-.2cm}
      \overset{
        \mbox{
          \tiny
          \color{blue}
          \begin{tabular}{c}
            Horizontal
            \\
            chord diagram
            \\
            $\phantom{a}$
          \end{tabular}
        }
      }{
      \;\;
        D
      \;\;
      }
      \!\!
      \longmapsto
      \;\;\;\;
      \overset{
        \mathclap{
        \mbox{
          \tiny
          \color{blue}
          \begin{tabular}{c}
            $\sigma$-trace of value
            \\
            of $\rho$ Lie algebra weight system
            \\
            on $D$-shaped cord diagram
            \\
            with stacks of $\vec N_{\mathrm{c}}$ coincident strands
          \end{tabular}
        }
        }
      }{
        \mathrm{tr}_\sigma
        \circ
        w_{(V,\rho)}
        \circ
        \Delta^{\vec N_{\mathrm{c}}}(D)
      }
    \;\; \Big)
    }
  }
\end{equation}
For example (see also \eqref{HorizontalChordDiagramClosedToSullivanChordDiagram}):

\vspace{-.4cm}

$$
  \hspace{-.1cm}
  \begin{array}{l}
  \mathrm{tr}_{{}_{(215634)}}
  \circ
  w_{\!({\mathfrak{g}\otimes C \overset{\rho}{\to}C })}
  \circ
  \Delta^{(1,3,2)}
  \left(
    \left[
      \raisebox{-49pt}{
        \scalebox{1}{
\begin{tikzpicture}

 \begin{scope}[shift={(-4,1)}]
   \clip (0,-.1) rectangle (0+.1,+.1);
   \draw[draw=orangeii, fill=orangeii] (0,0) circle (.07);
 \end{scope}

 \begin{scope}[shift={(-2,1)}]
   \clip (-.1,-.1) rectangle (0,+.1);
   \draw[draw=orangeii, fill=orangeii] (0,0) circle (.07);
 \end{scope}

 \draw[thick, orangeii] (-4,1) to (-2,1);

 \begin{scope}[shift={(-2,0)}]
   \clip (0,-.1) rectangle (0+.1,+.1);
   \draw[draw=orangeii, fill=orangeii] (0,0) circle (.07);
 \end{scope}

 \begin{scope}[shift={(0,0)}]
   \clip (-.1,-.1) rectangle (0,+.1);
   \draw[draw=orangeii, fill=orangeii] (0,0) circle (.07);
 \end{scope}

 \draw[thick, orangeii] (-2,0) to (0,0);

 \begin{scope}[shift={(-4,-1)}]
   \clip (0,-.1) rectangle (0+.1,+.1);
   \draw[draw=orangeii, fill=orangeii] (0,0) circle (.07);
 \end{scope}

 \begin{scope}[shift={(0,-1)}]
   \clip (-.1,-.1) rectangle (0,+.1);
   \draw[draw=orangeii, fill=orangeii] (0,0) circle (.07);
 \end{scope}

 \draw[thick, orangeii] (-4,-1) to (0,-1);

 \draw[thick, white] (-2-.055,-1) to (-2+.055,-1);

 \draw[darkblue, ultra thick] (-4,1.8) to (-4,-1.8);

 \draw[darkblue, ultra thick] (-2,1.8) to (-2,-1.8);

 \draw[darkblue, ultra thick] (0,1.8) to (0,-1.8);

\end{tikzpicture}
        }
      }
    \right]
  \right)
  \;\;\;\;\;=\;\;
%%  \\
 % \\
  %\;\;\;\;
    \raisebox{-66pt}{
      \scalebox{.8}{
\begin{tikzpicture}

 \draw[black, thick, dashed] (.5+.67,.65) ellipse (.8 and 3.2);

 \draw[draw=white, fill=white] (-4.5,3) rectangle (.15+.5,-1.8);

 \draw[thick, orangeii] (-4,1) to (-2-.3,1);

 \draw[thick, orangeii] (-2+.3,0) to (-.15,0);

 \draw[thick, orangeii] (-4,-1) to (0-.15,-1);

 \draw[thick, white] (-2-.3-.055,-1) to (-2-.3+.055,-1);
 \draw[thick, white] (-2-.055,-1) to (-2+.055,-1);
 \draw[thick, white] (-2+.3-.055,-1) to (-2+.3+.055,-1);

 \draw[darkblue, ultra thick] (-4,1.8) to (-4,-1.8);

 \draw[darkblue, ultra thick] (-2,1.8) to (-2,-1.8);
 \draw[darkblue, ultra thick] (-2-.3,1.8) to (-2-.3,-1.8);
 \draw[darkblue, ultra thick] (-2+.3,1.8) to (-2+.3,-1.8);

 \draw[darkblue, ultra thick] (0-.15,1.8) to (0-.15,-1.8);
 \draw[darkblue, ultra thick] (0+.15,1.8) to (0+.15,-1.8);

 \begin{scope}[shift={(-4,1)}]
   %\clip (0,-.1) rectangle (0+.1,+.1);
   \draw[draw=orangeii, fill=orangeii] (0,0) circle (.07);
   \draw[draw=darkblue, fill=white, thick]
     (-.15,-.13)
     rectangle node{\tiny $\rho$}
     (.15,.17);
 \end{scope}

 \begin{scope}[shift={(-2-.3,1)}]
   %\clip (-.1,-.1) rectangle (0,+.1);
   \draw[draw=darkblue, fill=white, thick]
     (-.15,-.13)
     rectangle node{\tiny $\rho$}
     (.15,.17);
 \end{scope}

 \begin{scope}[shift={(-4,-1)}]
   %\clip (0,-.1) rectangle (0+.1,+.1);
   \draw[draw=orangeii, fill=orangeii] (0,0) circle (.07);
   \draw[draw=darkblue, fill=white, thick]
     (-.15,-.13)
     rectangle node{\tiny $\rho$}
     (.15,.17);
 \end{scope}

 \begin{scope}[shift={(0-.15,-1)}]
   %\clip (-.1,-.1) rectangle (0,+.1);
   \draw[draw=orangeii, fill=orangeii] (0,0) circle (.07);
   \draw[draw=darkblue, fill=white, thick]
     (-.15,-.13)
     rectangle node{\tiny $\rho$}
     (.15,.17);
 \end{scope}

 \begin{scope}[shift={(-2+.3,0)}]
   %\clip (0,-.1) rectangle (0+.1,+.1);
   \draw[draw=orangeii, fill=orangeii] (0,0) circle (.07);
   \draw[draw=darkblue, fill=white, thick]
     (-.15,-.13)
     rectangle node{\tiny $\rho$}
     (.15,.17);
 \end{scope}

 \begin{scope}[shift={(-.15,0)}]
   %\clip (-.1,-.1) rectangle (0,+.1);
   \draw[draw=orangeii, fill=orangeii] (0,0) circle (.07);
   \draw[draw=darkblue, fill=white, thick]
     (-.15,-.13)
     rectangle node{\tiny $\rho$}
     (.15,.17);
 \end{scope}

 \draw[darkblue, ultra thick]
   (-2-.3,1.9) .. controls (-2-.3,2.6) and (-4,2.4)  .. (-4,3);

 \begin{scope}[shift={(-.06,0)}]
  \draw[white, ultra thick]
    (-4,1.9) .. controls (-4,2.6) and (-2-.3,2.4)  .. (-2-.3,3);
 \end{scope}
 \begin{scope}[shift={(+.06,0)}]
  \draw[white, ultra thick]
    (-4,1.9) .. controls (-4,2.6) and (-2-.3,2.4)  .. (-2-.3,3);
 \end{scope}

 \draw[darkblue, ultra thick]
   (-4,1.9) .. controls (-4,2.6) and (-2-.3,2.4)  .. (-2-.3,3);

 \draw[darkblue, ultra thick]
   (-.15,1.9) .. controls (-.15,2.5) and (-2,2.3)  .. (-2,3);
 \draw[darkblue, ultra thick]
   (+.15,1.9) .. controls (+.15,2.7) and (-2+.3,2.5)  .. (-2+.3,3);

 \begin{scope}[shift={(-.06,0)}]
 \draw[white, ultra thick]
   (-2,1.9) .. controls (-2,2.7) and (-.15,2.5)  .. (-.15,3);
 \draw[white, ultra thick]
   (-2+.3,1.9) .. controls (-2+.3,2.5) and (-.15+.3,2.3)  .. (-.15+.3,3);
 \end{scope}
 \begin{scope}[shift={(+.06,0)}]
 \draw[white, ultra thick]
   (-2,1.9) .. controls (-2,2.7) and (-.15,2.5)  .. (-.15,3);
 \draw[white, ultra thick]
   (-2+.3,1.9) .. controls (-2+.3,2.5) and (-.15+.3,2.3)  .. (-.15+.3,3);
 \end{scope}

 \draw[darkblue, ultra thick]
   (-2,1.9) .. controls (-2,2.7) and (-.15,2.5)  .. (-.15,3);
 \draw[darkblue, ultra thick]
   (-2+.3,1.9) .. controls (-2+.3,2.5) and (-.15+.3,2.3)  .. (-.15+.3,3);

 \draw[gray, dashed,  thick] (-4.5,3) to (.15+.5,3);
 \draw[gray, dashed, thick] (-4.5,-1.8) to (.15+.5,-1.8);

 \begin{scope}[shift={(.15,.44)}]
   \draw[gray, dashed,  thick] (-4.5,3) to (.15+.5,3);
 \end{scope}
 \begin{scope}[shift={(.15,-3-1.8-.44)}]
   \draw[gray, dashed,  thick] (-4.5,3) to (.15+.5,3);
 \end{scope}

 \begin{scope}[shift={(.33,.84)}]
   \draw[gray, dashed,  thick] (-4.5,3) to (.15+.5,3);
 \end{scope}
 \begin{scope}[shift={(.45,-3-1.8-.755)}]
   \draw[gray, dashed,  thick] (-4.5,3) to (.15+.5,3);
 \end{scope}

\end{tikzpicture}
      }
    }
    \;
    +
    \\
\end{array}
$$

\vspace{.2cm}

\noindent {\bf Fundamental theorem on horizontal weight systems.}
With \cite[Cor. 2.6]{BarNatan96} we now obtain:
\begin{prop}
  \label{FundamentalTheremOfWeightSystems}
  All weight systems on horizontal chord diagrams
  (Def. \ref{HorizontalWeightSystems}) are
  linear combinations of Lie algebra
  weight systems with stacks of coincident strands \eqref{LieAlgebraWeightSystemOnHoizontalChordDiagramsWithStacksOfCoincidentStrands}
  for (at least) the general linear Lie algebras
  $\mathfrak{g}
    \;\in\;
    \big\{
      \mathfrak{gl}(N)
      \;\vert\;
      N \in \mathbb{N}_{\geq 2}
    \big\}$
  over the given ground field:
  For these Lie algebras the construction
  \eqref{LieAlgebraWeightSystemOnHoizontalChordDiagramsWithStacksOfCoincidentStrands}
  is surjective,
  so that on the quotient $(-)_{/\sim}$ by its kernel
  it is a linear bijection.
    In particular, with $\mathfrak{gl}(2, \mathbb{C}) \simeq \mathfrak{su}(2)_{\mathbb{C}} \oplus \mathbb{C}$ we have:

  \vspace{-.4cm}

  \begin{equation}
  \label{FundamentalTheoremEquivalence}
  \hspace{-1cm}
    \xymatrix@C=4em{
    \mathrm{Span}
    \Big(
      \overset{
        \mathclap{
        \mbox{
          \tiny
          \color{blue}
\begin{tabular}{c}
            finite-dimensional
            \\
            $(\mathfrak{su}(2)_{\mathbb{C}} \oplus \mathbb{C})$-representations
            \\
            $\phantom{a}$
          \end{tabular}
        }
        }
      }{
        (
          \mathfrak{su}(2)_{\mathbb{C}}
          \oplus
          \mathbb{C}
        )
        \mathrm{MetMod}_{/\sim}
      }
      \times
      \overset{
        \mathclap{
        \mbox{
          \tiny
          \color{blue}
          \begin{tabular}{c}
            Stacks of
            \\
            coincident strands
            \\
            $\phantom{a}$
          \end{tabular}
        }
        }
      }{
        \big(
          \underset{\mathbb{N}}{\oplus}
          \mathbb{N}
        \big)
      }
      \underset{\mathbb{N}}{\times}
      \overset{
        \mbox{
          \tiny
          \color{blue}
          \begin{tabular}{c}
            Winding
            \\
            monodromies
            \\
            $\phantom{a}$
          \end{tabular}
        }
      }{
        \big(
          \underset{N_{\mathrm{f}} \in \mathbb{N}}{\sqcup}
          \mathrm{Sym}(N_{\mathrm{f}})
        \big)
      }
    \Big)_{
          \underset{
            \mathclap{
            \mbox{
              \tiny
              \color{blue}
              \begin{tabular}{c}
                relations
              \end{tabular}
            }
            }
          }{
            /\sim
          }
    }
      \ar[rr]_<<<<<<<<<<<<{
        \scalebox{1.1}{
          \color{magenta}
          $\simeq$
        }
      }^<<<<<<<<<<<<{
        \overset{
        \mbox{
          \tiny
          \color{blue}
          \begin{tabular}{c}
            assign Lie algebra weights
            \\
            $\phantom{a}$
          \end{tabular}
        }
        }
        {
          \mathrm{tr}\, \circ \, w  \, \circ \, \Delta
        }
      }
      && \;\;\;
      \overset{
        \mathclap{
        \mbox{
          \tiny
          \color{blue}
          \begin{tabular}{c}
            Weight systems on
            \\
            horizontal chord diagrams
            \\
            $\phantom{a}$
          \end{tabular}
        }
        }
      }{
        \mathcal{W}^{{}^{\mathrm{pb}}}
      }.
    }
  \end{equation}
\end{prop}

\noindent {\bf Conclusion:}
Weight systems on horizontal chord diagrams
is a theory of
1) metric super Lie representations,
2) stacks of coincident strands,
3) winding monodromies,
subject to 4) dualities. In \cref{WeightSystemsAsObservablesOnIntersectingBranes} we match this to physics.

\newpage

\thispagestyle{empty}

\begin{example}[Fundamental $\mathfrak{gl}(2,\mathbb{C})$-weight system]
  \label{TheFundamentalgl2WeightSystem}
  Considesr the
  Lie algebra
  $\mathfrak{su}(2)_{\mathbb{C}} \oplus \mathbb{C}
  \;\simeq\; \mathfrak{gl}(2,\mathbb{C})$
  equipped with the metric

  \vspace{-.9cm}

  \begin{equation}
    \label{gl2FundamentalTraceMetric}
    g(-,-)
    \;\coloneqq\;
    \mathrm{tr}_{\mathbf{2}}(-\cdot -)
  \end{equation}
  given by the trace in its
  defining fundamental representation $\mathbf{2}$;
  and consider the corresponding Lie algebra weight system
\eqref{LieAlgebraWeightSystemOnHoizontalChordDiagramsWithStacksOfCoincidentStrands}
with trivial winding monodromy and no stacks of coincident strands:

\vspace{-.5cm}

\begin{equation}
  \label{Fundamentalgl2WeightSystem}
  w_{\mathbf{2}}
  \;:=\;
  \mathrm{tr}_{\mathrm{id}}
  \circ
  w_{(\mathfrak{gl}(2,\mathbb{C}), \mathbf{2})}
  \circ
  \Delta^{(1,1\cdots, 1)}
  \,.
\end{equation}

\vspace{-.1cm}

  \noindent An elementary computation reveals that
  the value of the Lie algebra weight system \eqref{Fundamentalgl2WeightSystem}
  on a single chord is the braiding operation \cite[Fact 6]{BarNatan96}
  (see also \cref{LieAlgebraWeightSystemsEncoded3Algebras}
  and \cref{RoundWeightSystemsEncodeOpenStringAmplitudes} below):
\begin{equation}
  \label{gl2ChordIsBraiding}
\mbox{

}
\\
\hline
\end{tabular}
}
\end{equation}
This directly implies that the value
(according to \cref{LieAlgebraWeightSystems})
of the weight system \eqref{gl2FundamentalTraceMetric}
on a horizontal chord diagram $D \in \mathcal{D}^{\mathrm{pb}}$\
equals 2 taken to the power of the number of cycles in the
corresponding permutation:

\vspace{-.8cm}

  \begin{equation}
    \label{ValueOfFundamentalgl2WeightSystemOnChordDiagram}
    w_{\mathbf{2}}([D])
    \;=\;
    2^{ \#\mathrm{cycles}(\mathrm{perm}(D)) }
    \,,
    \phantom{A}
    \mbox{where}
    \phantom{AA}
    \xymatrix{
      \overset{
        \mathclap{
        \mbox{
          \tiny
          \color{blue}

  }
  $\;\;\; = \;\;\; 2^2$
  \end{tabular}
  }
\end{equation}
\end{example}

\newpage

%%%%%%%%%%%%%%%%%%%%%%%%%%%%%%%%%%%%%%%%%%%%%%%
\subsection{Quantum states on chord diagrams}
\label{QuantumStatesOnChordDiagrams}
%%%%%%%%%%%%%%%%%%%%%%%%%%%%%%%%%%%%%%%%%%%%%%%

We observe here that the algebra of horizontal chord diagrams
is canonically a star-algebra (Prop. \ref{StarStructureOnHorizontalChordDiagrams} below),
and as such qualifies as an \emph{algebra of observables}
according to quantum probability theory (see \cite{Swart17}\cite{Landsman17} for reviews).
This exhibits weight systems as linear maps assigning probability amplitudes to observables.
We may therefore consider those weight systems which are
\emph{quantum states} (density matrices) in that they
assign consistent expectation values to
real (i.e., self-adjoint) observables (Def. \ref{States} below).
An example is the fundamental
$\mathfrak{gl}(2,\mathbb{C})$-weight system
(Example \ref{gl2FundamentalWeightSystemIsAState} below)
which, further below in \cref{LargeNLimitM2M5BraneBoundStates},
we identify  with the state of two coincident transversal M5-branes
in the BMN matrix model.

\medskip

The following Definition \ref{StarAlgebra}
is traditionally considered for Banach algebras, where it
yields the concept of $C^*$-algebras;
see for instance \cite[Def. C.1]{Landsman17}.
We need the simple specialization to
finite-dimensional star-algebras (e.g., \cite[2.1]{Swart17}),
or rather the evident mild generalization of that
to degreewise finite-dimensional graded star-algebras:
\begin{defn}[Star-algebra]
  \label{StarAlgebra}
  A \emph{star-algebra} (for the present purpose) is a
  degreewise finite-dimensional graded associative algebra
  $\mathcal{A}$ over the complex numbers, equipped with an
  involutive anti-linear anti-homomorphism $(-)^\ast$,
   the \emph{star-operation}, hence with a
  function
  $$
    \xymatrix{
      \mathcal{A}
      \ar[r]^-{(-)^\ast}
      &
      \mathcal{A}
    }
  $$
  which satisfies:

  \medskip
\hspace{-.8cm}
  \begin{tabular}{lll}
  {\bf (i) Degree}:
    & $\mathrm{deg}(A) = \mathrm{deg}(A^\ast)$
    &
    \multirow{1}{*}{
      for all homogeneous $A \in \mathcal{A}$
    }.
  \\
  {\bf (ii)  Anti-linearity}:
    &
    $\big( a_1 A_1 + a_2 A^2\big)^\ast = \bar a_1 A_1^\ast + \bar a_2 A_2^\ast$.
    &
    \multirow{3}{*}{
      for all $a_i \in \mathbb{C}$, $A_i \in \mathcal{A}$
    }.
  \\
  {\bf (iii) Anti-homomorphism}:
    &
    $\big( A_1 A_2\big)^\ast = A_2^\ast A_1^\ast$
  \\
  {\bf (iv) Involution}:
    & $\left( (A)^\ast \right)^\ast = A$,
  \end{tabular}\\

  \noindent where $\bar a_i$ denotes the complex conjugate of
  $a_i$.
\end{defn}

\begin{prop}[Star-structure on horizontal chord diagrams]
  \label{StarStructureOnHorizontalChordDiagrams}
  The algebra of horizontal chord diagrams \eqref{AlgebraOfHorizontalChordDiagrams}
  becomes a complex star-algebra (Def. \ref{StarAlgebra})
  via the star-operation
  $$
    \xymatrix@R=-1pt{
      \mathcal{A}^{{}^{\mathrm{pb}}}
      \ar[rr]^-{ (-)^\ast}
      &&
      \mathcal{A}^{{}^{\mathrm{pb}}}
      \\
      a_1 \cdot D_1
      +
      a_2 \cdot D_2
      \ar@{|->}[rr]
      &&
      \bar a_1 \cdot D_1^\ast
      +
      \bar a_2 \cdot D_2^\ast
    }
  $$
  where
  $$
    \xymatrix{
      \mathcal{D}^{{}^{\mathrm{pb}}}
      \ar[rr]^-{ (-)^\ast}
      &&
      \mathcal{D}^{{}^{\mathrm{pb}}}
    }
  $$
  is the operation on horizontal chord diagrams
  \eqref{HorizontalChordDiagrams}
  that reverses the orientation of strands in a chord diagram.
\end{prop}

\noindent For example:
$$
\left(
a
\cdot
\left[
\scalebox{.8}{
\raisebox{-105pt}{

}
}
\right]
$$

\begin{remark}[Loop and configuration spaces]
  By Prop. \ref{HigherObservablesEquivalentToCohomologyOfLoopedConfigurationSpace}
  the algebra of horizontal chord diagrams
  is equivalently the
  homology Pontrjagin algebra of a based loop space
  (namely of an ordered configuration space of points).
  As such it is a Hopf algebra with involutive antipode,
  and this is the star-structure of Prop. \ref{StarStructureOnHorizontalChordDiagrams}.
\end{remark}

\medskip

The following Definition \ref{States} is standard in quantum probability theory
and in algebraic quantum (field) theory
(see, for instance, \cite[Def. 2.4]{Landsman17}\cite[2.3]{Swart17}).
\begin{defn}[Quantum state on a star-algebra]
  \label{States}
  Given a complex star-algebra $(\mathcal{A}, (-)^\ast)$
  (Def. \ref{StarAlgebra}),
  a (possibly mixed-, quantum-)\emph{state} (or \emph{density matrix})
  is a complex-linear function
  $$
    \rho
    :
    \mathcal{A}
    \longrightarrow
    \mathbb{C}
  $$
  which satisfies:\\

\hspace{-.8cm}
  \begin{tabular}{lll}
    {\bf (i) Positivity}:
      & $ \rho\big( A A^\ast \big) \geq 0 \in \mathbb{R} \subset \mathbb{C}$
      & for all $A \in \mathcal{A}$.
    \\
    {\bf (ii) Normalization}:
      & $\rho(\mathbf{1}) = 1$
      & for $\mathbf{1} \in \mathcal{A}$ the algebra unit.
  \end{tabular}\\

\end{defn}

\begin{remark}[Interpretation]
  \label{NormalizationToStates}
  The point of Def. \ref{States} is the positivity condition
  (which might rather deserve to be called \emph{semi-positivity},
  but \emph{positivity} is the established terminology here)
  while the normalization condition
  is just that: If $\rho$ is a (semi-)positive linear map
  with $\rho(\mathbf{1}) \neq 0$ then
  $\frac{1}{\rho(\mathbf{1})}\rho$ is a state.
\end{remark}

\begin{example}[Fundamental $\mathfrak{gl}(2)$-weight system is a state]
  \label{gl2FundamentalWeightSystemIsAState}
  The normalization (Remark \ref{NormalizationToStates})
  of the fundamental $\mathfrak{gl}(2)$-weight system
  $w_{\mathbf{2}}$ (Example \ref{TheFundamentalgl2WeightSystem})
  is a quantum state (Def. \ref{States}) with respect to the
  star-algebra structure on horizontal chord diagrams
  from Prop. \ref{StarStructureOnHorizontalChordDiagrams}.
\end{example}
The full proof establishing Example \ref{gl2FundamentalWeightSystemIsAState}
is relegated to \cite{CSS20},
here we just indicate the idea by proving the first non-trivial case:

\begin{remark}[Bilinear form on permutations]
  \label{BilinearFormOnPermutations}
  On the complex-linear span of the set of permutations on
  $N_{\mathrm{f}}$ elements, consider the sesqui-linear form
  \begin{equation}
    \label{SesquilinearFormOnPermutations}
    \xymatrix@R=-2pt{
      \mathbb{C}[\mathrm{Sym}(N_{\mathrm{f}})]
      \times
      \mathbb{C}[\mathrm{Sym}(N_{\mathrm{f}})]
      \ar[rr]
      &&
      \mathbb{C}
      \\
      (a_1 \sigma_1, \, a_2 \sigma_2)
      \ar@{|->}[rr]
            &&
      a_1 \bar a_2
      \,
      2^{\#\mathrm{cycles}(\sigma_1 \circ \sigma_2^{-1})}
    }
  \end{equation}
  The statement of example \ref{gl2FundamentalWeightSystemIsAState}
  is equivalent,
  by \eqref{ValueOfFundamentalgl2WeightSystemOnChordDiagram} in  Example \ref{TheFundamentalgl2WeightSystem},
  to the statement that the sesqui-linear form \eqref{SesquilinearFormOnPermutations}
  is
  positive semi-definite, in that for any formal linear combination
  of permutations
  $\Sigma \in \mathbb{C}[\mathrm{Sym}(N_{\mathrm{f}})]$ we have
  $$
    \left\vert \Sigma \right\vert^2
    \;:=\;
    \left\langle \Sigma,\Sigma\right\rangle
    \;\geq\;
    0
    \;
    \in \mathbb{R}
    \subset \mathbb{C}
    \,.
  $$
\end{remark}

\begin{lemma}[Positivity of sesquilinear form for length 2]
  \label{SesquilinearFormPositiveOnLinearCombinationOfLength2}
  The fundamental $\mathfrak{gl}(2,\mathbb{C})$-weight system
  $w_{\mathbf{2}}$ (Example \ref{TheFundamentalgl2WeightSystem})
  is (semi-)positive on the subspace
  of $\mathcal{A}^{{}^{\mathrm{pb}}}_{N_{\mathrm{f}}}$
  consisting of formal linear combinations of length two:
  $$
    w_{\mathbf{2}}
    \big(
      a_1 [D_1] + a_2 [D_2]
    \big)
    \;\geq\;
    0
    \phantom{AAA}
    \mbox{for all $a_i \in \mathbb{C}$, $D_i \in \mathcal{D}^{{}^{\mathrm{pb}}}_{N_{\mathrm{f}}}$}
    \,.
  $$
\end{lemma}
\begin{proof}
  By Remark \ref{BilinearFormOnPermutations}, we equivalently
  have to show that
  \begin{equation}
    \label{FormOnPermutationsPositive}
    \big\vert
      a_1 [\sigma_1] + a_2 [\sigma_2]
    \big\vert^2
    \;\geq\;
    0
    \phantom{AAA}
    \mbox{for all $a_i \in \mathbb{C}$, $\sigma_i \in \mathrm{Sym}(N_{\mathrm{f}})$}
    \,.
  \end{equation}
  Observing that
  $$
    \begin{aligned}
      \big\vert
        a_1 [\sigma_1] + a_2 [\sigma_2]
      \big\vert^2
      &
      =
      \left(
      \left\vert a_1\right\vert^2
      +
      \left\vert a_2\right\vert^2
      \right)
      2^{N_{\mathrm{f}}}
      +
      \left(
        a_1 \bar a_2
        +
        a_2 \bar a_1
      \right)
      2^{ \#\mathrm{cycles}(\sigma_1 \circ \sigma_2^{-1}) }
      \\
      & =
      \underset{
        \mathclap{
        \gt 0 \in \mathbb{R}
        }
      }{
      \underbrace{
        N_{\mathrm{f}}
      }}
      \Bigg(
        \left\vert a_1\right\vert^2
        +
        \left\vert a_2\right\vert^2
        +
        \left(
          a_1 \bar a_2
          +
          a_2 \bar a_1
        \right)
        \underset{
          \in (0,1]
        }{
        \underbrace{
          \frac{
            2^{ \#\mathrm{cycles}(\sigma_1 \circ \sigma_2^{-1}) }
          }
          {
            2^{N_{\mathrm{f}}}
          }
        }
        }
      \Bigg)
    \end{aligned}
    \,,
  $$
  the statement \eqref{FormOnPermutationsPositive}
  follows from the ``cosine rule''
  $
    \left\vert
      a_1 \bar a_2
      +
      a_2 \bar a_1
    \right\vert
    \;\leq\;
    \left\vert a_1\right\vert^2
    +
    \left\vert a_2\right\vert^2
  $.
\end{proof}

\newpage

%%%%%%%%%%%%%%%%%%%%%%%%%%%%%%%%%%%%%%%%%%%%%%%
\section{Chord diagrams and intersecting branes}
\label{WeightSystemsAsObservablesOnIntersectingBranes}
%%%%%%%%%%%%%%%%%%%%%%%%%%%%%%%%%%%%%%%%%%%%%%

By the isomorphism \eqref{ObservablesAreWeightSystems}, the higher observables
\eqref{HigherObservablesOnD6D8Intersections} on the moduli space
of $\mathrm{D}p \!\perp\! \mathrm{D}(p+2)$-brane intersections,
as described in diagram \eqref{ConfigurationSpaceOfBraneIntersections},
are given by weight systems on horizontal chord diagrams,
discussed in \cref{WeightSystemsOnChordDiagrams}.
Here we discuss how, under this interpretation,
these weight systems
from \cref{WeightSystemsOnChordDiagrams}
turn out to capture various structures known,
or rather: expected, in intersecting brane physics.

%%%%%%%%%%%%%%%%%%%%%%%%%%%%%%%%%%%%%%%%%%%%%%%%%%%%%%%%%%%%%%%%%%%%%%
\subsection{Lie algebra weight systems give matrix model single trace observables}
\label{LieAlgebraWeightSystemsAreSingleTraceObservables}
%%%%%%%%%%%%%%%%%%%%%%%%%%%%%%%%%%%%%%%%%%%%%%%%%%%%%%%%%%%%%%%%%%%%%%

\begin{observation}
 \label{TheObservation}
By Prop. \ref{RelationOfWeightSystems} and  Prop \ref{FundamentalTheremOfWeightSystems}
\emph{all} weight systems
(Def. \ref{HorizontalWeightSystems}, \ref{WeightSystemsOnRoundChordDiagrams}), on any
of {\bf (a)} horizontal chord diagrams \eqref{HorizontalChordDiagrams},
{\bf (b)} round chord diagrams \eqref{RoundChordDiagrams},
and {\bf (c)} Jacobi diagrams \eqref{JacobiDiagrams}
evaluate, in the end, to a sum of circular string diagrams. The latter,
in turn, by the rules of Penrose notation/string diagram
calculus from  \cref{LieAlgebraWeightSystems},
evaluate to a trace of a long product of matrices
and summed over sets of pairs of matrices.
For example the diagram on the left below evaluates to the
trace expression shown on the right
(with $(\rho_a \cdot \rho_b)^i{}_j = \rho_a{}^i{}_l \, \rho_b{}^l{}_j$ denoting the matrix product):
\begin{equation}
\label{LieWeightSystemOnJacobiDiagram}
\mathllap{
  \mbox{
   \tiny
   \color{blue}

}
\;\;\;\;
 =
\;\;\;\;
\mathrm{Tr}_{{}_V}
\big(
  \,
  \rho_a
    \cdot
  \rho_b
    \cdot
  \rho_c
    \cdot
  \rho^b
    \cdot
  \rho_d
    \cdot
  \rho^c
    \cdot
  \rho_a
    \cdot
  \rho^e
    \cdot
  \rho^d
  \,
\big)
\end{equation}
Notice that the string diagram on the left may be, but \emph{need not be}, the
exact image of a round chord diagram of the same shape.
In general it is the result a process of duplication and of reconnecting
of strands, according to \eqref{LieAlgebraWeightSystemOnHoizontalChordDiagramsWithStacksOfCoincidentStrands}.
However, the end result is always a sum over terms of this circular shape,
hence is a sum of traces as on the right of \eqref{LieWeightSystemOnJacobiDiagram}.
(This if the monodromy permutation in \eqref{LieAlgebraWeightSystemOnHoizontalChordDiagramsWithStacksOfCoincidentStrands}
has a single cycle, otherwise one gets traces along several connected circles, discussed in \cref{HorizontalChordDiagramsEncodesStringTopologyOperations}.)
\end{observation}

\noindent
{\bf Single trace observables subject to Wick's theorem are weight systems.}
Given a metric Lie representation
$\mathfrak{g}\otimes V \overset{\rho}{\to} V$ as in \cref{LieAlgebraWeightSystems},
consider a quantum field or a random variable $Z$ with values in
$\mathfrak{g}$, hence with component expansion
$
  Z = Z_a \, \rho^a
  $.
A \emph{single trace observable} in $Z$ is an operator/random variable
of the form
\begin{equation}
  \label{SingleTraceObservable}
  \mathcal{O}
  \;=\;
  \mathrm{Tr}\big( Z \cdot Z \cdot \cdots \cdot Z \big)
  \,.
\end{equation}
Assume then that the component fields $Z_a$ are free quantum fields,
or random variables of multivariate Gaussian distribution
with covariance given by the metric $k$ on $V$:
$
  \big\langle
    Z_a Z_b
  \big\rangle
  =
  k_{a b}
  $.
Then \emph{Wick's theorem} says that the higher moments
of $Z$ are sums of contractions labelled by \emph{linear chord diagrams},
as shown, by example, in the first two lines here:

\vspace{-.3cm}

\begin{equation}
\label{LinearChordsTraced}
\raisebox{-.7cm}{
\hspace{-.5cm}\includegraphics[width=\textwidth]{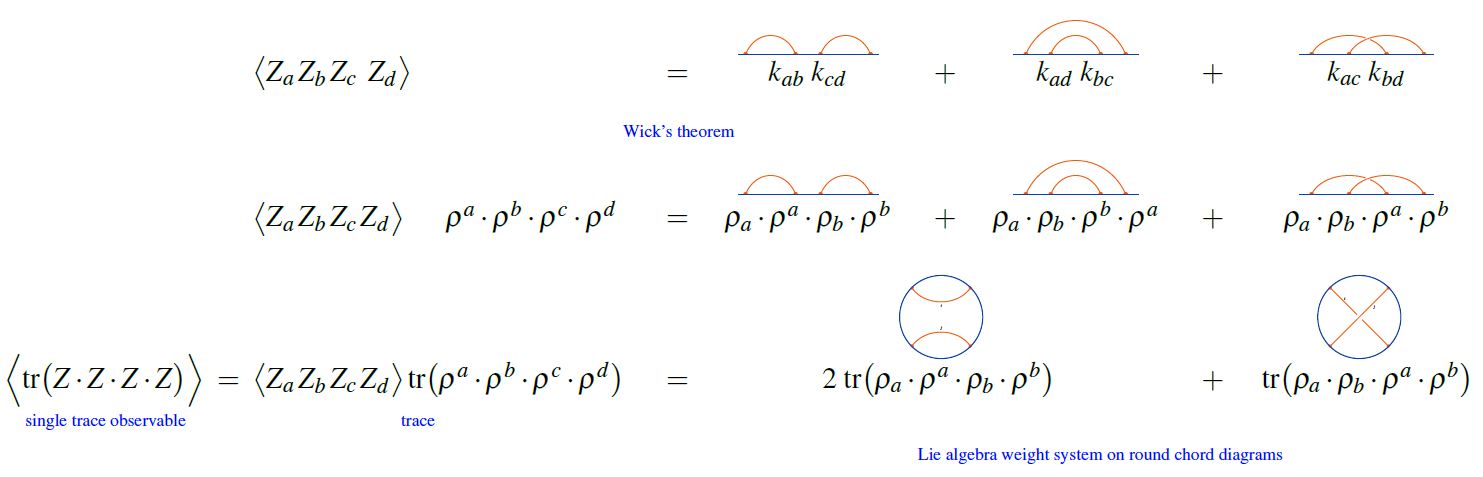}
}
\end{equation}

\newpage
But then, as shown by example in the last line, the trace
that defines the single trace observables closes up
the resulting matrix product such that the terms that were previously
controlled by linear chord diagrams are now labelled by round
chord diagrams \eqref{RoundChordDiagrams}.
Moreover,  comparison with Observation \ref{TheObservation} shows that
the summands contributing to the expectation value of
the single trace observables are exactly the values of
Lie algebra weight systems on these round chord diagrams.

\medskip

\noindent {\bf The SYK-model compactification of M5-branes.}
An observation along the lines of \eqref{LinearChordsTraced}
(with emphasis on the appearance of chord diagrams,
but without the identification of weight systems)
was recently found to be crucial for the analysis of single trace observables
in the SYK-model (review in \cite{Rosenhaus18}) and analogous systems; see
\cite[Sec. 2.2]{GGJV18}\cite[Sec. 4]{JiaVerbaarschot18}\cite[Sec. 2.1]{BNS18}\cite[Sec. 2]{BINT18}\cite[5-21]{Narovlansky19}.
Notice that from the point of view of string/M-theory,
the SYK-model is the {(near-)CFT} which is the holographic dual to
the full compactification of the M5-brane; see \cite[4.1]{LLL18}\cite{BHT18}.

\medskip

\noindent {\bf The BMN matrix model and fuzzy sphere states.}
The \emph{BFSS matrix model} famously is a (0+1)-dimensional super Yang-Mills
theory which is thought to describe at least a sector of M-theory
(but see \cite[[p. 43-44]{Moore14}) with
un-wrapped M2-branes \cite{NicolaiHelling98}\cite{DasguptaNicolaiPleftka02},
or equivalently strongly coupled type IIA string theory
with stacks of unbound D0-branes \cite{BFSS96},
both on asymptotically Minkowski spacetime backgrounds
(review in \cite{Banks97}\cite{Taylor01}).
The \emph{BMN matrix model} \cite[5]{BMN02}\cite{DSJVR02},
which is the KK-compactification over $S^3$ of $D=4 \mathcal{N}=4$ super Yang-Mills theory \cite{KKP03},
generalizes this to asymptotically gravitational pp-wave backgrounds, which arise as
Penrose limits of both the
$\mathrm{Ads}_{4,7} \times S^{7,4}$ near horizon geometries of black M2-branes and M5-branes (\cite[4.7]{Blau04}), and which deform the
action functional of the BFSS model by a mass and a Chern-Simons term.
These extra terms in the BMN model
lift the notoriously problematic ``flat directions'' of the
BFSS model (\cite{dWLN89}, see \cite{BGR18})
thus leading to a well-defined quantum mechanics,
which describes wrapped M2-branes (giant gravitons)
or equivalently of $\mathrm{D}p\perp \mathrm{D}(p+2)$-brane bound states
for $p = 0$ \cite{Lin04}:

\medskip
The supersymmetric solutions are precisely
\cite[(5.4)]{BMN02}\cite[4.2]{DSJVR02} those
matrix configurations that constitute complex
$\mathrm{su}(2)$-representations
$\mathfrak{su}(2)_{\mathbb{C}} \otimes V \overset{\rho}{\to} V$,
interpreted as systems of fuzzy 2-sphere geometries
(discussed as such below in \cref{su2WeightSystemsAreFuzzyFunnelObservables}).
This means
by Observation \ref{TheObservation}
with \eqref{FundamentalTheoremEquivalence} that:

\medskip
\noindent {\it
The $S^2$-rotation invariant single-trace observables
of the BMN matrix model are equivalently round chord diagrams $D$,
evaluated on supersymmetric ground states
$(\mathfrak{su}(2)_{\mathbb{C}} \otimes V \overset{\rho}{\to}V)$
by pairing them with the corresponding Lie algebra weight system
\eqref{LieWeightSystemOnJacobiDiagram}.}
This generalizes to multi-trace observables, discussed in \cref{HorizontalChordDiagramsEncodesStringTopologyOperations} below.

\medskip
\noindent In view of this it may be worthwhile to briefly recall:

\medskip
\noindent {\bf The general relevance of single trace observables in $\mathrm{AdS}/\mathrm{CFT}$.}
Single trace observables
$
  \mathcal{O}
  \;=\;
  \mathrm{Tr}
  \big(
    Z \cdot Z \cdot \cdots \cdot Z
  \big)
$
on the gauge theory side play a special role in the AdS/CFT correspondence.
They map to
single string excitations on the AdS side, in a way
that identifies the string quite literally with the
\emph{string of characters} $Z \cdot Z \cdot Z \cdots$ in the
expression of the single trace observables.
An early account of the general mechanism is in \cite{Polyakov02},
whose author already outlines the grand picture,
indicating  that space-time is gradually disappearing in the regions of large curvature, and the
natural description is provided by a gauge theory in which the basic objects are the texts formed from the gauge-invariant words, and the theory provides us with the expectation values assigned to the various texts, words and sentences.
The first concrete realization in $D = 4$, $\mathcal{N}=4$ SYM
is due to \cite{BMN02}, whose authors find that the ``string of $Z$s'' becomes the physical string and that each $Z$ carries one unit of $J$ which is one unit of momentum, and that locality along the worldsheet of the string comes from the fact that planar diagrams allow only contractions of neighboring operators. This led the authors to conclude that  the Yang-Mills theory gives a string bit model where each bit is a $Z$ operator.
%\begin{quote}
%  {\it In summary, the ``string of $Z$s'' becomes the physical string and that each $Z$ carries one unit of $J$ which is one unit of $p_+$. Locality along the worldsheet of the string comes from the fact that planar diagrams allow only contractions of neighboring operators. So the Yang Mills theory gives a string bit model where each bit is a $Z$ operator.}
%\end{quote}
%
%\begin{quote}
%\it
%The picture which slowly arises from the above considerations is that of the space-time gradually disappearing in the regions of large curvature. The natural description in this case is provided by a gauge theory in which the basic objects are the texts formed from the gauge-invariant words. The theory provides us with the expectation values assigned to the various texts, words and sentences.
%\end{quote}
%The first concrete realization in $D = 4$, $\mathcal{N}=4$ SYM
%is due to \cite{BMN02}, whose authors find:
%\begin{quote}
%  {\it In summary, the ``string of $Z$s'' becomes the physical string and that each $Z$ carries one unit of $J$ which is one unit of $p_+$. Locality along the worldsheet of the string comes from the fact that planar diagrams allow only contractions of neighboring operators. So the Yang Mills theory gives a string bit model where each bit is a $Z$ operator.}
%\end{quote}
See also \cite{GKP02} for similar arguments.

\medskip
The correspondence between single trace operators
in CFT and string excitations on AdS came to full fruition when it
was realized that the single trace operators of a given length
behave as integrable spin chains
when the dilatation operator is regarded as their Hamiltonian.
This led to the celebrated precision checks of $\mathrm{AdS}_5/\mathrm{CFT}_4$
starting with \cite{BFST03}, reviewed in \cite{Bea10}.

\medskip

%%%%%%%%%%%%%%%%%%%%%%%%%%%%%%%%%%%%%%%%%%%%%%%%%%%%%%%%%%%%%%%%%%%%%%%
\subsection{Lie algebra weight systems give fuzzy funnel observables}
\label{su2WeightSystemsAreFuzzyFunnelObservables}
%%%%%%%%%%%%%%%%%%%%%%%%%%%%%%%%%%%%%%%%%%%%%%%%%%%%%%%%%%%%%%%%%%%%%%%

\hypertarget{FigureC}{}
\begin{minipage}[l]{7.9cm}

{\bf Fuzzy funnels of  $\mathrm{D}p \!\perp \! \mathrm{D}(p+2)$ intersections.}
The configuration of $N_c$ coincident $\mathrm{D}p$-branes
ending on a $\mathrm{D}(p+2)$-brane is famously a noncommutative
``fuzzy funnel''
geometry
\cite{Diaconescu97}\cite{CMT99}\cite[4]{Myers01}\cite[3.4.3]{GaiottoWitten08}
(\hyperlink{su2RepToFuzzyFunnel}{see Figure 3}),
where the three
$\mathfrak{u}(N_{\mathrm{c}})$-valued scalar
fields $\{X^1, X^2, X^3\}$ on the $\mathrm{D}p$-branes
solve Nahm's equation with a pole as
\begin{equation}
  \label{SolutionOfNahmEquation}
  X^a(y)
    \;=\;
  \frac{1}{y}
    \,
  \frac{2}{\sqrt{N^2 - 1 } }
    \,
  \rho^a
  \,,
\end{equation}
for $y$ the transversal distance from the $\mathrm{D}(p+2)$-brane
and  $\{\rho_a\}$
the matrices of the $N_{\mathrm{c}}$-dimensional irreducible representation
of $\mathfrak{su}(2)_{\mathbb{C}}$.
Due to the Casimir relation
$$
  X_a \cdot X^a
    \;=\;
  \tfrac{1}{y^2} \, \mathrm{1}_{{}_{N_{\mathrm{c}}\times N_{\mathrm{c}}}}
$$
this means that at fixed distance $y$ the algebra of functions
generated by the scalar fields is that on the fuzzy 2-sphere
$S^2_{N_{\mathrm{c}}}$ \cite{Madore92} of radius $R = 1/y$.

\medskip
\noindent  {\bf Shape observables on fuzzy 2-spheres.}
The fuzziness of the fuzzy 2-sphere $S^2_{N_{\mathrm{c}}}$ is
reflected in the fact that functions of its radius $R$
are not all constant, due to ordering ambiguity
in the observables of the schematic form ``$R^{2k}$''.
After averaging/integration
over the fuzzy 2-sphere, hence under the trace operation, the remaining ordering ambiguities
are fully reflected by round chord diagrams,
as shown on the right by the first few examples.
Hence these
\emph{radius fluctuation amplitude} observables on the
fuzzy 2-sphere are equivalently
the values of $\mathfrak{su}(2)_{\mathbb{C}}$-weight systems on round
chord diagrams, as in \cref{LieAlgebraWeightSystems},
see Prop. \ref{FundamentalTheremOfWeightSystems}.
In fact, these fuzzy shape observables are instances of
single trace observables as in \cref{LieAlgebraWeightSystemsAreSingleTraceObservables}.

\medskip
\noindent {\bf $1/N_{\mathrm{N}_c}$-Corrections
to the $\mathrm{D}p\perp \mathrm{D}(p+2)$-systems.}
In the large $N_{\mathrm{c}}$ limit the fuzzy 2-sphere
$S^2_{N_{\mathrm{c}}}$ approaches the ordinary 2-sphere,
and its
fuzzy shape observables all converge to unity.
This large $N_{\mathrm{c}}$ limit of the
$\mathrm{D}p\perp \mathrm{D}(p+2)$-intersections
had been studied in \cite{CMT99}.
But discussion of small $N_{\mathrm{c}}$ corrections,
or even of the full matrix model mechanics
of $\mathrm{D}p \!\perp\! \mathrm{D}(p+2)$-intersections
requires handling the multitude of fuzzy shape observables
as shown on the right.
That and how these computations are crucially organized by round
chord diagrams was noticed in
\cite[Sec. 3.2]{RST04}
for review see \cite[A]{MPRS06}\cite[4]{McNamara06}\cite[p. 161-162]{Papageorgakis06}.

\end{minipage}
\hspace{3pt}
\scalebox{.82}{

%\hspace{-5cm}
}
%\\
%&
%\vspace{2cm}
\\
&
\\
&
\begin{minipage}[r]{8cm}
\footnotesize
{\bf Figure 3 -- Fuzzy funnel geometry}
of $\mathrm{D}p \!\!\perp\!\! \mathrm{D}(p+2)$-brane intersections
with fuzzy 2-sphere cross-sections $S^2_N$,
encoded by $\mathfrak{su}(2)_{\mathbb{C}}$-representations
$\mathfrak{su}(2)_{\mathbb{C}} \otimes \mathbf{N}
\overset{\rho}{\to} \mathbf{N}$.
\end{minipage}
%\hspace{2cm}
%\\
%\vspace{.1cm}
%\end{tabular}

\end{tabular}

\medskip

\noindent {\bf Enhancement to $\mathfrak{gl}(2,\mathbb{C})$-weight systems.}
Adjoining a commuting element \eqref{GaugeFieldBoundaryCondition}
to the $\mathfrak{su}(2)_{\mathbb{C}}$-representation $\rho$
\eqref{SolutionOfNahmEquation} equivalently means to extend the
representation to a $\mathfrak{gl}(2,\mathbb{C})$-representation along
the canonical inclusion $\iota$.
\begin{equation}
  \label{ExtensionAndRestrictionBetweensl2Andgl2}
  \xymatrix@C=34pt@R=6pt{
    \mathllap{
      \mathfrak{sl}(2,\mathbb{C})
      \;\simeq
      \;\;
    \mathfrak{su}(2)_{\mathbb{C}}
    }
    \;\; \ar@{^{(}->}[rr]^-{ \iota = (\mathrm{id},0) }
    &&
  \;\;   \mathfrak{su}(2)_{\mathbb{C}}
    \oplus
    \mathbb{C}
    \mathrlap{
      \;\;
      \simeq
      \;
      \mathfrak{gl}(2,\mathbb{C})
    }
    \\
    \underset{
      \mathclap{
      \mbox{
        \tiny
        \color{blue}
        \begin{tabular}{c}
          values of scalar fields
          \\
          at $\mathrm{D}p\!\!\perp\!\!\mathrm{D}(p+2)$
        \end{tabular}
      }
      }
    }{
    \underbrace{
    \overset{
      \mathclap{
        \overset{
          \mbox{
            \tiny
            \color{blue}
            \begin{tabular}{c}
              scalar fields
              \\
              on $\mathrm{D}p$
            \end{tabular}
          }
        }{
        \langle
          X_1, X_2, X_3
        \rangle
        }
      }
    }{
      \overbrace{
        \mathfrak{su}(2)_{\mathbb{C}}
      }
    }
    \mathrm{MetMod}
    }
    }
\;\;    \ar@<+10pt>@{<-}[rr]^-{
      0\;
      \mapsfrom
  \;      A_y
    }_-{ \iota^\ast }
    \ar@<-10pt>[rr]_-{ A_y := 1 }
    &&
    \underset{
      \mbox{
        \tiny
        \color{blue}
        \begin{tabular}{c}
          values of scalars \& gauge field
          \\
          at $\mathrm{D}p \!\!\perp\!\! D(p+2)$
        \end{tabular}
      }
    }{\;\;\;
    \underbrace{
    \big(
      \mathfrak{su}(2)_{\mathbb{C}}
      \oplus
      \overset{
        \mathclap{
          \overset{
            \mbox{
              \tiny
              \color{blue}
              \begin{tabular}{c}
                gauge field
                \\
                on $\mathrm{D}p$
              \end{tabular}
            }
          }{
            \langle A_y \rangle
          }
        }
      }{
        \overbrace{
          \mathbb{C}
        }
      }
    \big)
    \mathrm{MetMod}
    }
    }
  }
\end{equation}
Such extensions always exist, the canonical one
being given by choosing for $A_y$ the identity element.
With this choice, the fundamental (defining) representation
$\mathbf{2}$ of
$\mathfrak{sl}(2,\mathbb{C})$ is extended to the fundamental
(defining) representation of $\mathfrak{gl}(2,\mathbb{C})$.

\medskip

\noindent {\bf Fuzzy funnel states as weight systems.}
In summary, the invariant radius fluctuation observables on fuzzy funnel
$\mathrm{D}p \!\perp\! \mathrm{D}(p+2)$-configurations are
encoded by chord diagrams and given by evaluating chord diagrams in
$\mathfrak{gl}(2,\mathbb{C})$-weight systems.
(While above we discussed only round chord diagrams and single-trace
observables, this generalizes to horizontal chord diagrams and
multi-trace observables, see \cref{MatrixModelObservables} below).

\medskip
Conversely,
$\mathfrak{gl}(2,\mathbb{C})$-weight systems
thus reflect exactly the invariantly observable content of
fuzzy funnel $\mathrm{D}p \!\perp\! \mathrm{D}(p+2)$-brane intersections,
hence their states. But, by Prop. \ref{FundamentalTheremOfWeightSystems},
\emph{all} weight systems may be identified with
 $\mathfrak{gl}(2,\mathbb{C})$-weight systems,
and hence with states of fuzzy funnel
$\mathrm{D}p \!\perp\! \mathrm{D}(p+2)$-brane intersections.

\medskip
Below in \cref{LieAlgebraWeightSystemsEncoded3Algebras},
\cref{MatrixModelObservables} and \cref{LargeNLimitM2M5BraneBoundStates}
we find the analogous statement from the point of
view of the M-theory lift of such intersections to
$\mathrm{M2}/\mathrm{M5}$-brane bound states.

\newpage

%%%%%%%%%%%%%%%%%%%%%%%%%%%%%%%%%%%%%%%%%%%%%%%%%%%%%%%%%%%%%%%%%%%%
\subsection{Lie algebra weight systems encode M2-brane 3-algebras}
\label{LieAlgebraWeightSystemsEncoded3Algebras}
%%%%%%%%%%%%%%%%%%%%%%%%%%%%%%%%%%%%%%%%%%%%%%%%%%%%%%%%%%%%%%%%%%%%

The M-theory lift of fuzzy funnel $\mathrm{D}p \!\perp\! \mathrm{D}(p+2)$-brane
intersections (\cref{su2WeightSystemsAreFuzzyFunnelObservables})
to $\mathrm{M}2\perp\mathrm{M5}$-brane intersections
has famously been argued \cite{BasuHarvey05}\cite{BaggerLambert06}
(review in \cite[2.2]{BLMP12}) to involve generalization of
the $\mathfrak{su}(2)_{\mathbb{C}}$ Lie bracket, a binary operation,
to a trinary ``BLG 3-algebra'' structure.
At the same time M/TypeII-duality
requires that observables on M-brane intersections are equivalent
to those on the corresponding D-brane intersections.

We now observe that when identifying D-brane intersections
with Lie algebra weight systems on chord diagrams (as in \cref{su2WeightSystemsAreFuzzyFunnelObservables}) then the BLG 3-algebras
indeed emerge as the fundamental building blocks of the corresponding observables, namely as their value on single chords.
(The full M2/M5-brane states emerge below in \cref{LargeNLimitM2M5BraneBoundStates}.)

\medskip

\noindent{\bf The value of a weight system on a single chord.}
The value of a Lie algebra weight system \eqref{LieAlgebraWeightSystem}
on a chord diagram is a contraction of many copies of the one tensor
assigned to a single chord, according to \cref{LieAlgebraWeightSystems}:

\vspace{-.5cm}

\begin{equation}
\label{SingleChordEvaluation}
\raisebox{-0cm}{

    % }
     node[above, near start] {\tiny \color{black} \raisebox{-4pt}{$\mathfrak{g}$}}
   (-90-110:1.7);

  \begin{scope}

  \draw[fill=black] (-35-5:1.7) circle (.02);
  \draw[fill=black] (-30-5:1.7) circle (.02);
  \draw[fill=black] (-25-5:1.7) circle (.02);

  \draw[fill=black] (-145+5:1.7) circle (.02);
  \draw[fill=black] (-150+5:1.7) circle (.02);
  \draw[fill=black] (-155+5:1.7) circle (.02);

  \draw[darkblue, ultra thick] (-45-5:1.7) arc (-45-5:-135+5:1.7);

  \draw (-1.1,-1.6) node {\small $V$};
  \draw (+1.1,-1.6) node {\small $V$};

  \begin{scope}[shift={(0,-1.7)}]
    \draw[draw=orangeii, fill=orangeii] (0,0) circle (.07);
   \draw[draw=darkblue, thick, fill=white]
     (-.15,-.13)
     rectangle node{\tiny $\rho$}
     (.15,.17);
  \end{scope}

  \end{scope}

  \begin{scope}[rotate=(-110)]

  \draw[fill=black] (-35-5:1.7) circle (.02);
  \draw[fill=black] (-30-5:1.7) circle (.02);
  \draw[fill=black] (-25-5:1.7) circle (.02);

  \draw[fill=black] (-145+5:1.7) circle (.02);
  \draw[fill=black] (-150+5:1.7) circle (.02);
  \draw[fill=black] (-155+5:1.7) circle (.02);

  \draw[darkblue, ultra thick] (-45-5:1.7) arc (-45-5:-135+5:1.7);

  \draw (-1.1,-1.6) node {\small $V$};
  \draw (+1.1,-1.6) node {\small $V$};

  \begin{scope}[shift={(0,-1.7)}]
    \draw[draw=orangeii, fill=orangeii] (0,0) circle (.07);
   \draw[draw=darkblue, thick, fill=white]
     (-.15,-.13)
     rectangle node{\tiny $\rho$}
     (.15,.17);
  \end{scope}

  \end{scope}

  \end{scope}

\end{tikzpicture}
}
&
\scalebox{1}{
  $\rho_a{}^m{}_l\,\rho^a{}_{j i}$
}
\\
&&&
\\
\hline
&&&
\\

}
&
\scalebox{1}{
  $\rho_a{}^m{}_l \, \rho^a{}_{}^j{}_i$
}
\\
&&&
\\
\hline
\end{tabular}
\end{tabular}
}
\end{equation}

\vspace{0cm}

\noindent {\bf M2-brane 3-algebras.}
As shown on the right, this tensor assigned to a single chord
is exactly the tensor considered in \cite[above Prop. 10 \& (22)]{dMFMR09} (for
the first case in the above table) or in \cite[(34)]{dMFMR09} (for the second case).
By \cite[Prop. 10]{dMFMR09} these tensors are the 3-brackets constituting
generalized BLG 3-algebras
\cite{BaggerLambert06}\cite[4]{CherkisSaemann08}\cite[3]{BLMP12}
and the Faulkner construction \cite{Faulkner73}, respectively, as shown above.
In fact, by \cite[Theorem 11]{dMFMR09} this construction constitutes a bijective
equivalence between (generalized real) BLG 3-algebras and metric Lie representations.
Hence all 3-algebras come from weight systems on chord diagrams.

\newpage

%%%%%%%%%%%%%%%%%%%%%%%%%%%%%%%%%%%%%%%%%%%%%%%%%%%%%%%%%
\subsection{Round weight systems are 3d gravity observables}
\label{RoundWeightSystemsExhibitHolography}
%%%%%%%%%%%%%%%%%%%%%%%%%%%%%%%%%%%%%%%%%%%%%%%%%%%%%%%%%

We have seen in \cref{LieAlgebraWeightSystemsAreSingleTraceObservables}
and \cref{su2WeightSystemsAreFuzzyFunnelObservables} that
weight systems on round chord diagrams have the form
of observables on (fully compactified) worldvolume theories of branes,
where the circle in the chord diagram is what represents the
worldvolume.
Here we observe that generating functions of
weight systems dually encode Chern-Simons amplitudes that
may be thought of as propagating in a bulk spacetime away from these
brane worldvolumes.

\hspace{-.95cm}
\begin{tabular}{ll}
\begin{minipage}[l]{9.5cm}
Write
$(\mathcal{G}^\bullet, d_{\mathcal{G}})$
for Kontsevich's graph complex in its version
with framed Wilson loops \cite[5]{Kontsevich93}\cite[2]{AF96}\cite[7]{CCRL02}.
This is spanned by the graded set
of Jacobi diagrams \eqref{JacobiDiagrams} together with analogous
sets of graphs of valence higher than 3; the differential acts by
sending any graph to the signed formal linear combination of the
results of contracting any edge that is not a chord.
We may then regard any Jacobi diagram \eqref{JacobiDiagrams}
both as an element of
$\mathcal{A} := \mathcal{A}^{{}^{\mathrm{t}}} \simeq \mathcal{A}^{\mathrm{c}}$ \eqref{RoundChordDiagramsModulo4TIsJavcobidDiagramsModuloSTU}
as before, or as an element of the graph complex,
as shown on the right.
\end{minipage}
&
\hspace{.6cm}
$
  \raisebox{30pt}{
  \xymatrix@C=7pt@R=7pt{
    &
    \overset{
      \mbox{
        \tiny
        \color{blue}
        \begin{tabular}{c}
          Graded set of
          \\
          Feynman diagrams for
          \\
          Chern-Simons theory
          \\
          with a Wilson loop
          \\
          $\phantom{a}$
        \end{tabular}
      }
    }{
      \mathcal{D}^{{}^{\mathrm{t}}}
    }
    \ar[dl]_{
      \mathllap{
        \mbox{
          \tiny
          \color{blue}
          \begin{tabular}{c}
          \end{tabular}
        }
      }
      [-]_{{}_{\mathcal{A}}}
    }
    \ar[dr]^{[-]_{{}_{\mathcal{G}}}}
    \ar@{}[d]|-{
      \mbox{
        \tiny
        \color{blue}
        \begin{tabular}{c}
          $\phantom{a}$
          \\
          send graph to
          \\
          the respective vector
          \\
          that it represents
        \end{tabular}
      }
    }
    \\
    \underset{
      \mathclap{
      \mbox{
        \tiny
        \color{blue}
        \begin{tabular}{c}
          $\phantom{a}$
          \\
          Graded linear span of
          \\
          Jacobi diagrams
          \\
          modulo STU-relations
        \end{tabular}
      }
      }
    }{
      \mathcal{A}_\bullet
    }
    &&
    \underset{
      \mathclap{
      \mbox{
        \tiny
        \color{blue}
        \begin{tabular}{c}
          $\phantom{a}$
          \\
          Kontsevich graph complex
          \\
          for framed knots
        \end{tabular}
      }
      }
    }{
      \mathcal{G}^\bullet
    }
  }
  }
$
\end{tabular}

\noindent Using \cite[Thm. 1]{AF96} we find that the sum
over Jacobi diagrams of the tensor product of these two
of their vector incarnations is a graph cocycle with
coefficients in Jacobi diagrams in the linear dual
of round weight systems \eqref{SpaceOfRoundWeightSystems}

\vspace{-.2cm}

\begin{equation}
  \label{UniversalWilsonLoopObservable}
  \hspace{-1cm}
  \begin{aligned}
  \overset{
    \mbox{
      \tiny
      \color{blue}
      \begin{tabular}{c}
        Universal Wilson loop observable
        \\
        $\phantom{a}$
      \end{tabular}
    }
  }{
  \mbox{
  $
  \left[
  \scalebox{.8}{
    $
    (w, \mathcal{K})
    \mapsto
    $
  }
  \left\langle
    \mathrm{Tr}_{{}_w}
    \mathrm{P}
    \!
    \exp
    \!
    \left(
      \int_{{}_{\scalebox{.5}{$\mathcal{K}$}}} \!A
    \right)
    \!
  \right\rangle
  \,
  \right]
  $
  }
  }
  \;
  :=
  \;\;
  \underset{
    n \in \mathbb{N}
  }{\sum}
  \,\,
  \hbar^n
  \;\;\;
  \underset{
    \mathclap{
    {\Gamma \in}
    {
      (\mathcal{D}^{{}^{\mathrm{t}}})_n
    }
    }
  }{\sum}
 \; \left(
    \tfrac{1}{
      \left\vert
        \mathrm{Aut}(\Gamma)
      \right\vert
    }
    \,
    [\Gamma]_{{}_{\mathcal{A}}} \otimes [\Gamma]_{{}_{\mathcal{G}}}
  \right)
  \;\;\;
  \in
  \;
  \;
  &
  \;\;
  \overset{
    \mathclap{
    \mbox{
      \tiny
      \color{blue}
      \begin{tabular}{c}
        0-cohomology of
        \\
        graph complex with
        \\
        values in Jacobi diagrams
        \\
        $\phantom{a}$
      \end{tabular}
    }
    }
  }{
    H^0
    \big(
      \mathcal{A}_\bullet
      \otimes
      \,
      \mathcal{G}^\bullet
    \big)
  }
  \\
  \simeq
  \;\;
  &
  \underset{
    \mathclap{
    \mbox{
      \tiny
      \color{blue}
      \begin{tabular}{c}
        $\phantom{a}$
        \\
        Graded linear maps from
        \\
        weight systems on chord diagrams
        \\
        to graph cohomology
      \end{tabular}
    }
    }
  }{
    \mathrm{Hom}
    \big(
      \mathcal{W}^\bullet
      ,
      H^\bullet(\mathcal{G})
    \big)
  }
  \end{aligned}
\end{equation}
%\vspace{-.5cm}
%\noindent
%\begin{minipage}[l]{12cm}
  which, dually, is a graded linear map, as shown on the right.
%\end{minipage}
Since \eqref{UniversalWilsonLoopObservable}
is a \emph{universal Vassiliev invariant} \cite[Thm. 1]{AF96}
(following \cite[Thm. 2.3]{Kontsevich93}\cite[4.4.2]{BarNatan95},
reviewed in \cite[8.8]{CDM11}\cite[18]{JacksonMoffat19} )
this map in fact identifies weight systems on round chord diagrams \eqref{SequenceOfMapsOfWeightSystems}
with the space of Vassiliev knot invariants \cite{Vassiliev92}
via identification (see \cite[Prop. 7.6 using Thm. 7.3]{CCRL02})
with the graph cohomology spanned by trivalent graphs
$$
  \mathcal{W}^\bullet
  \overset{\simeq}{\longrightarrow}
  H^\bullet(\mathcal{G})_{{}_{
\!\!\!\!\!\!\!
\raisebox{-2pt}{
\scalebox{.11}{
\begin{tikzpicture}
  \begin{scope}[rotate=(-40)]
    \draw[draw=orangeii, fill=orangeii] (0,0) circle (.08);
    \draw[orangeii, ultra thick] (0:1) to (0,0);
    \draw[orangeii, ultra thick] (120:1) to (0,0);
    \draw[orangeii, ultra thick] (240:1) to (0,0);
  \end{scope}
\end{tikzpicture}
}
}
}
}
 \!\!
  \subset
  H^\bullet(\mathcal{G})
  \,.
$$

\noindent {\bf Dual Chern-Simons theory appears.}
But the  graph complex is what organizes
Feynman diagrams for perturbative Chern-Simons theory in the presence of
a framed Wilson loop \cite{BarNatan91}\cite{BarNatan95CS}.
The construction
\cite[p. 11-12]{Kontsevich92}\cite{AxelrodSinger93}\cite[3]{AF96}
of Feynman amplitudes as configuration space fiber integrals of
wedge products of Chern-Simons propagators,
regarded as differential forms on the configuration space of points,
(see Def. \ref{ConfigurationSpaces}) sends graph cocycles to the higher
observables of Chern-Simons theory with a Wilson loop
(reviewed in \cite[3-4]{Volic13}).
This is given by evaluating the corresponding Lie algebra
weight system \eqref{LieAlgebraWeightSystemOnHoizontalChordDiagramsWithStacksOfCoincidentStrands}
(restricted to round chord diagrams via \eqref{SequenceOfMapsOfWeightSystems})
\vspace{-4mm}
\begin{equation}
  \label{LieAlgebraWeightSystemsOnRoundChordDiagrams}
  \xymatrix{
    \overset{
      \mathclap{
      \mbox{
        \tiny
        \color{blue}
        \begin{tabular}{c}
          Metric Lie representations
          \\
          $\phantom{a}$
        \end{tabular}
      }
      }
    }{
    \mathcal{L}
    \;:=\;
    \mathrm{Span}
    \big(
      \mathrm{MetLieMod}_{/\sim}
    \big)
    }
    \ar[rrr]^-{
      \mathclap{
      \mbox{
        \tiny
        \color{greenii}
        \begin{tabular}{c}
          Lie algebra weight systems
        \end{tabular}
      }
      }
    }
    &&&
    \overset{
      \mbox{
        \tiny
        \color{blue}
        \begin{tabular}{c}
          Chordal
          \\
          Vassiliev invariants
          \\
          $\phantom{a}$
        \end{tabular}
      }
    }{
    \underset{n \in \mathbb{N}}{\prod}
    \big(
      \underset{
        \mathclap{
        \mbox{
          \tiny
          \color{blue}
          \begin{tabular}{c}
          \\
            Round
            \\
            degree $n$
            \\
            weight systems
          \end{tabular}
        }
        }
      }{
        \mathcal{W}^n
      }
      \langle \hbar^n \rangle
    \big)
    }
  }
\end{equation}

\vspace{-2mm}
\noindent on the universal Wilson loop observable \eqref{UniversalWilsonLoopObservable}, thus multiplying the
bare Chern-Simons amplitude with the Lie algebraic contraction
and tracing of the gauge indices.
Here on the right of \eqref{LieAlgebraWeightSystemsOnRoundChordDiagrams}
we recognize the space of generating functions of round weight systems
\eqref{SpaceOfRoundWeightSystems}
as that of Vassiliev knot invariants, via the
\emph{fundamental theorem of Vassiliev invariants}
\cite[Thm. 2.3]{Kontsevich93}\cite[Thm. 1 (3)]{BarNatan95}
(reviewed in \cite[8.8]{CDM11}\cite[18]{JacksonMoffat19}).

\medskip
In summary, we thus find that,
via the Chern-Simons Wilson loop observable,
the generating functions of weight systems on round chord diagrams are equivalently
the perturbative quantum observables of Chern-Simons theory with a Wilson loop knot $\mathcal{K}$:

\vspace{-1.5cm}

\begin{equation}
  \label{UniversalWilsonLoopAsAMap}
  \hspace{-.5cm}
  \raisebox{-50pt}{
  \xymatrix@R=-10pt@C=18pt{
    \overset{
      \mathclap{
      \mbox{
        \tiny
        \color{blue}

         }
       }
       $
     \end{rotate}
   };

\end{tikzpicture}
}
\end{tabular}

\medskip
\medskip
\noindent {\bf Holographic wrapped 5-branes appear.} Now consider instead the case that the
knot $\mathcal{K}$ in \eqref{UniversalWilsonLoopAsAMap} is
a hyperbolic knot \cite{FKP17}, hence such that its complement $S^3 \setminus \mathcal{K}$
carries the structure of a hyperbolic space with \emph{finite} volume,
then necessarily unique, by Mostow rigidity \cite{Mostow68}
(reviewed in \cite{Bourdon18}).
In this case the \emph{volume conjecture} asserts
\cite{Kashaev96}\cite{MurakamiMurakami01} (reviewed in \cite{Murakami10})
that the Wilson loop observables \eqref{UniversalWilsonLoopAsAMap}
for the $N$-dimensional irreducible representation
of $\mathfrak{su}(2)_{\mathbb{C}}$
tends in the large $N$ limit, $N \to \infty$, to that finite volume.
Moreover, the \emph{3d-3d correspondence} (see the review \cite{Dimofte14}) asserts that
the Wilson loop observables \eqref{UniversalWilsonLoopAsAMap}
are dually observables on the worldvolume theory of M5-branes
wrapped on $\Sigma^3 := S^3 \setminus \mathcal{K}$. Furthermore,
with this identification the statement of the volume conjecture
is part of the statement of holographic AdS/CFT duality for such configurations \cite[3.2]{GangKimLee14b} (see also \cite{BGL16}).
We thus have the following web of relations connecting to \hyperlink{HypothesisH}{Hypothesis H}:

\vspace{-5mm}
$$
\hspace{2cm}
  \xymatrix@C=20pt@R=8pt{
  &&
    \underset{
    \mbox{
      \tiny

      }
    }{
      \mathrm{SCFT}_3[\Sigma^3]
    }
    \\
    \\
    &&&
    \fbox{\small \hyperlink{HypothesisH}{Hypothesis H}}
    \ar@{~>}[uu]_-{ \mbox{ \tiny \eqref{UniversalWilsonLoopAsAMap} } }
  }
$$

\newpage

%%%%%%%%%%%%%%%%%%%%%%%%%%%%%%%%%%%%%%%%%%%%%%%%%%%%%%%%%%
\subsection{Round weight systems contain supersymetric indices}
\label{RoundWeightSystemsEncodeMonopoleModuli}
%%%%%%%%%%%%%%%%%%%%%%%%%%%%%%%%%%%%%%%%%%%%%%%%%%%%%%%%%%%

We observe here that round weight systems encode the
Witten indices of $D= 3, \mathcal{N}=4$ super Yang-Mills theories,
computing the $\widehat A$-genus of
Coulomb branches of intersecing branes given by
Atiyah-Hitchin moduli space of Yang-Mills monopoles.

\medskip

\noindent
{\bf Coulomb branches of $D=3, \mathcal{N}=3$ SYM and monopole moduli.}
The worldvolume gauge theory of $\mathrm{D}p\perp \mathrm{D}(p+2)$-brane
intersections is thought to be $D=3, \mathcal{N}=4$ super Yang-Mills theory,
at least for $p = 3$ \cite{HananyWitten97}.
The moduli spaces of vacua of
$D=3$ $\mathcal{N}=4$ super Yang-Mills theory,
both the Coulomb branches and the Higgs branches,
are hyperk{\"a}hler manifolds $\mathcal{M}^{4n}$
\cite{SeibergWitten96} (see, e.g., \cite{dBHOO97}),
which are either
\begin{enumerate}[{\bf (1)}]
\vspace{-2mm}
\item  asymptotically flat (ALE-spaces) and dual to branes transversal to ADE-singularities;
\vspace{-2mm}
\item or compact and dual to branes transversal to a K3 surfaces or to a 4-torus $\mathbb{T}^4$.
\end{enumerate}
\vspace{-1mm}
Specifically, the (classical) Coulomb branches of these theories
are the Atiyah-Hitchin moduli spaces of Yang-Mills monopoles \cite{AtiyahHitchin88}
on the transversal space \cite{DKMTV97}\cite{Tong99}\cite{BullimoreDimofteGaiotto15},
which are often identified with Hilbert schemes of points
\cite{dBHOOY96}\cite{dBHOO97}\cite[(4.4)]{CHZ14}.

\medskip
In particular, if the transversal space is a
K3 surface $\Sigma^4_{\mathrm{K3}}$, then the corresponding moduli space
is the Hilbert scheme of points
$\mathcal{M}^{4n} = (\Sigma^4_{\mathrm{K3}})^{[n]}$
\cite{VafaWitten94}\cite{Vafa96}, which is an
example of a \emph{compact} hyperk{\"a}hler manifold. In fact, all
known examples of compact hyperk{\"a}hler manifolds are Hilbert schemes
either of K3 surfaces or of the 4-torus \cite{Beauville83}, with two exceptional variants found in \cite{OGrady98}\cite{OGrady00}
(reviewed in \cite[5.3]{Sawon04}).
These compact Coulomb branches come from $D=3, \mathcal{N}=4$ SYM theories that are obtained by KK-compactification of little string theories
\cite{Intriligator99}.

\medskip

\noindent {\bf Rozansky-Witten theory.}
The topological C-twist of $D = 3$ $\mathcal{N}=4$ SYM is
Rozansky-Witten theory \cite{RozanskyWitten97},
which, after gauge fixing and
suitable field identifications, turns out to have same Feynman rules
as 3d Chern-Simons theory. This is in the sense that the only relevant propagator
is the Chern-Simons propagator, and the only relevant Feynman diagrams
are trivalent, the only difference being that the Lie algebra weights
of Chern-Simons theory are replaced by \emph{Rozansky-Witten weights}
\cite[3]{RozanskyWitten97}\cite{RobertsWillerton06}
which depend (only) on the hyperk{\"a}hler moduli space
$\mathcal{M}^{4n}$, and in fact only on
its underlying holomorphic symplectic manifold \cite{Kapranov99}.

\medskip
Hence the assignment of Rozansky-Witten weights is a linear map
from the linear span of the set of isomorphism classes of such
gauge theories
\begin{equation}
  \label{RozansyWittenWeightSystemAssignment}
  \xymatrix{
    \overset{
      \mathclap{
      \mbox{
        \tiny
        \color{blue}
        \begin{tabular}{c}
          $D= 3$, $\mathcal{N}=4$
          gauge theories
          \\
          $\phantom{a}$
        \end{tabular}
      }
      }
    }{
      \mathcal{G}
      :=
      \mathrm{Span}
      \big(
        \mathrm{SYM}^{D = 3, \mathcal{N}=4}_{/\sim}
      \big)
    }
    \ar[rrrr]^-{
      \mathclap{
      \mbox{
        \tiny
        \color{greenii}
        \begin{tabular}{c}
          $\phantom{a}$
          \\
          Rozansky-Witten
          weight systems
        \end{tabular}
      }
      }
    }
    &&&&
    \overset{
      \mathclap{
      \mbox{
        \tiny
        \color{blue}
        \begin{tabular}{c}
          Chordal
          \\
          Vassiliev invariants
          \\
          $\phantom{a}$
        \end{tabular}
      }
      }
    }{
    \underset{n \in \mathbb{N}}{\prod}
    \big(
      \underset{
        \mathclap{
        \mbox{
          \tiny
          \color{blue}
          \begin{tabular}{c}
            \\
            Round
            \\
            degree $n$
            \\
            weight systems
          \end{tabular}
        }
        }
      }{
        \mathcal{W}^n
      }
      \langle \hbar^n\rangle
    \big)
    }
  }
\end{equation}
directly analogous to the assignment of
Lie algebra weight systems \eqref{LieAlgebraWeightSystemsOnRoundChordDiagrams}.
Furthermore, the Wilson loop observables of
Rozansky-Witten C-twisted $D = 3, \mathcal{N}=4$ super Yang-Mills theory
are obtained by evaluating these weights on the universal Vassiliev
Wilson loop observable, in direct analogy to the Wilson loop
observables \eqref{UniversalWilsonLoopAsAMap} of Chern-Simons theory:

\begin{equation}
  \label{EvaluatingRzanskyWittenWeightsOnUniversalWilsonLoopObservable}
  \hspace{.1cm}
  \xymatrix@R=-10pt@C=18pt{
    \overset{
      \mathclap{
      \mbox{
        \tiny
        \color{blue}
\begin{tabular}{c}
          $D = 3$, $\mathcal{N}=4$
          \\
          gauge theories
          \\
          $\phantom{a}$
        \end{tabular}
      }
      }
    }{
      \mathclap{{\phantom{\Big(}}}
      \mathcal{G}
    }
    \ar[rr]_-{
      \mbox{
        \tiny
        \color{greenii}
        \begin{tabular}{c}
          $\phantom{a}$
          \\
          $\phantom{a}$
          \\
          $\phantom{a}$
          \\
          Rozansky-Witten weight systems
        \end{tabular}
      }
    }
    &&
    \overset{
      \mathclap{
      \mbox{
        \tiny
        \color{blue}
        \begin{tabular}{c}
          Chordal
          \\
          Vassiliev invariants
          \\
          $\phantom{a}$
        \end{tabular}
      }
      }
    }{
    \underset{n}{\prod}
    \Big(
    \underset{
      \mathclap{
      \mbox{
        \tiny
        \color{blue}
        \begin{tabular}{c}
          Weight
          \\
          systems
        \end{tabular}
      }
      }
    }{
      \mathcal{W}^n
      \langle \hbar^n \rangle
    }
    }
    \Big)
    \ar[rrrr]_-{
    \underset{
      \mbox{
        \tiny
        \color{blue}
        \begin{tabular}{c}
          $\phantom{a}$
          \\
          Universal Wilson loop observable
        \end{tabular}
      }
    }{
  \scalebox{.7}{
  $
  \left[
  \!\!
  \scalebox{.8}{
    $
    (w, \mathcal{K})
    \mapsto
    $
  }
  \!\!
  \left\langle
    \!
    \mathrm{Tr}_{{}_w}
    \mathrm{P}
    \!
    \exp
    \!
    \left(
      \int_{{}_{\scalebox{.5}{$\mathcal{K}$}}}
      \!A
    \right)
    \!
  \right\rangle
  \!
  \right]
  $
  }
  }
    }^-{
      \simeq
    }
    &&&&
    \overset{
      \mathclap{
      \mbox{
        \tiny
        \color{blue}
        \begin{tabular}{c}
          Graphical
          \\
          Vassiliev invariants
          \\
          $\phantom{a}$
        \end{tabular}%
    }
    }
  }{
  \left\langle
    \mathrm{Tr}
    \mathrm{P}
    \!
    \exp
    \!
    \left(
      \int_{{}_{\scalebox{.5}{$\mathcal{K}$}}}
      \!
      (\Gamma + \Omega)
    \right)
    \!
  \right\rangle_{T}
  }
  \,
  \!\!\!\Big).
  \;\;\;\;\;\;
  }
\end{equation}

\newpage

\noindent {\bf The index of $D=3$, $\mathcal{N}=4$ SYM.}
In the case that the knot $\mathcal{K} = \bigcirc$ in \eqref{EvaluatingRzanskyWittenWeightsOnUniversalWilsonLoopObservable}
is the unknot, the Rozansky-Witten Wilson loop observable
\eqref{EvaluatingRzanskyWittenWeightsOnUniversalWilsonLoopObservable}
computes the square root of the
$\widehat A$-genus of the moduli space $\mathcal{M}^4n_T$ of the
given C-twisted $D=3, \mathcal{N}=4$ SYM theory $T$
(\cite[Lem. 8.6]{RobertsWillerton06},
using the wheeling theorem \cite{BNTT03} and Hitchin-Sawon theorem
\cite{HitchinSawon99}):
$$
  \left\langle
    \mathrm{Tr}
    \left(
    \mathrm{P}\exp\big(
      \int_{\bigcirc}
      (\Gamma + \Omega)
    \big)
    \right)
  \right\rangle_T
  \;=\;
  \sqrt{
    \widehat{A}(\mathcal{M}^{4n}_T)
  }
  \,.
$$
This genus is part of the expression of the
Witten index of the theory $T$ \cite{BFK18}.

\medskip

%%%%%%%%%%%%%%%%%%%%%%%%%%%%%%%%
\begin{observation}[\bf Dualities]
\label{Dualities}
%%%%%%%%%%%%%%%%%%%%%%%%%%%%%%%%
From the point of view of
\hyperlink{HypothesisH}{\it Hypothesis H},
the genuine observables on the brane configurations are
the abstract weight systems in
$\underset{n\in \mathbb{N}}{\prod}\big(
  \mathcal{W}^n \langle \hbar^n \rangle
\big)$, by Prop. \ref{HigherObservablesEquivalentToCohomologyOfLoopedConfigurationSpace}.
One may then ask which physics is compatible with
these observables, much like one asks which target space geometry
emerges from a given worldsheet CFT.
We saw  in \cref{RoundWeightSystemsExhibitHolography}
and \cref{RoundWeightSystemsEncodeMonopoleModuli}
that a range of quantum field theories has these weight systems as
their observables, including Chern-Simons theories and
Rozansky-Witten C-twisted 3d super Yang-Mills theories. These are reflected
by canonical maps \eqref{LieAlgebraWeightSystemsOnRoundChordDiagrams}
and  \eqref{RozansyWittenWeightSystemAssignment}
from the spaces of these theories into
the space of observables:
$$
  \xymatrix{
    \overset{
      \mathclap{
      \mbox{
        \tiny
        \color{blue}
\begin{tabular}{c}
          Chern-Simons \& Rozansky-Witten
          \\
          field theories
          \\
          $\phantom{a}$
        \end{tabular}
      }
      }
    }{
      \mathcal{L}
      \oplus
      \mathcal{G}
    }
    \;\;\;\;\;\;
    \ar[rrrrr]^-{
      \mbox{
        \tiny
        \color{blue}
        \begin{tabular}{c}
          $\phantom{a}$
          \\
          Extract observables
        \end{tabular}
      }
    }_-{
      \mbox{
        \tiny
        \begin{tabular}{c}
          \color{greenii} Lie algebra
          \\
          \color{greenii} weight systems
          \\
          \eqref{LieAlgebraWeightSystemsOnRoundChordDiagrams}
        \end{tabular}
      }
      \!\!
        \raisebox{3pt}{
          $\oplus$
        }
      \!\!
      \mbox{
        \tiny
        \begin{tabular}{c}
          \color{greenii} Rozansky-Witten
          \\
          \color{greenii} weight systems
          \\
          \eqref{RozansyWittenWeightSystemAssignment}
        \end{tabular}
      }
    }
    &&&&&
    \;\;
    \overset{
      \mathclap{
      \mbox{
        \tiny
        \color{blue}
        \begin{tabular}{c}
          Chordal
          \\
          Vassiliev invariants
          \\
          $\phantom{a}$
        \end{tabular}
      }
      }
    }{
      \underset{n \in \mathbb{N}}{\prod}
      \big(
        \mathcal{W}^n\langle \hbar^n \rangle
      \big)
    }
    \;\;
    \ar@{^{(}->}[rr]_-{
      \mbox{
        \tiny
        \begin{tabular}{c}
          $\mathrm{comp}^\ast$
          \\
          \eqref{SequenceOfMapsOfWeightSystems}
        \end{tabular}
      }
    }
    &&
    \overset{
      \mathclap{
      \mbox{
        \tiny
        \color{blue}
        \begin{tabular}{c}
          Higher observables
          \\
          on D6-D8-brane intersections
          \\
          by \hyperlink{HypothesisH}{\it Hypothesis H}
          \\
          $\phantom{a}$
        \end{tabular}
      }
      }
    }{
      \underset{n \in \mathbb{N}}{\prod}
      \Big(
        (\mathcal{W}^{{}^{\mathrm{pb}}})^n
        \langle
          \hbar^n
        \rangle
      \Big).
    }
  }
$$
But this operation of extracting observables
from field theories has a large kernel, equivalently a
non-trivial fiber product
\begin{equation}
  \label{PullbackDualPairs}
  \xymatrix{
    &
    \fbox{
      \color{greenii}
      Dual pairs
    }
    \ar[dl]
    \ar[dr]
    \ar@{}[dd]|-{\mbox{\tiny (pb)}}
    \\
    \mathllap{
      \mbox{
        \tiny
        \color{blue}
        Field theories
      }
      \;\;
    }
    \mathcal{L}\oplus \mathcal{G}
    \ar[dr]_-{
      \mbox{
        \tiny
        \color{blue}
        \begin{tabular}{c}
          extract
          \\
          observables
        \end{tabular}
      }
    }
    &&
    \mathcal{L}\oplus \mathcal{G}
    \ar[dl]^-{
      \mbox{
        \tiny
        \color{blue}
        \begin{tabular}{c}
          extract
          \\
          observables
        \end{tabular}
      }
    }
    \mathrlap{
      \;\;
      \mbox{
        \tiny
        \color{blue}
        Field theories
      }
    }
    \\
    &
    \underset{
      \mbox{
        \tiny
        \color{blue}
        \begin{tabular}{c}
          $\phantom{a}$
          \\
          Space of
          \\
          higher observables
          \\
          by \hyperlink{HypothesisH}{\it Hypothesis H}
        \end{tabular}
      }
    }{
      \underset{n \in \mathbb{N}}{\prod}
      \big(
        \mathcal{W}^n\langle \hbar^n \rangle
      \big)
    }
  }
\end{equation}
corresponding to
different gauge theories which have \emph{indistinguishable observables},
hence which are physical duals. We thus see that
\hyperlink{HypothesisH}{\it Hypothesis H},
not only sees the genuine observables on the brane configurations as
the abstract weight systems but also encodes duality in the corresponding
field theories in a compatible manner.
\end{observation}

%\newpage

%%%%%%%%%%%%%%%%%%%%%%%%%%%%%%%%%%%%%%%%%%%%%%%%%%%%%%%%%%%%%%%%%
\subsection{Round weight systems encode 't Hooft string amplitudes}
\label{RoundWeightSystemsEncodeOpenStringAmplitudes}
%%%%%%%%%%%%%%%%%%%%%%%%%%%%%%%%%%%%%%%%%%%%%%%%%%%%%%%%%%%%%%%%%

We have seen in \eqref{RoundChordDiagramsModulo4TIsJavcobidDiagramsModuloSTU}  that
round chord diagrams modulo 4T relations are equivalently
Jacobi diagrams \eqref{JacobiDiagrams} modulo STU-relations,
and that weight systems \eqref{SequenceOfMapsOfWeightSystems} exhibit
the latter as the Feynman diagrams of Chern-Simons theory (\cref{RoundWeightSystemsExhibitHolography})
and of Rozansky-Witten theory
(\cref{RoundWeightSystemsEncodeMonopoleModuli}).
We now observe that Lie algebra weight systems (\cref{LieAlgebraWeightSystems}) also know about the
't Hooft double line reformulation \cite{tHooft74}
of these Feynman diagrams as well as about the resulting
identification of Chern-Simons amplitudes
with topological open string amplitudes
\cite{Witten92} (reviewed in \cite{Marino04}).

\medskip

\noindent {\bf 't Hooft double line notation.}
One observes that Lie algebra weight systems \eqref{LieAlgebraWeightSystemsOnRoundChordDiagrams}
for $\mathfrak{g}$ a semisimple Lie algebra and
$V$ its fundamental representation,
evaluate a single chord \eqref{SingleChordEvaluation}
to a linear combination of a \emph{double line} of strands,
in terms of the Penrose notation from \cref{LieAlgebraWeightSystems},
as shown in the following table:
\begin{equation}
  \label{MetricContractionOfFundamentalActions}
  \raisebox{0pt}{


  }
\end{equation}
Applying this iteratively on the right hand side of the Jacobi identity/Lie action property \eqref{ActionProperty}

\vspace{-.4cm}

$$
  \mbox{
  \raisebox{-40pt}{
%
  }
$$

\vspace{2mm}
\noindent identifies the corresponding Lie algebra weight of any
Jacobi/Feynman diagram with that of a linear combiantion of purely
\emph{double line diagrams}
where, in Feynman diagram language, all virtual gluon lines
have turned into double quark lines.
For example:

\begin{equation}
\label{HooftDoubleLineOfTrivalentVertex}
\raisebox{0pt}{
\begin{tabular}{cccccc}
\raisebox{-50pt}{
\begin{tikzpicture}

  \draw[thick, orangeii] (90+60:1.7) to (0,0);
  \draw[thick, orangeii] (90-60:1.7) to (0,0);
  \draw[thick, orangeii] (90-180:1.7) to (0,0);

  \draw[draw=orangeii, fill=orangeii] (0,0) circle (0.07);

  \draw[darkblue, ultra thick] (0,0) circle (1.7);

  \begin{scope}[rotate=(+60)]
  \begin{scope}[shift={(0,+1.7)}]
    \draw[draw=orangeii, fill=orangeii] (0,0) circle (.07);
   \draw[draw=darkblue, thick, fill=white]
     (-.15,-.13)
     rectangle node{\tiny $\rho$}
     (.15,.17);
  \end{scope}
  \end{scope}

  \begin{scope}[rotate=(-60)]
  \begin{scope}[shift={(0,+1.7)}]
    \draw[draw=orangeii, fill=orangeii] (0,0) circle (.07);
   \draw[draw=darkblue, thick, fill=white]
     (-.15,-.13)
     rectangle node{\tiny $\rho$}
     (.15,.17);
  \end{scope}
  \end{scope}

  \begin{scope}[rotate=(180)]
  \begin{scope}[shift={(0,+1.7)}]
    \draw[draw=orangeii, fill=orangeii] (0,0) circle (.07);
   \draw[draw=darkblue, thick, fill=white]
     (-.15,-.13)
     rectangle node{\tiny $\rho$}
     (.15,.17);
  \end{scope}
  \end{scope}

\end{tikzpicture}
}
&
$\;=\;$
&
$\phantom{3 \cdot}$
\raisebox{-50pt}{
\begin{tikzpicture}

  \draw[ultra thick, darkblue]
    (90+60-8:1.7) arc (90+60-8:90-60+8:1.7);
  \draw[ultra thick, darkblue]
    (90-60-8:1.7) arc (90-60-8:90-180+8:1.7);
  \draw[ultra thick, darkblue]
    (90-180-8:1.7) arc (90-180-8:90-180-120+8:1.7);

  \draw[ultra thick, darkblue]
    (90-180-8:1.7)
      .. controls (90-180-8+5:1.7) and (90-180-8+5:1.7) ..
    (90-180-8-12:.3);

  \draw[ultra thick, darkblue]
    (90-180+8:1.7)
      .. controls (90-180+8-5:1.7) and (90-180+8-5:1.7) ..
    (90-180+8+12:.3);

  \draw[ultra thick, darkblue]
    (90-180-8-12:.3)
    .. controls (90-180-60-8:.12) and (90-180-60+8:.12) ..
    (90-180+8+12-120:.3);

  \begin{scope}[rotate=(120)]
  \draw[ultra thick, darkblue]
    (90-180-8:1.7)
      .. controls (90-180-8+5:1.7) and (90-180-8+5:1.7) ..
    (90-180-8-12:.3);

  \draw[ultra thick, darkblue]
    (90-180+8:1.7)
      .. controls (90-180+8-5:1.7) and (90-180+8-5:1.7) ..
    (90-180+8+12:.3);

  \draw[ultra thick, darkblue]
    (90-180-8-12:.3)
    .. controls (90-180-60-8:.12) and (90-180-60+8:.12) ..
    (90-180+8+12-120:.3);

  \end{scope}

  \begin{scope}[rotate=(-120)]
  \draw[ultra thick, darkblue]
    (90-180-8:1.7)
      .. controls (90-180-8+5:1.7) and (90-180-8+5:1.7) ..
    (90-180-8-12:.3);

  \draw[ultra thick, darkblue]
    (90-180+8:1.7)
      .. controls (90-180+8-5:1.7) and (90-180+8-5:1.7) ..
    (90-180+8+12:.3);

  \draw[ultra thick, darkblue]
    (90-180-8-12:.3)
    .. controls (90-180-60-8:.12) and (90-180-60+8:.12) ..
    (90-180+8+12-120:.3);

  \end{scope}

\end{tikzpicture}
}

&
$\;-\;$
&
$3 \cdot$
\raisebox{-50pt}{
\begin{tikzpicture}

  \draw[ultra thick, darkblue]
    (90+60-8:1.7) arc (90+60-8:90-60+8:1.7);
  \draw[ultra thick, darkblue]
    (90-60-8:1.7) arc (90-60-8:90-180+8:1.7);
  \draw[ultra thick, darkblue]
    (90-180-8:1.7) arc (90-180-8:90-180-120+8:1.7);

  \draw[ultra thick, darkblue]
    (90-180-8:1.7)
      .. controls (90-180-8+5:1.7) and (90-180-8+5:1.7) ..
    (90-180+8+12:.3);

  \begin{scope}

  \clip (-.5,-.3) rectangle (.5,-1.3);

  \begin{scope}[shift={(.04,0)}]
  \draw[ultra thick, white]
    (90-180+8:1.7)
      .. controls (90-180+8-5:1.7) and (90-180+8-5:1.7) ..
    (90-180-8-12:.3);
  \end{scope}

  \begin{scope}[shift={(-.04,0)}]
  \draw[ultra thick, white]
    (90-180+8:1.7)
      .. controls (90-180+8-5:1.7) and (90-180+8-5:1.7) ..
    (90-180-8-12:.3);
  \end{scope}

  \end{scope}

  \draw[ultra thick, darkblue]
    (90-180+8:1.7)
      .. controls (90-180+8-5:1.7) and (90-180+8-5:1.7) ..
    (90-180-8-12:.3);

  \draw[ultra thick, darkblue]
    (90-180-8-12:.3)
    .. controls (90-180-60-8:.12) and (90-180-60+8:.12) ..
    (90-180+8+12-120:.3);

  \begin{scope}[rotate=(120)]
  \draw[ultra thick, darkblue]
    (90-180-8:1.7)
      .. controls (90-180-8+5:1.7) and (90-180-8+5:1.7) ..
    (90-180-8-12:.3);

  \draw[ultra thick, darkblue]
    (90-180+8:1.7)
      .. controls (90-180+8-5:1.7) and (90-180+8-5:1.7) ..
    (90-180+8+12:.3);

  \draw[ultra thick, darkblue]
    (90-180-8-12:.3)
    .. controls (90-180-60-8:.12) and (90-180-60+8:.12) ..
    (90-180+8+12-120:.3);

  \end{scope}

  \begin{scope}[rotate=(-120)]
  \draw[ultra thick, darkblue]
    (90-180-8:1.7)
      .. controls (90-180-8+5:1.7) and (90-180-8+5:1.7) ..
    (90-180-8-12:.3);

  \draw[ultra thick, darkblue]
    (90-180+8:1.7)
      .. controls (90-180+8-5:1.7) and (90-180+8-5:1.7) ..
    (90-180+8+12:.3);

  \draw[ultra thick, darkblue]
    (90-180-8-12:.3)
    .. controls (90-180-60-8:.12) and (90-180-60+8:.12) ..
    (90-180+8+12-120:.3);

  \end{scope}

\end{tikzpicture}
}
\\
\end{tabular}
}
\end{equation}

In the context of gauge theory this was famously observed in \cite{tHooft74} for
$\mathfrak{g} = \mathfrak{u}(N)$
(see also \cite[(34)]{BarNatan95}), in which case only the first summands in \eqref{MetricContractionOfFundamentalActions} and
\eqref{HooftDoubleLineOfTrivalentVertex}
appear; the generalization to arbitrary semisimple Lie algebras
was observed in \cite[Figure 14]{Cvitanovic76} and partially again in
\cite{Cicuta82}. Later \cite[6.3]{BarNatan95} reconsidered this,
apparently independently, in the general context of Lie algebra weight systems,
which is reviewed in \cite[6.2.6]{CDM11}. The case $\mathfrak{g} = \mathrm{sl}(N)$ is also discussed in \cite[6.1.8]{CDM11}\cite[14.4]{JacksonMoffat19}.

\medskip

We focus on the case $\mathfrak{g} = \mathfrak{so}(N)$.

\medskip

\noindent {\bf Emerging string worldsheets.}
The 't Hooft double line construction \eqref{HooftDoubleLineOfTrivalentVertex}
exhibits each Jacobi/Feynman diagram as a linear combination of
ribbon graphs (``fatgraphs''), underlying which are
isomorphism classes of surfaces
with marked boundaries (see \cite[Def. 1.12]{BarNatan95}).
This defines a linear function

\vspace{-.2cm}

\begin{equation}
  \label{tHooftConstructionOnSetOfJacobiDiagrams}
  \xymatrix{
    \overset{
      \mathclap{
      \mbox{
        \tiny
        \color{blue}
\begin{tabular}{c}
          Linear combinations of
          \\
          CS/RZ Feynman diagrams
          \\
          (Jacobi diagrams)
          \\
          $\phantom{a}$
        \end{tabular}
      }
      }
    }{
    \mathrm{Span}
    \big(
      \mathcal{D}^{{}^{\mathrm{t}}}
    \big)
    }
    \;\;\;\;
    \ar[rrr]^-{
      \overset{
        \mathclap{
        \mbox{
          \tiny
          \color{blue}
          \begin{tabular}{c}
            't Hooft construction
            \\
            $\phantom{a}$
          \end{tabular}
        }
        }
      }{
        \mathrm{tH}_{\mathfrak{g}}
      }
    }
    &&&
    \;\;\;\;
    \overset{
      \mathclap{
      \mbox{
        \tiny
        \color{blue}
        \begin{tabular}{c}
          Linear combinations of
          \\
          topological open string worldsheets
          \\
          (marked surfaces)
          \\
          $\phantom{a}$
        \end{tabular}
      }
      }
    }{
    \mathrm{Span}
    \big(
      \mathrm{MarkedSurfaces}_{/\sim}
    \big)
    }
  }
\end{equation}
from the linear span of the set \eqref{JacobiDiagrams} of Jacobi diagrams
to the linear span of the set of isomorphism classes of marked
surfaces.
Specifically for
$\mathfrak{g} = \mathfrak{so}(N)$ this function is given
on single chords \eqref{SingleChordEvaluation} by

\begin{center}
\begin{tabular}{cccccccc}

 \raisebox{-25pt}{
 \begin{tikzpicture}

 \draw (-1,.75) node {\scalebox{.8}{$V$}};
 \draw (+1,.75) node {\scalebox{.8}{$V$}};
 \draw (-1,-.75) node {\scalebox{.8}{$V$}};
 \draw (+1,-.75) node {\scalebox{.8}{$V$}};

 \begin{scope}
   \clip (-1,.5) rectangle (1,-.5);
   \draw[draw=orangeii, fill=orangeii] (-1,0) circle (.07);
   \draw[draw=orangeii, fill=orangeii] (+1,0) circle (.07);
 \end{scope}

 \draw[thick, orangeii] (-1,0) to (1,0);

 \draw[ultra thick, darkblue] (-1,.5) to (-1,-.5);
 \draw[ultra thick, darkblue] (+1,.5) to (+1,-.5);

 \begin{scope}[shift={(-1,0)}]
 \draw[draw=darkblue, thick, fill=white]
   (-.15,-.13)
   rectangle node{\tiny $\rho$}
   (.15,.17);
 \end{scope}
 \begin{scope}[shift={(+1,0)}]
 \draw[draw=darkblue, thick, fill=white]
   (-.15,-.13)
   rectangle node{\tiny $\rho$}
   (.15,.17);
 \end{scope}

\end{tikzpicture}
}

 & \;\;=\;\; &

\raisebox{-25pt}{
\begin{tikzpicture}

  \draw[ultra thick, darkblue]
    (-1,.5) ..  controls (-1,.4) and (-1,.4) .. (-.9,.4);
  \draw[ultra thick, darkblue]
    (+1,.5) ..  controls (+1,.4) and (+1,.4) .. (+.9,.4);
  \draw[ultra thick, darkblue]
    (-1,-.5) ..  controls (-1,-.4) and (-1,-.4) .. (-.9,-.4);
  \draw[ultra thick, darkblue]
    (+1,-.5) ..  controls (+1,-.4) and (+1,-.4) .. (+.9,-.4);

  \draw[ultra thick, darkblue] (-.9,.4) to (.9,.4);
  \draw[ultra thick, darkblue] (-.9,-.4) to (.9,-.4);

 \draw (-1,.75) node {\scalebox{.8}{$V$}};
 \draw (+1,.75) node {\scalebox{.8}{$V$}};
 \draw (-1,-.75) node {\scalebox{.8}{$V$}};
 \draw (+1,-.75) node {\scalebox{.8}{$V$}};

\end{tikzpicture}
}
&\;\;-\;\;&
\raisebox{-25pt}{
\begin{tikzpicture}

  \draw[ultra thick, darkblue]
    (-.9,-.4) .. controls (-.4,-.4) and (+.4,+.4) .. (+.9,+.4);

  \begin{scope}[shift={(.03,.03)}]
  \draw[ultra thick, white]
    (-.9,.4) .. controls (-.4,.4) and (+.4,-.4) .. (+.9,-.4);
  \end{scope}

  \begin{scope}[shift={(-.03,-.03)}]
  \draw[ultra thick, white]
    (-.9,.4) .. controls (-.4,.4) and (+.4,-.4) .. (+.9,-.4);
  \end{scope}
  \draw[ultra thick, darkblue]
    (-.9,.4) .. controls (-.4,.4) and (+.4,-.4) .. (+.9,-.4);

  \draw[ultra thick, darkblue]
    (-1,.5) ..  controls (-1,.4) and (-1,.4) .. (-.9,.4);
  \draw[ultra thick, darkblue]
    (+1,.5) ..  controls (+1,.4) and (+1,.4) .. (+.9,.4);
  \draw[ultra thick, darkblue]
    (-1,-.5) ..  controls (-1,-.4) and (-1,-.4) .. (-.9,-.4);
  \draw[ultra thick, darkblue]
    (+1,-.5) ..  controls (+1,-.4) and (+1,-.4) .. (+.9,-.4);

 \draw (-1,.75) node {\scalebox{.8}{$V$}};
 \draw (+1,.75) node {\scalebox{.8}{$V$}};
 \draw (-1,-.75) node {\scalebox{.8}{$V$}};
 \draw (+1,-.75) node {\scalebox{.8}{$V$}};

\end{tikzpicture}
}
\\
\end{tabular}
\end{center}

\noindent and on single internal vertices by

\vspace{3mm}
\hspace{-.9cm}
%
}
\end{equation}

\vspace{.3cm}

\noindent This is the generalization to unoriented open string worldsheets of
the 't Hooft construction for Chern-Simons theory
as an open topological string theory \cite[Figures 1 \& 2]{Witten92}.\footnote{Beware that
(only) for \emph{closed} string gravity duals
of Chern-Simons theory \cite{GopakumarVafa99}
are these open worldsheets turned into closed string
worldsheets by gluing disks onto all their free boundaries,
see \cite[1.1]{GaiottoRastelli03}\cite[III, p. 14]{Marino04}.}

\medskip

\noindent {\bf Chern-Simons observables as topological string amplitudes. }
We now observe that the open topological string worldsheets
as in \eqref{AtHooftSurfaceExample},
given by the 't Hooft construction \eqref{tHooftConstructionOnSetOfJacobiDiagrams},
are reflected by the higher observables (Prop. \ref{HigherObservablesOnIntersectingBranesAreWeightSystems})
in that the image of \emph{stringy weight systems},
assigning weight amplitudes to open string worldsheets,
embed into the space of weight systems on chord diagrams.

\medskip

To fully account for the quark/Wilson loop,
consider the function
\begin{equation}
  \label{ReOrderingVerticesOfJacobiDiagram}
  \xymatrix{
    \mathrm{Span}
    \big(
      \mathcal{D}^{{}^{\mathrm{t}}}
    \big)
    \ar[rr]^-{ \mathrm{perm} }
    &&
    \mathrm{Span}
    \big(
      \mathcal{D}^{{}^{\mathrm{t}}}
    \big)
  }
\end{equation}
which sends a Jacobi diagram with $n$ external vertices
to the
linear combination of the $n!$ ways of permuting them
along the Wilson loop circle.
Then the composition of $\mathrm{perm}$ \eqref{ReOrderingVerticesOfJacobiDiagram}
with the 't Hooft double line construction
$\mathrm{tH}_{\mathfrak{so}}$ \eqref{tHooftConstructionOnSetOfJacobiDiagrams}
respects the STU-relations \eqref{STURelations}
(\cite[Thm. 10 with Thm. 8]{BarNatan95})
and thus descends to linear map on
$\mathcal{A}_\bullet := \mathcal{A}^{{}^{\mathrm{t}}}$ \eqref{RoundChordDiagramsModulo4TIsJavcobidDiagramsModuloSTU}

\begin{equation}
  \label{tHooftLimitOnDiagrams}
  \xymatrix@R=11pt{
    \overset{
      \mathclap{
      \mbox{
        \hspace{-30pt}
        \tiny
        \color{blue}

  }
  }
  \right]
  \mathrlap{\; + \cdots}
  }
\end{equation}

Given that weight systems \eqref{SequenceOfMapsOfWeightSystems}
on Jacobi diagrams
reflect assignments of Feynman amplitudes to Feynman diagrams
(for Chern-Simons theory \cref{RoundWeightSystemsExhibitHolography}
or Rozansky-Witten theory \cref{RoundWeightSystemsEncodeMonopoleModuli})
we are to regard \emph{stringy weight systems}
\begin{equation}
  \label{StringyWeightSystems}
  \xymatrix{
    \overset{
      \mathclap{
      \mbox{
        \tiny
        \color{blue}
        \begin{tabular}{c}
          Stringy weight systems
          \\
          $\phantom{a}$
        \end{tabular}%
        }
      }
    }{
      \underset{n \in \mathbb{N}}{\prod}
      \mathcal{W}^n\langle \hbar^n \rangle
    }
  }
\end{equation}
as encoding open string scattering amplitudes.

One finds (\cite[Thm. 11]{BarNatan95}) that
stringy weight systems span those Lie algebra weight
systems \eqref{LieAlgebraWeightSystemsOnRoundChordDiagrams} that come from metric Lie representations of $\mathfrak{gl}(N)$
and $\mathfrak{so}(N)$, and contains those coming from $\mathfrak{sp}(N)$,
with $N$ ranging over the natural numbers :

\vspace{-.5cm}

$$
  \xymatrix@R=11pt{
  \overset{
    \mathclap{
    \mbox{
      \tiny
      \color{blue}
      \begin{tabular}{c}
        Stringy weight systems
        \\
        $\phantom{a}$
      \end{tabular}
    }
    }
  }{
    \mathrm{Im}
    \big(
      \mathcal{S}
    \big)
  }
  \;\; \ar@{^{(}->}[dr]_-{ \mbox{\tiny \eqref{StringyWeightSystems}} }
  \ar@{}[r]|-{\simeq}
  &
  \overset{
    \mathclap{
    \mbox{
      \tiny
      \color{blue}
      \begin{tabular}{c}
        Classical
        \\
        Lie algebra weight systems
        \\
        $\phantom{a}$
      \end{tabular}
    }
    }
  }{
    \mathrm{Im}
    \big(
      \mathcal{L}_{\mathfrak{gl}}
      \oplus
      \mathcal{L}_{\mathfrak{so}}
    \big)
  }
  \ar@{^{(}->}[d]^-{
    \mbox{
      \tiny
      \eqref{LieAlgebraWeightSystemsOnRoundChordDiagrams}
    }
  }
  \ar@{}[r]|-{\supset}
  &
  \mathrm{Im}
  \big(
    \mathcal{L}_{\mathfrak{sp}}
  \big).
  \ar@{^{(}->}[dl]^-{
    \mbox{
      \tiny
      \eqref{LieAlgebraWeightSystemsOnRoundChordDiagrams}
    }
  }
  \\
  &
  \underset{n \in \mathbb{N}}{\prod}
  \big(
    \underset{
      \mathclap{
      \mbox{
        \tiny
        \color{blue}
        \begin{tabular}{c}
          $\phantom{a}$
          \\
          $\phantom{a}$
          \\
          All weight
          systems
        \end{tabular}
      }
      }
    }{
      \mathcal{W}^n
    }
    \langle \hbar^n \rangle
  \big)
  }
$$
It follows from this weight-theoretic result that the perturbative Wilson loop observables
of Chern-Simons theory \eqref{UniversalWilsonLoopAsAMap},
for $\mathfrak{g} =\mathfrak{\mathfrak{so}}(N)$ and with the Wilson loop in the
fundamental representation, are equivalent,
under the 't Hooft construction \eqref{tHooftLimitOnDiagrams},
to observables of a unoriented open topological string theory,
as in the identification of Chern-Simons theory as a
topological string theory in \cite{Witten92} (reviewed in \cite{Marino04}).

\medskip

%%%%%%%%%%%%%%%%%%%%%%%%%%%%%%%%%%%%%%%%%%%%%%%%%%%%%%%%%%%%%%%%%%%%%%%%%
\subsection{Horizontal weight systems observe string topology operations}
\label{HorizontalChordDiagramsEncodesStringTopologyOperations}
%%%%%%%%%%%%%%%%%%%%%%%%%%%%%%%%%%%%%%%%%%%%%%%%%%%%%%%%%%%%%%%%%%%%%%%%%

\begin{equation}
\label{SullivanChordDiagrams}
\mbox{
\hspace{-.71cm}

}
}
\!\!\!\!
\right\}
}
}
}
\end{tabular}
}
\end{equation}

\vspace{.1cm}

\noindent For example, every round chord diagram \eqref{RoundChordDiagrams}
is a Sullivan chord diagram, but a Jacobi diagram \eqref{JacobiDiagrams}
is only a Sullivan chord diagram if its internal edges form a tree
(so the Jacobi diagram
shown in \eqref{tHooftLimitOnDiagrams} is \emph{not} a
Sullivan chord diagram).

\vspace{2mm}
\noindent {\bf String topology TQFT.} The tree-condition
ensures \cite[2]{CohenGodin04}
that for $D \in \mathcal{D}^{{}^{\mathrm{s}}}$
a Sullivan chord diagram with $n_{\mathrm{in}}, n_{\mathrm{out}}$
in/out-going boundary components,
the pull-push operation in homology
through the mapping space out of $\mathrm{tH}_{\mathfrak{u}}(D)$ exists
\cite[Theorem 4]{CohenGodin04},
for $X$ an oriented target manifold, with free loop space
$\mathcal{L}X := \mathrm{Maps}(S^1, X)$:
\begin{equation}
  \label{PullPushOperations}
  \hspace{-2cm}
  \raisebox{20pt}{
  \xymatrix@R=-12pt@C=3.5em{
    &
    \overset{
      \mathclap{
      \mbox{
        \tiny
        \color{blue}
\begin{tabular}{c}
          Space of maps
          \\
          from worldsheet $\mathrm{tH}_{\mathfrak{u}}(D)$
          \\
          to spacetime $X$
      %    \\
       %  $\phantom{a}$
        \end{tabular}
      }
      }
    }
    {
      \mathrm{Maps}\big(\mathrm{tH}_{\mathfrak{u}}(D), X\big)
    }
    \ar[dl]_-{ p_{\mathrm{in}} }
    \ar[dr]^-{ p_{\mathrm{out}} }
    \\
    \overset{
    \mathllap{
      \mbox{
        \tiny
        \color{blue}
        \begin{tabular}{c}
          Space of incoming
          \\
          string configurations
          \\
          $\phantom{a}$
        \end{tabular}%
    }
  }
  }
  }
  }
\end{equation}
There is a precise sense \cite[Ex. 7.1]{Schreiber14}
in which this pull-push operation \eqref{PullPushOperations}
is the \emph{cohomological path-integral}
of a topological closed string theory with target space $X$
and worldsheet geometry $\mathrm{tH}_{\mathfrak{u}}(D)$.
Indeed, as the Sullivan chord diagram $D$ and hence
the worldsheet topology $\mathrm{tH}_{\mathfrak{u}}(D)$ varies,
the operations \eqref{PullPushOperations} organize into the
propagators of a 2d topological field theory
(see \cite[3]{CohenVoronov05}).

\medskip

\noindent {\bf String topology operations from horizontal chord diagrams.}
We observe that Sullivan chord diagrams \eqref{SullivanChordDiagrams}
without any internal vertices,
and hence the corresponding string topology operations \eqref{PullPushOperations},
arise precisely as the closures of horizontal chord diagrams \eqref{HorizontalChordDiagrams}
with respect to general monodromy permutations $\sigma$
as in
\eqref{LieAlgebraWeightSystemOnHoizontalChordDiagramsWithStacksOfCoincidentStrands}.
If $\sigma = (N_{\mathrm{f}}12\cdots)$
has only a single cycle (single orbit)
the result is a round chord diagram as in \eqref{AHorizontalChordDiagramTraced}:
$$
  \xymatrix{
  \overset{
    \mathclap{
    \mbox{
      \tiny
      \color{blue}

    }
    }
  }{
    \mathcal{D}^{{}^{\mathrm{s}}}
  }
  }.
$$
But for general permutations
$\sigma$ with $n$ cycles ($n$ orbits) as in
\eqref{LieAlgebraWeightSystemOnHoizontalChordDiagramsWithStacksOfCoincidentStrands}
the result is a Sullivan chord diagram whose
corresponding cobordism has $n$ outgoing boundary components.
For example:

\begin{equation}
\label{HorizontalChordDiagramClosedToSullivanChordDiagram}
\hspace{-1cm}
\xymatrix@R=4pt{
  \overset{
    \mathclap{
    \mbox{
      \tiny
      \color{blue}
      %
}
}
\!\!\!\!\!\!\!\!
\right]
}
\mathrlap{+ \cdots}
\\
&
}
\end{equation}
On the bottom left of \eqref{HorizontalChordDiagramClosedToSullivanChordDiagram} we are showing the
form of the associated Lie algebra weights \eqref{LieAlgebraWeightSystemOnHoizontalChordDiagramsWithStacksOfCoincidentStrands},
which now are BMN multi-trace observables, see \cref{MatrixModelObservables}.

\newpage

%%%%%%%%%%%%%%%%%%%%%%%%%%%%%%%%%%%%%%%%%%%%%%%%%%%%%%%%%%%%%%%%%%%%%%
\subsection{Horizontal chord diagrams are BMN model multi-trace observables}
\label{MatrixModelObservables}
%%%%%%%%%%%%%%%%%%%%%%%%%%%%%%%%%%%%%%%%%%%%%%%%%%%%%

While the  single-trace gauge theory observables from
\cref{LieAlgebraWeightSystemsAreSingleTraceObservables}
correspond to single-string states
under the AdS/CFT correspondence, general multi-string
states correspond \cite{ChalmersSchalm99} to \emph{multi-trace observables} \cite{BDHM98},
hence to polynomials in single-trace observables \cite[p. 1]{Witten01}.

\medskip

\noindent {\bf Invariant multi-trace observables in the BMN matrix model.}
Thus, in generalization of the discussion in \cref{LieAlgebraWeightSystemsAreSingleTraceObservables},
a supersymmetric and $S^2$-rotation invariant multi-trace observable
in the BMN matrix model sends a supersymmetric state given by a
complex Lie algebra representation
$\mathfrak{su}(2)_{\mathbb{C}} \otimes V \overset{\rho}{\to} V$ to
expressions like the following:
\begin{equation}
\label{MultiTraceObservableInBMNMatrixModel}
\hspace{-1cm}
\raisebox{-60pt}{

}
\end{equation}

\noindent {\bf Horizontal chord diagrams are BMN matrix model multi-trace observables.}
The multi-trace expressions \eqref{MultiTraceObservableInBMNMatrixModel}
are manifestly the values \eqref{LieAlgebraWeightSystemOnHoizontalChordDiagramsWithStacksOfCoincidentStrands}
of the Lie algebra weight system
$w_{{}_V}$ corresponding to the given BMN model state on
horizontal chord diagrams encoding the multi-trace observable,
as in \eqref{HorizontalChordDiagramClosedToSullivanChordDiagram}.
But the {fundamental theorem of horizontal weight systems},
Prop. \ref{FundamentalTheremOfWeightSystems}, says
that \emph{every} horizontal weight system arises this way
\eqref{FundamentalTheoremEquivalence}\eqref{ExtensionAndRestrictionBetweensl2Andgl2}, hence that:

\medskip
\noindent {\it Weight systems
on horizontal chord diagrams are equivalently the
supersymmetric BMN model states as seen by
the colleciton of $S^2$-invariant multi-trace observables,
which in turn are encoded by chord diagrams.}

\medskip

\noindent {\bf In summary}, this means we have found the following
identifications (see \hyperlink{Figure2}{Figure 2}):

$$
  \xymatrix@R=20pt@C=4em{
        \left\{ \!\!\!\!\!
    \mbox{
      \begin{tabular}{c}
        Higher observables on
        \\
        $\mathrm{D}p\perp \mathrm{D}(p+2)$ intersections
        \\
        by \hyperlink{HypothesisH}{\it Hypothesis H}
      \end{tabular}
    }
   \!\!\!\!\! \right\}
    \ar@{<->}[dd]_-{ \simeq }
    &
    \left\{\!\!\!\!\!
    \mbox{
      \begin{tabular}{c}
        Higher co-observables on
        \\
        $\mathrm{D}p\perp \mathrm{D}(p+2)$ intersections
        \\
        by \hyperlink{HypothesisH}{\it Hypothesis H}
      \end{tabular}
    }
  \!\!\!\!\!  \right\}
    \ar@{<->}[dd]_-{ \simeq }
    \\
    \ar@{}[r]|-{
      \scalebox{.6}{
        \begin{tabular}{c}
        Prop. \ref{HigherObservablesEquivalentToCohomologyOfLoopedConfigurationSpace},
        \\
        Prop. \ref{HigherObservablesOnIntersectingBranesAreWeightSystems}
        \end{tabular}
      }
    }
    &
    \\
    \left\{ \!\!\!\!\!
    \mbox{
      \begin{tabular}{c}
        Horizontal weight systems
      \end{tabular}
    }
  \!\!\!\!\!  \right\}
    \ar@{<->}[dd]_-{
      \simeq
    }
    &
    \left\{\!\!\!\!\!
    \mbox{
      \begin{tabular}{c}
        Horizontal chord diagrams
      \end{tabular}
    }
\!\!\!\!\!    \right\}
    \ar@{<->}[dd]_-{
      \simeq
    }
    \\
    \ar@{}[r]|-{
      \scalebox{.6}{
        \begin{tabular}{c}
          Prop. \ref{FundamentalTheremOfWeightSystems}
        \end{tabular}
      }
    }
    &
    \\
    \left\{\!\!\!\!\!
    \mbox{
      \begin{tabular}{c}
        Supersymmetric states
        \\
        of the BMN matrix model
        \\
        and $\mathrm{D}p\perp \mathrm{D}(p+2)$ fuzzy funnels
      \end{tabular}
    }
   \!\!\!\!\! \right\}
    &
    \left\{\!\!\!\!\!
    \mbox{
      \begin{tabular}{c}
        Invariant multi-trace observables
        \\
        of the BMN matrix model
        \\
        and $\mathrm{D}p\perp \mathrm{D}(p+2)$ fuzzy funnels
      \end{tabular}
    }
   \!\!\!\!\! \right\}
  }
$$

\newpage

%%%%%%%%%%%%%%%%%%%%%%%%%%%%%%%%%%%%%%%%%%%%%%%%%%%%%%%%%%%%%%%%%%%%%%%%%%
\subsection{Horizontal weight systems contain M2/M5-brane states}
\label{LargeNLimitM2M5BraneBoundStates}
%%%%%%%%%%%%%%%%%%%%%%%%%%%%%%%%%%%%%%%%%%%%%%%%%%%%%%%%%%%%%%%%%%%%%%%%%%

We discuss here that
transversal microscopic M2/M5 brane bound states and their large N
macroscopic limits are identified in weight systems on chord diagrams.
While we had observed
that the supersymmetric states of the BMN matrix model
(\cref{LieAlgebraWeightSystemsAreSingleTraceObservables}),
given by fuzzy 2-sphere geometries (\cref{su2WeightSystemsAreFuzzyFunnelObservables}),
are seen by multi-trace observables
as their image in weight systems on Sullivan chord diagrams
(\cref{MatrixModelObservables})
hence on horizontal chord diagrams closed by some winding monodromy
permutation (\cref{OnStacksOfCoincidentStrands}),
the BMN matrix model has, of course, a tower of
excited states beyond the fully supersymmetric ground states.
One might therefore suspect that the theory of weight systems, and hence cohomotopy,
 reflects
only a negligible corner of the M-theory captured by
the BMN matrix model. Remarkably, the \emph{opposite} is the case:

\medskip

\noindent {\bf M2/M5 brane bound states in the BMN matrix model.}
It was suggested in \cite{MSJVR03} and checked in \cite{AIST17a}
(surveyed in \cite{AIST17b}) that
finite numbers of stacks of coincident light-cone transversal
M2- and M5-brane states are given by
isomorphism classes of some kind of limit sequences in the
set of supersymmetric states of the BMN matrix model.
Concretely, let
\begin{equation}
  \label{DecompositionofARepresentation}
  V \;:=\;
  \underset{
    \underset{
      \mathclap{
      \mbox{
        \tiny
        \begin{tabular}{c}
          \color{blue}
          stacks of
          coincident branes
          \\
          ({\color{olive}}direct sum over irreps)
        \end{tabular}
      }
      }
    }{
    \underbrace{
      \scalebox{.8}
      {
        $i$
      }
    }
    }
  }{\medoplus}
  \!
  \big(
    \overset{
      \mathclap{
      \mbox{
        \tiny
        \begin{tabular}{c}
          \color{blue}
          M2/M5-brane charge in $i$th stack
          \\
          ({\color{olive}}$i$th irrep with multiplicity)
        \end{tabular}
      }
      }
    }{
    \overbrace{
      N^{{}^{(\mathrm{M2})}}_i
      \cdot
      \mathbf{N}^{{}^{(\mathrm{M5})}}_i
    }
    }
  \big)
  \;\in\;
  \mathfrak{su}(2)_{\mathbb{C}} \mathrm{MetMod}_{/\sim}
\end{equation}
denote the isomorphism class of the $\mathfrak{su}(2)_{\mathbb{C}}$ representation with
$N^{{}^{(\mathrm{M2})}}_i \in \mathbb{N}$ direct summands of the
$N^{{}^{(\mathrm{M5})}}_i$-dimensional irreducible representation,
for $i$ in some finite index set,
hence with total dimension
\begin{equation}
  N
    \;:=\;
  \underset{i}{\sum}
  N^{{}^{(\mathrm{M2})}}_i
  {N}^{{}^{(\mathrm{M5})}}_i
  \;\;
  \in
  \mathbb{N}.
\end{equation}
Then a \emph{sequence} of such states/representations
corresponds to
a finite number of stacks of macroscopic M2 branes
\emph{or} macroscopic M5-branes depending on how
the sequence behaves in the large $N$ limit:
\begin{equation}
\label{LargeNLimitForM2M5Branes}
\hspace{-5mm}
\raisebox{-4pt}{
\mbox{
\begin{tabular}{ll|ll}
  \multicolumn{3}{c}{
    {\bf Stacks of macroscopic}
  }
  \\
 \hline
 \multicolumn{1}{c}{$\phantom{A}$} &
  \multicolumn{1}{c|}{
    {\it M2-branes}
  }
  &
  \multicolumn{1}{c}{
    {\it M5-branes}
  }
  \\
 If for all $i$: &
  $N^{{}^{(\mathrm{M5})}}_i \to \infty$
  &
  $N^{{}^{(\mathrm{M2})}}_i \to \infty$
  &
  (the relevant large N limit)
  \\
  with fixed &
  $N^{{}^{(\mathrm{M2})}}_i$
  &
   $N^{{}^{(\mathrm{M5})}}_i$
  &
  (the number of coincident branes in the $i$th stack)
  \\
  and fixed & $N^{{}^{(\mathrm{M2})}_i}\!\!\!/N$
  &
  $N^{{}^{(\mathrm{M5})}_i}\!\!\!/N$
  &
  (the charge/light-cone momentum carried by the $i$th stack)
  \\
  \hline
\end{tabular}
}
}
\end{equation}
cf. (\cite[Figure 2]{MSJVR03}\cite[(1.2)-(1.4)]{AIST17b}).

\medskip

\noindent {\bf Open problem.}
In order to make precise sense of the suggestion \eqref{LargeNLimitForM2M5Branes}
one needs to say where these limits are to be taken.
They cannot actually be taken in the set \eqref{DecompositionofARepresentation}
as this set is discrete: \emph{no} sequence
with $N \to \infty$ has a limit in this set.
Instead, the proper operational definition of
limits of states is that seen by limits of the
values of observables evaluated on these states.
For a holographic gauge theory like the BMN matrix model
the relevant observables $\mathcal{O}_D$ are multi-trace observables
\cref{MatrixModelObservables}; their structure is encoded
by Sullivan chord diagrams
$D \in \underset{n \in \mathbb{N}}{\oplus}\mathcal{A}^n$
and their value on a state
$
\Psi_{\scaleto{
  \left(
  \big\{
    N^{{}^{(\mathrm{M2})}}_i, \,
    N^{{}^{(\mathrm{M5})}}_i
  \big\}
  \right)
}{10pt}}
$
given by \eqref{DecompositionofARepresentation} is the
value of the corresponding Lie weight system (\cref{LieAlgebraWeightSystems})
$\underset{i \in \mathbb{N}}{\sum}
N^{{}^{(\mathrm{M2})}}_i \cdot w_{\mathbf{N}^{{}^{(\mathrm{M5})}}_i} \in \underset{n \in \mathbb{N}}{\prod} \mathcal{W}^n$
on $D$:

\vspace{-.8cm}

\begin{equation}
  \label{StatesInObservables}
  \hspace{1cm}
  \overset{
    \mathclap{
    \mbox{
      \tiny
      \color{blue}
      \begin{tabular}{c}
        Invariant
        \\
        multi-trace
        \\
        observable
        \\
        $\phantom{a}$
      \end{tabular}
    }
    }
  }{
    \mathcal{O}_D
  }
  \;:\;
  \overset{
    \mathclap{
    \mbox{
      \tiny
      \color{blue}
      \begin{tabular}{c}
        Supersymmetric
        \\
        state of
        \\
        BMN matrix model
        \\
        $\phantom{a}$
      \end{tabular}
    }
    }
  }{
\Psi_{\scaleto{
  \left(
  \big\{
    N^{{}^{(\mathrm{M2})}}_i, \,
    N^{{}^{(\mathrm{M5})}}_i
  \big\}
  \right)
}{10pt}}
  }
  \;\;\;\longmapsto\;\;\;
  \underset{i \in \mathbb{N}}{\sum}
  N^{{}^{(\mathrm{M2})}}_i
  \,
  \overset{
    \mathclap{
    \mbox{
      \tiny
      \color{blue}
      \begin{tabular}{c}
        Value of weight system
        \\
        on chord diagrams
        \\
        $\phantom{a}$
      \end{tabular}
    }
    }
  }{
    w_{\mathbf{N}^{\scalebox{.5}{$(\mathrm{M5})$}}_i}(D)
  }.
\end{equation}

This means that the large $N$ limits \eqref{LargeNLimitForM2M5Branes}
are to be considered in the space of weight systems on
Sullivan/horizontal chord diagrams.
It just remains to determine the proper normalizations:

\newpage

\noindent {\bf Coincident M5-Brane quantum states.}
Notice that, with
the identification \eqref{StatesInObservables}, the state of
2 coincident M5-branes, according to \eqref{LargeNLimitForM2M5Branes}
is given, via \eqref{ExtensionAndRestrictionBetweensl2Andgl2},
by the normalization of the
fundamental $\mathfrak{gl}(2,\mathbb{C})$-weight system
(Example \ref{TheFundamentalgl2WeightSystem}) regarded
as a quantum state (Example \ref{gl2FundamentalWeightSystemIsAState}).

\medskip

\noindent {\bf Single M2-brane states and normalization.}
By the discussion in \cref{su2WeightSystemsAreFuzzyFunnelObservables},
we have that the fuzzy 2-sphere state of
a \emph{single} M2-brane at any
$N^{{}^{(\mathrm{M5})}}$ is given,
under the identification \eqref{StatesInObservables},
by the following weight system:
\begin{equation}
  \label{SingleM2BraneAsWeightSystem}
  \underset{
    \mbox{
      \tiny
      \begin{tabular}{c}
        $\phantom{a}$
        \\
        \color{blue}
        Single M2-brane state in BMN model
        \\
        ({\color{olive}multiple of $\mathfrak{su}_{\mathbb{C}}$-weight system})
      \end{tabular}
    }
  }{
  \underbrace{
  \tfrac{
    4\pi \, 2^{2n}
  }{
    \left(
      \big(
        N^{{}^{(\mathrm{M5})}}
      \big)^2
      -
      1
    \right)^{\sfrac{1}{2} + n}
  }
  \,
  w_{\mathbf{N}^{{}^{(\mathrm{M5})}}}
  }
  }
  \;\;\;\in\;\;\;\;\;\;\;
  {
  \overset{
    \mathclap{
    \mbox{
      \tiny
      \begin{tabular}{c}
        \color{blue}
        States as seen by multi-trace observables...
        \\
        ({\color{olive}weight systems on horizontal chord diagrams})
      \end{tabular}
    }
    }
  }{
    \overbrace{
      \underset{
    \mathclap{
    \mbox{
      \tiny
      \begin{tabular}{c}
        \color{blue}
        ...of any length $2n$
        \\
        ({\color{olive}with any number $n$ of chords})
      \end{tabular}
    }
    }
      }{
      \underbrace{
        \underset{n \in \mathbb{N}}{\prod}
      }
      }
      \mathcal{W}^{n}
    }
  }
  }
\end{equation}
Here the power of $\sfrac{1}{2}$ in the
normalization factor accounts for the
normalization of the fuzzy integration in \cref{su2WeightSystemsAreFuzzyFunnelObservables},
while the power of the degree $n$
(which is the number of edges in a chord diagram \eqref{HorizontalChordDiagrams}, hence half the number of
insertions in any multi-trace observable evaluated on this state)
accounts for the normalization \eqref{SolutionOfNahmEquation}
of the functions on the fuzzy 2-sphere.
Thus, with this normalization the evaluation on any round chord
diagram produces the correct fuzzy sphere observable.
For example: \begin{equation}
  \label{SingleTraceObservableAsFuzzySphereShapeCoefficient}
\begin{aligned}
  \overset{
    \!\!\!
    \mathclap{
    \mbox{
      \tiny

  }
  \end{aligned}
\end{equation}

\medskip

\noindent {\bf Normalization and DLCQ.} This prescription gives observables
of the \emph{relative} radius/shape of the fuzzy spheres
incolved in an M2/M5-brane bound state. The absolute radius
is not observed. For example, for a single M2-brane we reduce to
states whose single-trace observables \eqref{SingleTraceObservableAsFuzzySphereShapeCoefficient}
measure fluctuations of the fuzzy 2-sphere (of any bit number $N^{{}^{(\mathrm{M5})}}$ but)
of unit radius:

\begin{equation}
  \label{SingleTraceObservableAsFuzzySphereShapeCoefficient}
\begin{aligned}
  \overset{
    \!\!\!
    \mathclap{
    \mbox{
      \tiny

  }
  }
  \;\;\;
  =
  1\;.
  \end{aligned}
\end{equation}
This is just as it must be for there to be
a large $N$-limit: In this limit the bare brane scale necessarily
diverges, and needs to be normalized against
the radius $R_{11}$ of the longitudinal spacetime circle in the
DLCQ prescription, to yield finite $p^+ = N/R_{11}$.

\newpage

\noindent {\bf M2-M5 brane bound states as weight systems.}
It follows that the weight systems corresponding
to M2/M5 branes states as in \eqref{LargeNLimitForM2M5Branes}
are to be mixtures of the single M2-brane states
\eqref{SingleM2BraneAsWeightSystem}:
\begin{equation}
  \label{M2M5BraneBoundStateWeights}
  \hspace{-1cm}
  \raisebox{40pt}{
  \scalebox{.9}{
  \xymatrix{
  \overset{
    \mathclap{
    \mbox{
      \tiny

        }
        }
      }{
      \underbrace{
        \mathbf{N}^{{}^{(\mathrm{M5})}}_i
      }
      }
      }
      }
   \, \Big)
    \;\left\vert
    \scalebox{.7}{
    $
    {\begin{aligned}
      &
      \left\{
        \big(
          N^{{}^{(\mathrm{M2})}}_i,
          N^{{}^{(\mathrm{M5})}}_i
        \big)
      \right\}_{i \in \mathbb{N}}
      \\
      &\in
      \underset{i \in \mathbb{N}}{\oplus}
      (\mathbb{N}\times \mathbb{N})
    \end{aligned}}
    \!\!\!
    $
    }
    \right.
  \right\}
  \ar[r]
  &
  \left\{
    \overset{
      \mbox{
        \tiny
        \color{blue}
        Mixture
      }
    }
    {
    \overbrace{
    \frac{
      \mathclap{\phantom{\vert^{\vert^{\vert}}}}
      1
    }{
     \underset{
       i \in \mathbb{N}
     }{
       \sum
     }
     N^{{}^{(\mathrm{M2})}}_i
    }
    \underset{i \in \mathbb{N}}{\sum}
    }
    }
    \;
    \overset{
      \mbox{
        \tiny
        \color{blue}
        Normalized radii
      }
    }{
    \overbrace{
      \frac{
        N^{{}^{(\mathrm{M2})}}_i
        \,
        4\pi
        \,
        2^{2n}
      }{
        \Big(
          \!
          \big(
            N^{{}^{(\mathrm{M5})}}_i
          \big)^2
          -1
          \!
        \Big)^{\!\!\sfrac{1}{2} + n}
      }
    }
    }
    \overset{
      \mbox{
        \tiny
        \color{olive}
        Lie weights
      }
    }{
    \overbrace{
      \mathclap{\phantom{\vert^{\vert^{\vert^{\vert^{\vert^b}}}}}}
      w_{
        \mathbf{N}^{{}^{(\mathrm{M5})}}_i
      }
    }
    }
    \;\left\vert
    \scalebox{.7}{
    $
    {\begin{aligned}
      &
      \left\{
        \big(
          N^{{}^{(\mathrm{M2})}}_i,
          N^{{}^{(\mathrm{M5})}}_i
        \big)
      \right\}_{i \in \mathbb{N}}
      \\
      &\in
      \underset{i \in \mathbb{N}}{\oplus}
      (\mathbb{N}\times \mathbb{N}_{\geq 1})
    \end{aligned}}
    \!\!\!
    $
    }
    \right.
  \right\}_{\bigg/\sim}
  \\
\raisebox{-47pt}{

}
  }
  }
  }
\end{equation}

\medskip

\noindent {\bf Large $N$ limit and macroscopic brane states.}
As graphically indicated on the bottom of
\eqref{M2M5BraneBoundStateWeights},
for M2/M5-brane bound states in the BMN matrix model
formulated, via \eqref{StatesInObservables}, as weight systems
\eqref{M2M5BraneBoundStateWeights},
the large $N$ limits suggested by \eqref{DecompositionofARepresentation}
do exist, in the vector space of weight systems:
$$
  \xymatrix@R=12pt@C=12pt{
      \overset{
      \mathclap{
      \mbox{
        \tiny
        \color{blue}
        %
    }
  }
  }
$$
This follows via \eqref{SingleTraceObservableAsFuzzySphereShapeCoefficient} by the standard convergence of the fuzzy sphere $S^2_N$ to the round 2-sphere for $N \to \infty$.

\newpage

%%%%%%%%%%%%%%%%%%%%%%%%%%%%%%%%%%%%%%%%%%%%%%%%%%%%%%%%%%%%%%%%%%%%%%%%
\subsection{Horizontal chord diagrams encode Hanany-Witten states}
\label{HananyWittenTheory}
%%%%%%%%%%%%%%%%%%%%%%%%%%%%%%%%%%%%%%%%%%%%%%%%%%%%%%%%%%%%%%%%%%%%%%%%

{\bf The graded-commutative algebra of horizontal chord diagrams.}
Recall from \eqref{HorizontalChordDiagramsModulo2TAnd4TRelations}
that $\mathcal{A}^{{}^{\mathrm{pb}}}_{N_{\mathrm{f}}}$
is the free \emph{graded associative} algebra on
generators $\big\{t_{i j } = t_{j i}\vert i\neq j \in \{1,\cdots N_{\mathrm{f}}\}\big\}$ in degree 1,
modulo the 2T and 4T relations.
By skew-symmetrizing this induces the
\emph{graded commutative} algebra obtained from the
same generators and relations:
\begin{equation}
  \label{GradedCommutativeAlgebraOfHorizontalChords}
  \mathcal{A}^{{}^{\mathrm{hw}}}_{N_{\mathrm{f}}}
  \;:=\;
  \mathrm{GradedComm}
  \Big(
    \big\{
      \underset{
        \mbox{\tiny $\mathrm{deg} = 1$ }
      }{t_{i j = t_{j i}}}
      \vert i\neq j \in \{1,\cdots N_{\mathrm{f}}\}\big\}
  \Big)\big/ (\mathrm{2T}, \mathrm{4T})\;.
\end{equation}
Horizontal chord diagrams \eqref{HorizontalChordDiagrams}
still represent generators in this graded-commutative algebra.
To indicate that we think of a horizontal chord diagram as a generator
in $\mathcal{A}^{{}^{\mathrm{hw}}}_{N_{\mathrm{f}}}$,
we complete each chord by a gray line to the left or to the right, as in the
 following example:
\begin{equation}
\label{AGenericVanishingElementOfSkewSymmetrizedHorizontal}
\left[
\raisebox{-107pt}{

}
\right]
\;\;=\;\;
t_{45} \wedge t_{35} \wedge t_{25} \wedge t_{15} \wedge t_{14}
  \wedge t_{2 4}
  \;\;\in\;\;
  \mathcal{A}^{{}^{\mathrm{hw}}}_{N_{\mathrm{f}} = 5}\;.
\end{equation}

\medskip
\noindent In fact, this element vanishes, because the 2T-relations \eqref{2TRelationsOnHorizontalChordDiagrams} now say that the product of two
chords vanishes if they do not connect to one common strand. In the example \eqref{AGenericVanishingElementOfSkewSymmetrizedHorizontal}
the 2T relation gives
$$
  t_{15} \wedge t_{24} \;=\; 0
  \,.
$$
Therefore, a non-vanishing homogeneous element
in \eqref{GradedCommutativeAlgebraOfHorizontalChords}
has to look either like this:
$$
\left[
\scalebox{.9}{
\raisebox{-78pt}{

}
}
\right]
\;\;\;=\;\;
\begin{aligned}
  & \phantom{A}
  \\
  & t_{12} \wedge t_{13} \wedge t_{14} \wedge t_{15}
  \\
  &
  \in\;\;
  \mathcal{A}^{{}^{\mathrm{hw}}}_{N_{\mathrm{f}} = 5}
\end{aligned}
$$

\medskip

\noindent or like this:
$$
\left[
\scalebox{.9}{
\raisebox{-78pt}{
%
}
}
\right]
\;\;\;=\;\;
\begin{aligned}
  & \phantom{A}
  \\
  & t_{45} \wedge t_{35} \wedge t_{25} \wedge t_{15}
  \\
  &
  \in\;\;
  \mathcal{A}^{{}^{\mathrm{hw}}}_{N_{\mathrm{f}} = 5}
\end{aligned}
$$
up to permutation of strands.

\medskip

\noindent {\bf Hanany-Witten theory.}
We observe that the elements of the skew-symmetrized graded-commutative algebra
of $\mathcal{A}^{{}^{\mathrm{hw}}}_{N_{\mathrm{f}}}$\eqref{GradedCommutativeAlgebraOfHorizontalChords}
of horizontal chord diagrams reflect the diagrammatics of Hanany-Witten
$\mathrm{D}p-\mathrm{D}(p+2)$ brane configurations according to
\cite[6]{HananyWitten97}\cite[3]{GaiottoWitten08}
(see also \cite[23]{HoriOoguriOz98}\cite[p. 83-]{GiveonKutasov99}\cite[(6.12)]{GKSTY01}\cite[Fig. 3.13]{Fazzi17}) if we identify:

\noindent
\begin{equation}
\label{IdentificationWithHananyWittenConfigurations}
\mbox{
\hspace{-3cm}
\begin{minipage}[l]{7cm}
\begin{enumerate}[{\bf (i)}]
  \item strands as $\mathrm{D}(p+2)$-branes;
  \item chords as $\mathrm{D}p$-branes,
    \newline
    stretching between $\mathrm{D}(p+2)$s;
  \item green dots as $\mathrm{NS}5$-branes;
  \item gray lines as $\mathrm{D}p$-branes,
    \newline
    stretching
    from $\mathrm{NS}5$ to $\mathrm{D}(p+2)$.
\end{enumerate}
\end{minipage}
\scalebox{.9}{
\raisebox{-84pt}{

}
}
}
\end{equation}
With this identification we find that the algebra of horizontal chord diagrams
 reflects the following rules of Hanany-Witten theory:

(1) {\it The s-rule.}

(2) {\it The breaking of $\mathrm{D}p$-branes on $\mathrm{D}(p+2)-branes$.}

(3) {\it The ordering constraint.}

\medskip

\noindent {\bf (1) The s-rule.}
A direct consequence of the graded-commutativity in \eqref{GradedCommutativeAlgebraOfHorizontalChords}
and the fact that the chord generators are in degree 1
is that diagrams of the following form vanish:
$$
\left[
\raisebox{-50pt}{
%
}
\right]
\;\;=\;\;
0
$$
Under the identification \eqref{IdentificationWithHananyWittenConfigurations},
these are the configurations where two $\mathrm{D}p$-branes
end on the same $\mathrm{D}(p+2)$-brane. That these configurations
are excluded (if supersymmetry is required) is
known as the \emph{s-rule} of Hanany-Witten theory, going back to
the discussion of \emph{s-configurations} in \cite{HananyWitten97}
and made explicit in \cite[p. 83-]{GiveonKutasov99}.

\medskip
We notice that in \cite[2.3]{BGS97}\cite{BachasGreen98}
the s-rule has been argued to be nothing but the implication of the
\emph{Pauli exclusion principle} for the fermions on the
intersecting branes. But of course the mathematical reflection
of the Pauli exclusion principle is, at its core,
precisely the graded-commutativity as in \eqref{GradedCommutativeAlgebraOfHorizontalChords}.

\medskip

\noindent {\bf (2) Breaking of $\mathrm{D}p$-branes on $\mathrm{D}(p+2)$-branes.}
A non-vanishing element of $\mathcal{A}^{{}^{\mathrm{hw}}}_{N_{\mathrm{f}}}$
\eqref{GradedCommutativeAlgebraOfHorizontalChords}
may also be of the form
$$
\left[
\raisebox{-67pt}{
\begin{tikzpicture}

 \begin{scope}[shift={(0,3)}]
   %\clip (0,-.1) rectangle (.1,+.1);
   \draw[->][thick, orangeii] (0,0) to (0.122,0);
   \draw[thick, white] (0,0) to (.05,0);
 \end{scope}
 \begin{scope}[shift={(1,3)}]
   %\clip (-.1,-.1) rectangle (0,+.1);
   \draw[->][thick, orangeii] (0,0) to (-0.122,0);
   \draw[thick, white] (0,0) to (-.05,0);
 \end{scope}

 \draw[thick, orangeii] (0+.1,3) to (1-.1,3);

 \begin{scope}[shift={(0,3)}]
   %\clip (-.1,-.1) rectangle (0,+.1);
   \draw[->][thick, gray] (0,0) to (-0.122,0);
   \draw[thick, white] (0,0) to (-.05,0);
 \end{scope}

 \begin{scope}[shift={(2,2)}]
   %\clip (-.1,-.1) rectangle (0,+.1);
   \draw[->][thick, orangeii] (0,0) to (-0.122,0);
   \draw[thick, white] (0,0) to (-.05,0);
 \end{scope}
 \begin{scope}[shift={(1,2)}]
   %\clip (0,-.1) rectangle (.1,+.1);
   \draw[->][thick, orangeii] (0,0) to (0.122,0);
   \draw[thick, white] (0,0) to (.05,0);
 \end{scope}

 \draw[thick, orangeii] (1+.1,2) to (2-.1,2);

 \begin{scope}[shift={(1,2)}]
   %\clip (-.1,-.1) rectangle (0,+.1);
   \draw[->][thick, gray] (0,0) to (-0.122,0);
   \draw[thick, white] (0,0) to (-.05,0);
 \end{scope}

 \draw[ultra thick, darkblue] (0,4) to (0,1);
 \draw[ultra thick, darkblue] (1,4) to (1,1);
 \draw[ultra thick, darkblue] (2,4) to (2,1);

 \draw[ultra thick, white] (0,+.3) to (0,-.3);

 \draw (0,1-.3) node {\small $1$};
 \draw (1,1-.3) node {\small $2$};
 \draw (2,1-.3) node {\small $3$};

 \draw[thick,gray] (-2, 3) to (0-.1,3);
 \draw[thick,gray] (-2, 2) to (1-.1,2);

 \draw[draw=green, fill=green] (-2,3) circle (.085);
 \draw[draw=green, fill=green] (-2,2) circle (.085);

\end{tikzpicture}
}
\right]
\;\;=\;\;
t_{12} \wedge t_{23}
  \;\;\in\;\;
  \mathcal{A}^{{}^{\mathrm{hw}}}_{N_{\mathrm{f}} = 5}
$$
Under the identification \eqref{IdentificationWithHananyWittenConfigurations} this
corresponds to a $\mathrm{D}p$-brane which crosses a
$\mathrm{D}(p+2)$-brane without having broken up into segments.
But the {\it 4T-relation} \eqref{Round4TRelations}
in the graded commutative algebra \eqref{GradedCommutativeAlgebraOfHorizontalChords}
now implies that this configuration equivalently
transmutes to the one on the right of the following:
$$
\left[
\raisebox{-50pt}{

}
\right]
$$
Under the identification \eqref{IdentificationWithHananyWittenConfigurations},
this process is the breaking up of a $\mathrm{D}p$-brane
where it crosses a $\mathrm{D}(p+2)$-brane,
as expected in Hanany-Witten theory.

\medskip

\noindent {\bf (3) The ordering constraint.}
Under the identification
\eqref{IdentificationWithHananyWittenConfigurations}
and by the discussion in \cref{HigherObservablesOnIntersectingBraneModuli},
we obtain the higher observables on Hanany-Witten
$\mathrm{D}p\perp\mathrm{D}(p+2)$-configurations by passing to
weight systems evaluated on the skew-symmetrized horizontal chord
diagrams in \eqref{GradedCommutativeAlgebraOfHorizontalChords}.
By Prop. \ref{FundamentalTheremOfWeightSystems} this introduces
two extra pieces of data, namely:
\begin{enumerate}[{\bf (i)}]
\vspace{-2mm}
\item numbers $N_{\mathrm{c},i}$ of coincident $\mathrm{D}p$-branes
ending on the $i$th strand, and
\vspace{-2mm}
\item winding monodromies $\sigma$ of these
strands, modulo some equivalence relations.
\end{enumerate}
\vspace{-2mm}
But from \eqref{LieAlgebraWeightSystemOnHoizontalChordDiagramsWithStacksOfCoincidentStrands}
it is evident that up to these equivalence relations
only the conjugacy class of the winding monodromy
$\sigma \in \mathrm{Sym}(N_{\mathrm{f}})$ matters, where an equivalence
$$
  \sigma
  \;\;\sim\;\;
  \tilde \sigma \circ \sigma \circ \tilde \sigma^{-1}
$$
corresponds to reordering the strands according to
any other permutation $\tilde \sigma \in \mathrm{Sym}(N_{\mathrm{f}})$.
With the tuple $\vec N_{\mathrm{f}}$ of numbers of coincident
$\mathrm{D}p$-branes specified, this means that we may partially
gauge-fix this freedom in the winding monodromy $\sigma$
by requiring that the elements of $\vec N_{\mathrm{f}}$
are monotonically ordered:
$$
  N_{\mathrm{c},1}
  \;\leq\;
  N_{\mathrm{c},2}
  \;\leq\;
  \cdots
  \;\leq\;
  N_{\mathrm{c},N_{\mathrm{f}}}
  \,.
$$
Under the identification \eqref{IdentificationWithHananyWittenConfigurations}
this is the \emph{ordering constraint} that was found in
\cite[3.5]{GaiottoWitten08}.

\vspace{5mm}
\noindent {\bf Acknowledgements.}
We thank Vincent Braunack-Mayer, Qingtao Chen, Carlo Collari,
David Corfield, and Domenico Fiorenza for discussion.

\vspace{.8cm}

\noindent Hisham Sati, {\it Mathematics, Division of Science, New York University Abu Dhabi, UAE.}

 \medskip
\noindent Urs Schreiber,  {\it Mathematics, Division of Science, New York University Abu Dhabi, UAE, on leave from Czech Academy of Science, Prague.}


\newpage

\begin{thebibliography}{10}

\vspace{-3mm}
\bibitem[AFCS99]{AFCS99}
B. Acharya, J. Figueroa-O'Farrill, C. Hull, and B. Spence,
{\it Branes at conical singularities and holography},
Adv. Theor. Math. Phys. {\bf 2} (1999), 1249-1286,
[\href{https://arxiv.org/abs/hep-th/9808014}{\tt arXiv:hep-th/9808014}].


\vspace{-3mm}
\bibitem[AF96]{AF96}
D. Altschuler and L. Freidel,
{\it Vassiliev knot invariants and Chern-Simons perturbation theory to all orders},
Commun. Math. Phys. {\bf 187} (1997), 261-287,
[\href{https://arxiv.org/abs/q-alg/9603010}{\tt arXiv:q-alg/9603010}].

\vspace{-3mm}
\bibitem[ACI13]{ACI13}
M. Ammon, A. Castro, and N. Iqbal,
{\it Wilson Lines and Entanglement Entropy in Higher Spin Gravity},
J. High Energy Phys. {\bf 10} (2013), 110,
[\href{https://arxiv.org/abs/1306.4338}{\tt arXiv:1306.4338}].


\vspace{-3mm}
\bibitem[AIST17a]{AIST17a}
Y. Asano, G. Ishiki, S. Shimasaki, and S. Terashima,
{it On the transverse M5-branes in matrix theory},
Phys. Rev. {\bf D96} (2017),  126003,
[\href{https://arxiv.org/abs/1701.07140}{\tt arXiv:1701.07140}].

\vspace{-3mm}
\bibitem[AIST17b]{AIST17b}
Y. Asano, G. Ishiki, S. Shimasaki, and S. Terashima,
{\it Spherical transverse M5-branes from the plane wave matrix model},
J. High Energy Phys. {\bf 02} (2018), 076,
[\href{https://arxiv.org/abs/1711.07681}{\tt arXiv:1711.07681}].


\vspace{-3mm}
\bibitem[AC17]{AC17}
B. Assel and S. Cremonesi,
{\it The Infrared Physics of Bad Theories},
SciPost Phys. {\bf 3} (2017) 024,  \newline
[\href{https://arxiv.org/abs/1707.03403}{\tt arXiv:1707.03403}].

\vspace{-3mm}
\bibitem[AH16]{AtiyahHitchin88}
M. F. Atiyah and N. Hitchin,
{\it The Geometry and Dynamics of Magnetic Monopoles},
Princeton University Press, Princeton, NJ, 1988
[\href{https://www.jstor.org/stable/j.ctt7zv206 }{\tt jstor:j.ctt7zv206}].

\vspace{-3mm}
\bibitem[AS93]{AxelrodSinger93}
S. Axelrod and I. Singer,
{\it Chern-Simons Perturbation Theory II},
J. Diff. Geom. {\bf 39} (1994), 173-213,
[\href{https://arxiv.org/abs/hep-th/9304087}{\tt arXiv:hep-th/9304087}].


\vspace{-3mm}
\bibitem[BG98]{BachasGreen98}
C. Bachas and M. Green,
{\it A Classical Manifestation of the Pauli Exclusion Principle},
J. High Energy Phys. {\bf 9801} (1998) 015,
[\href{https://arxiv.org/abs/hep-th/9712187}{\tt arXiv:hep-th/9712187}].

\vspace{-3mm}
\bibitem[BGS97]{BGS97}
C. Bachas, M. Green, and A. Schwimmer,
{\it $(8,0)$ Quantum mechanics and symmetry enhancement in type I' superstrings},
J. High Energy Phys.  {\bf 9801} (1998), 006,
[\href{https://arxiv.org/abs/hep-th/9712086}{\tt arXiv:hep-th/9712086}].

\vspace{-3mm}
\bibitem[BGL16]{BGL16}
J.-B. Bae, D. Gang, and J. Lee,
{\it 3d $\mathcal{N}=2$ minimal SCFTs from Wrapped M5-branes},
J. High Energy Phys. {\bf 08} (2017) 118,
 [\href{https://arxiv.org/abs/1610.09259}{\tt arXiv:1610.09259}].

\vspace{-3mm}
\bibitem[BBGR14]{BBGR14}
A. Bagchi, R. Basu, D. Grumiller, and M. Riegler,
{\it Entanglement entropy in Galilean conformal field theories and flat holography},
Phys. Rev. Lett. {\bf 114} (2015), 111602,
[\href{https://arxiv.org/abs/1410.4089}{\tt arXiv:1410.4089}].

\vspace{-3mm}
\bibitem[BL06]{BaggerLambert06}
J. Bagger and N. Lambert,
{\it Modeling Multiple M2's},
Phys. Rev. {\bf D75} (2007), 045020, \newline
[\href{https://arxiv.org/abs/hep-th/0611108}{\tt arXiv:hep-th/0611108}].

\vspace{-3mm}
\bibitem[BLMP12]{BLMP12}
J. Bagger, N. Lambert, S. Mukhi, and C. Papageorgakis,
{\it Multiple Membranes in M-theory},
Phys. Rep. {\bf 527} (2013), 1-100,
[\href{https://arxiv.org/abs/1203.3546}{\tt arXiv:1203.3546}].


\vspace{-3mm}
\bibitem[Ba97]{Banks97}
T. Banks,
{\it Matrix Theory},
Nucl. Phys. Proc. Suppl. {\bf 67} (1998), 180-224,
[\href{https://arxiv.org/abs/hep-th/9710231}{\tt arXiv:hep-th/9710231}].



\vspace{-3mm}
\bibitem[Ba19]{Banks19}
T. Banks,
{\it On the Limits of Effective Quantum Field Theory:
Eternal Inflation, Landscapes, and Other Mythical Beasts}
[\href{https://arxiv.org/abs/1910.12817}{\tt arXiv:1910.12817}].


\vspace{-3mm}
\bibitem[BDHM98]{BDHM98}
T. Banks, M. Douglas, G. Horowitz, and E. Martinec,
{\it AdS Dynamics from Conformal Field Theory},
[\href{https://arxiv.org/abs/hep-th/9808016}{\tt arXiv:hep-th/9808016}].



\vspace{-3mm}
\bibitem[BFSS96]{BFSS96}
T. Banks, W. Fischler, S. Shenker and L. Susskind,
{\it M Theory As A Matrix Model: A Conjecture},
Phys. Rev. {\bf D55} (1997), 5112-5128,
[\href{https://arxiv.org/abs/hep-th/9610043}{\tt arXiv:hep-th/9610043}].


\vspace{-3mm}
\bibitem[Bar91]{BarNatan91}
D. Bar-Natan,
{\it Perturbative aspects of the Chern-Simons topological quantum field theory}, PhD thesis, Princeton
University, 1991,
[\href{http://inspirehep.net/record/323500}{\tt spire:323500}].

\vspace{-3mm}
\bibitem[Bar95a]{BarNatan95CS}
D. Bar-Natan,
{\it Perturbative Chern-Simons theory},
J. Knot Theory Ram. {\bf 04}  (1995), 503-547, \newline
[\href{https://doi.org/10.1142/S0218216595000247}{\tt doi:10.1142/S0218216595000247}].

\vspace{-3mm}
\bibitem[Bar95b]{BarNatan95}
D. Bar-Natan,
{\it On the Vassiliev knot invariants},
Topology {\bf 34} (1995),  423-472, \newline
[\href{https://doi.org/10.1016/0040-9383(95)93237-2}{\tt doi:10.1016/0040-9383(95)93237-2}].

\vspace{-3mm}
\bibitem[BN96]{BarNatan96}
D. Bar-Natan,
{\it Vassiliev and Quantum Invariants of Braids},
Geom. Topol. Monogr. {\bf 4} (2002) 143-160,
[\href{https://arxiv.org/abs/q-alg/9607001}{\tt arXiv:q-alg/9607001}].

\vspace{-3mm}
\bibitem[BNTT03]{BNTT03}
D. Bar-Natan, L. T. Q. Thang, and D. Thurston,
{\it Two applications of elementary knot theory to Lie algebras and Vassiliev invariants},
Geom. Topol. {\bf 7} (2003), 1-31,
[\href{https://projecteuclid.org/euclid.gt/1513883092}{\tt euclid:euclid.gt/1513883092}].


\vspace{-3mm}
\bibitem[BB04]{BarrettBowcock04}
J. Barrett and P. Bowcock,
{\it Using D-Strings to Describe Monopole Scattering}, \newline
[\href{https://arxiv.org/abs/hep-th/0402163}{\tt arXiv:hep-th/0402163}].

\vspace{-3mm}
\bibitem[BB05]{BarrettBowcock05}
J. Barrett and P. Bowcock,
{\it Using D-Strings to Describe Monopole Scattering - Numerical Calculations},
[\href{https://arxiv.org/abs/hep-th/0512211}{\tt arXiv:hep-th/0512211}].

\vspace{-3mm}
\bibitem[BH05]{BasuHarvey05}
A. Basu and J. Harvey,
{\it The M2-M5 Brane System and a Generalized Nahm's Equation},
Nucl. Phys. {\bf B713} (2005), 136-150,
[\href{https://arxiv.org/abs/hep-th/0412310}{\tt arXiv:hep-th/0412310}]

\vspace{-3mm}
\bibitem[BR15]{BasuRiegler15}
R. Basu and M. Riegler,
{\it Wilson Lines and Holographic Entanglement Entropy in Galilean Conformal Field Theories},
Phys. Rev. {\bf D 93} (2016),  045003,
[\href{https://arxiv.org/abs/1511.08662}{\tt arXiv:1511.08662}].


\vspace{-3mm}
\bibitem[Be83]{Beauville83}
A. Beauville,
{\it Vari\'et\'es K{\"a}hleriennes dont la premiere classe de Chern est nulle},
Jour. Diff. Geom. {\bf 18} (1983), 755-782,
[\href{https://projecteuclid.org/euclid.jdg/1214438181}{\tt euclid.jdg/1214438181}].


\vspace{-3mm}
\bibitem[BBBDN18]{BBBDN18}
C. Beem, D. Ben-Zvi, M. Bullimore, T. Dimofte, and A. Neitzke,
{\it Secondary products in supersymmetric field theory},
[\href{https://arxiv.org/abs/1809.00009}{\tt arXiv:1809.00009}].

\vspace{-3mm}
\bibitem[Bea10]{Bea10}
N. Beisert et. al.,
{\it Review of AdS/CFT Integrability: An Overview},
Lett. Math. Phys. {\bf 99}  (2012), 3-32,
[\href{https://arxiv.org/abs/1012.3982}{\tt arXiv:1012.3982}].

\vspace{-3mm}
\bibitem[BFST03]{BFST03}
N. Beisert, S. Frolov, M. Staudacher, and A. Tseytlin,
{\it Precision Spectroscopy of AdS/CFT},
J. High Energy Phys. {\bf 0310} (2003), 037,
[\href{https://arxiv.org/abs/hep-th/0308117}{\tt arXiv:hep-th/0308117}].


\vspace{-3mm}
\bibitem[BHT18]{BHT18}
I. Bena, P. Heidmann, and D. Turton,
{\it $\mathrm{AdS}_2$ Holography: Mind the Cap},
J. High Energy Phys. {\bf 1812} (2018) 028,
[\href{https://arxiv.org/abs/1806.02834}{\tt arXiv:1806.02834}].

\vspace{-3mm}
\bibitem[BMN02]{BMN02}
D. Berenstein, J. Maldacena,  and H. Nastase,
{\it Strings in flat space and pp waves in $\mathcal{N} = 4$ Super Yang Mills},
J. High Energy Phys. {\bf 0204} (2002) 013,
[\href{https://arxiv.org/abs/hep-th/0202021}{\tt arXiv:hep-th/0202021}].

\vspace{-3mm}
\bibitem[BBdRS01]{BBdRS01}
E. Bergshoeff, A.  Bilal, M. de Roo, and A. Sevrin,
{\it Supersymmetric non-abelian Born-Infeld revisited},
J. High Energy Phys. {\bf 0107} (2001),  029,
[\href{https://arxiv.org/abs/hep-th/0105274}{\tt arXiv:hep-th/0105274}].

\vspace{-3mm}
\bibitem[BNS18]{BNS18}
M. Berkooz, P. Narayan, and J. Sim{\'o}n,
{\it Chord diagrams, exact correlators in spin glasses and black hole bulk
reconstruction}, J. High Energy Phys. {\bf 08} (2018) 192,
[\href{https://arxiv.org/abs/1806.04380}{\tt arXiv:1806.04380}].

\vspace{-3mm}
\bibitem[BINT18]{BINT18}
M. Berkooz, M. Isachenkov, V. Narovlansky,  and G. Torrents,
{\it Towards a full solution of the large $N$ double-scaled SYK model},
J. High Energy Phys. {\bf 03} (2019) 079,
[\href{https://arxiv.org/abs/1811.02584}{\tt arXiv:1811.02584}].


\vspace{-3mm}
\bibitem[Bl04]{Blau04}
M. Blau,
{\it Plane waves and Penrose limits}, lecture notes, 2004-2011, \newline
[\href{http://www.blau.itp.unibe.ch/lecturesPP.pdf}{\tt www.blau.itp.unibe.ch/lecturesPP.pdf}]

\vspace{-3mm}
\bibitem[B{\"o}87]{Boedigheimer87}
C.-F. B{\"o}digheimer,
{\it Stable splittings of mapping spaces},
Algebraic topology, Springer 1987, pp. 174-187,
[\href{https://ncatlab.org/nlab/files/BoedigheimerStableSplittings87.pdf}{\tt ncatlab.org/nlab/files/BoedigheimerStableSplittings87.pdf}]


\vspace{-3mm}
\bibitem[BGR18]{BGR18}
L. Boulton, M. P. Garcia del Moral, and A. Restuccia,
{\it Measure of the potential valleys of the supermembrane theory},
Phys. Lett. {\bf B797} (2019), 134873,
[\href{https://arxiv.org/abs/1811.05758}{\tt arXiv:1811.05758}].

\vspace{-3mm}
\bibitem[Bo18]{Bourdon18}
M. Bourdon,
{\it Mostow type rigidity theorems},
in
{Handbook of Group Actions} (Vol. IV),
%Advanced Lectures in Mathematics  41,
 Ch. 4, pp. 139-188,
International Press, 2018,
[\href{http://math.univ-lille1.fr/~bourdon/papiers/Mostow.pdf}{\tt math.univ-lille1.fr/\~bourdon/papiers/Mostow.pdf}]

\vspace{-3mm}
\bibitem[BSS18]{GaugeEnhancement}
V. Braunack-Mayer, H. Sati, and U. Schreiber,
{\it Gauge enhancement for Super M-branes via Parameterized stable homotopy theory},
Comm. Math. Phys. {\bf 371} (2019), 197-265,
[\href{https://doi.org/10.1007/s00220-019-03441-4}{doi:10.1007/s00220-019-03441-4}], \newline
[\href{https://arxiv.org/abs/1805.05987}{\tt arXiv:1805.05987}][hep-th].


\vspace{-3mm}
\bibitem[BFS19]{BFS19}
T. D. Brennan, C. Ferko, and S. Sethi,
{\it A Non-Abelian Analogue of DBI from $T \bar T$},
[\href{https://arxiv.org/abs/1912.12389}{\tt arXiv:1912.12389}].


\vspace{-3mm}
\bibitem[Br93]{Bry}
J.-L. Brylinski, {\it Loop space, characteristic classes and geometric quantization},
Birkh\"auser Verlag,  (1993),

\vspace{-3mm}
\bibitem[BDG15]{BullimoreDimofteGaiotto15}
M. Bullimore, T. Dimofte, and D. Gaiotto,
{\it The Coulomb Branch of $3d$ $\mathcal{N}=4$ Theories},
Commun. Math. Phys. {\bf 354} (2017), 671-751,
[\href{https://arxiv.org/abs/1503.04817}{\tt arXiv:1503.04817}].

\vspace{-3mm}
\bibitem[BFK18]{BFK18}
M. Bullimore, A. E. V. Ferrari, and H. Kim,
{\it Twisted Indices of 3d $\mathcal{N}=4$ Gauge Theories and Enumerative Geometry of Quasi-Maps},
[\href{https://arxiv.org/abs/1812.05567}{\tt arXiv:1812.05567}].

\vspace{-3mm}
\bibitem[CM96]{CallanMaldacena96}
C. Callan and J. Maldacena,
{\it D-brane Approach to Black Hole Quantum Mechanics},
Nucl. Phys. {\bf B472} (1996), 591-610,
[\href{https://arxiv.org/abs/hep-th/9602043}{\tt arXiv:hep-th/9602043}].

\vspace{-3mm}
\bibitem[CCRL02]{CCRL02}
A. Cattaneo, P. Cotta-Ramusino, and R. Longoni,
{\it Configuration spaces and Vassiliev classes in any dimension},
Algebr. Geom. Topol. {\bf 2} (2002) 949-1000,
[\href{https://arxiv.org/abs/math/9910139}{\tt arXiv:math/9910139}].

\vspace{-3mm}
\bibitem[CS99]{ChalmersSchalm99}
G. Chalmers and K. Schalm,
{\it Holographic Normal Ordering and Multi-particle States in the AdS/CFT Correspondence},
Phys. Rev. {\bf D61} (2000), 046001,
[\href{https://arxiv.org/abs/hep-th/9901144}{\tt arXiv:hep-th/9901144}].

\vspace{-3mm}
\bibitem[CS02]{ChasSullivan02}
M. Chas, D. Sullivan,
{\it Closed string operators in topology leading to Lie bialgebras and higher string algebra},
In: O. A. Laudal , R. Piene (eds.),
{\it The Legacy of Niels Henrik Abel}, Springer, Berlin, Heidelberg,
2004,
[\href{https://arxiv.org/abs/math/0212358}{\tt arXiv:math/0212358}].

\vspace{-3mm}
\bibitem[Ch04]{Chemissany04}
W. Chemissany,
{\it On the way of finding the non-Abelian Born-Infeld theory},
Masters thesis, Groningen, 2004,
[\href{http://inspirehep.net/record/1286212}{\tt spire:1286212}].

\vspace{-3mm}
\bibitem[CS08]{CherkisSaemann08}
S. Cherkis and C. Saemann,
{\it Multiple M2-branes and Generalized 3-Lie algebras},
Phys. Rev. {\bf D78} (2008), 066019,
[\href{https://arxiv.org/abs/0807.0808}{\tt arXiv:0807.0808}].


\vspace{-3mm}
\bibitem[CP18]{ChesterPerlmutter18}
S. M. Chester and E. Perlmutter,
{\it M-Theory Reconstruction from $(2,0)$ CFT and the Chiral Algebra Conjecture},
J. High Energ. Phys. {\bf 2018} (2018) 116,
[\href{https://arxiv.org/abs/1805.00892}{\tt arXiv:1805.00892}].


\vspace{-3mm}
\bibitem[CDM11]{CDM11}
S. Chmutov, S. Duzhin, and J. Mostovoy,
{\it Introduction to Vassiliev knot invariants},
Cambridge University Press, 2012,
[\href{https://arxiv.org/abs/1103.5628}{\tt arXiv:1103.5628}].

\vspace{-3mm}
\bibitem[Ci82]{Cicuta82}
G. M. Cicuta,
{\it Topological Expansion for $\mathrm{SO}(N)$ and $\mathrm{Sp}(2n)$ Gauge Theories},
Lett. Nuovo Cim. {\bf 35} (1982), 87-92,
[\href{https://doi.org/10.1007/BF02754653}{\tt doi:10.1007/BF02754653}].

\vspace{-3mm}
\bibitem[CG01]{CohenGitler01}
F. Cohen and S. Gitler,
{\it Loop spaces of configuration spaces, braid-like groups, and knots},
In: J. Aguaad{\'e}, C. Broto, C. Casacuberta (eds.),
{Cohomological Methods in Homotopy Theory},
Progress in Mathematics, vol 196, Birkh{\"a}user, 2001,
[\href{https://doi.org/10.1007/978-3-0348-8312-2_7}{\tt doi:doi.org/10.1007/978-3-0348-8312-2\_7}].


\vspace{-3mm}
\bibitem[CG02]{CohenGitler02}
F. Cohen and S. Gitler,
{\it On loop spaces of configuration spaces},
Trans. Amer. Math. Soc. {\bf 354} (2002),  1705-1748,
[\href{https://www.jstor.org/stable/2693715}{\tt jstor:2693715}].

\vspace{-3mm}
\bibitem[CG04]{CohenGodin04}
R. Cohen and V. Godin,
{\it A Polarized View of String Topology},
In: G. Segal, U. Tillmann (eds.),
{\it Topology, Geometry and Quantum Field Theory},
LMS, Lecture Notes Series 308, 2004,
[\href{https://arxiv.org/abs/math/0303003}{\tt arXiv:math/0303003}].


\vspace{-3mm}
\bibitem[CV05]{CohenVoronov05}
R. Cohen and A. Voronov,
{\it Notes on string topology},
in:
R. Cohen, K. Hess, A. Voronov,
{\it String topology and cyclic homology},
Advanced Courses in Mathematics CRM Barcelona,
Birkh{\"a}user, 2006,
[\href{https://arxiv.org/abs/math/0503625}{\tt arXiv:math/0503625}].


\vspace{-3mm}
\bibitem[CSS20]{CSS20}
C. Collari, H. Sati, U. Schreiber,
{\it Weight systems which are quantum states},
in preparation.

\vspace{-3mm}
\bibitem[CL02]{ConstableLambert02}
N. Constable and N. Lambert,
{\it Calibrations, Monopoles and Fuzzy Funnels},
Phys. Rev. {\bf D66} (2002), 065016,
[\href{https://arxiv.org/abs/hep-th/0206243}{\tt arXiv:hep-th/0206243}].


\vspace{-3mm}
\bibitem[CMT99]{CMT99}
N. Constable, R. Myers, and O. Tafjord,
{\it The Noncommutative Bion Core},
Phys. Rev. {\bf D61} (2000), 106009,
[\href{https://arxiv.org/abs/hep-th/9911136}{\tt arXiv:hep-th/9911136}].


\vspace{-3mm}
\bibitem[Co17]{Corfield17}
D. Corfield,
{\it Duality as a category-theoretic concept},
Studies in History and Philosophy of Modern Physics
{\bf 59} (2017),  55-61
[\href{https://doi.org/10.1016/j.shpsb.2015.07.004}{\tt doi:10.1016/j.shpsb.2015.07.004}].


\vspace{-3mm}
\bibitem[CJS78]{CJS78}
E. Cremmer, B. Julia, and J. Scherk,
{\it Supergravity Theory in Eleven-Dimensions},
Phys. Lett. {\bf B76} (1978), 409--412,
[\href{http://inspirehep.net/record/129517}{\tt spire:129517}].


\vspace{-3mm}
\bibitem[CHZ14]{CHZ14}
S. Cremonesi, A. Hanany, and A. Zaffaroni,
{\it Monopole operators and Hilbert series of Coulomb branches
of 3d $\mathcal{N} = 4$ gauge theories}, J. High Energy Phys. {\bf 01} (2014) 005,
[\href{https://arxiv.org/abs/1309.2657}{\tt arXiv:1309.2657}].


\vspace{-3mm}
\bibitem[Cv76]{Cvitanovic76}
P. Cvitanovi{\'c},
{\it Group theory for Feynman diagrams in non-Abelian gauge theories},
Phys. Rev. {\bf D14} (1976), 1536-1553,
[\href{https://journals.aps.org/prd/abstract/10.1103/PhysRevD.14.1536}{\tt doi:10.1103/PhysRevD.14.1536}].

\vspace{-3mm}
\bibitem[DvR18]{DanielssonVanRiet18}
U. Danielsson and T. Van Riet,
{\it What if string theory has no de Sitter vacua?},
Int. J.  Mod. Phys. {\bf  D 27} (2018), 1830007,
[\href{https://arxiv.org/abs/1804.01120}{\tt arXiv:1804.01120}].

\vspace{-3mm}
\bibitem[DNP02]{DasguptaNicolaiPleftka02}
A. Dasgupta, H. Nicolai, and J. Plefka,
{\it An Introduction to the Quantum Supermembrane},
Grav. Cosmol. {\bf 8} (2002), 1; Rev. Mex. Fis. {\bf 49S1} (2003), 1-10,
[\href{https://arxiv.org/abs/hep-th/0201182}{\tt arXiv:hep-th/0201182}].


\vspace{-3mm}
\bibitem[DSJVR02]{DSJVR02}
K. Dasgupta, M. M. Sheikh-Jabbari, and M. Van Raamsdonk,
{\it Matrix Perturbation Theory For M-theory On a PP-Wave},
J. High Energy Phys. {\bf  0205} (2002), 056,
[\href{https://arxiv.org/abs/hep-th/0205185}{\tt arXiv:hep-th/0205185}].

\vspace{-3mm}
\bibitem[De02]{Deligne02}
P. Deligne,
{\it Cat{\'e}gorie Tensorielle},
Moscow Math. J. {\bf 2} (2002), 227-248, \newline
[\href{https://www.math.ias.edu/files/deligne/Tensorielles.pdf}{\tt www.math.ias.edu/files/deligne/Tensorielles.pdf}]

\vspace{-3mm}
\bibitem[dBHOO97]{dBHOO97}
J. de Boer, K. Hori, H. Ooguri, and Y. Oz,
{\it Mirror Symmetry in Three-Dimensional Gauge Theories, Quivers and D-branes},
Nucl. Phys. {\bf B493} (1997), 101-147,
[\href{https://arxiv.org/abs/hep-th/9611063}{\tt arXiv:hep-th/9611063}].

\vspace{-3mm}
\bibitem[dBHOO97]{dBHOOY96}
J. de Boer, K. Hori, H. Ooguri, Y. Oz, and Z. Yin,
{\it Mirror Symmetry in Three-Dimensional Gauge Theories,
$\mathrm{SL}(2,\mathbb{Z})$ and D-Brane Moduli Spaces},
Nucl. Phys. {\bf B493} (1997), 148-176,  \newline
[\href{https://arxiv.org/abs/hep-th/9612131}{\tt arXiv:hep-th/9612131}].


\vspace{-3mm}
\bibitem[DHMB15]{DHMB15}
S. De Haro, D. R. Mayerson, and J. Butterfield,
{\it Conceptual Aspects of Gauge/Gravity Duality},
Found. Phys. {\bf 46} (2016), 1381-1425,
[\href{https://arxiv.org/abs/1509.09231}{\tt arXiv:1509.09231}].


\vspace{-3mm}
\bibitem[dMFO10]{MF10}
P. de Medeiros and J. Figueroa-O'Farrill,
{\it Half-BPS M2-brane orbifolds},
Adv. Theor. Math. Phys. {\bf 16} (2012), 1349-1408,
[\href{https://arxiv.org/abs/1007.4761}{\tt arXiv:1007.4761}].

\vspace{-3mm}
\bibitem[dMFMR09]{dMFMR09}
P. de Medeiros, J. Figueroa-O'Farrill, E. M{\'e}ndez-Escobar, and P.  Ritter,
{\it On the Lie-algebraic origin of metric 3-algebras},
Commun. Math. Phys. {\bf 290} (2009), 871-902,
[\href{https://arxiv.org/abs/0809.1086}{\tt arXiv:0809.1086}].



\vspace{-3mm}
\bibitem[dWLN89]{dWLN89}
B. de Wit, M. L{\"u}scher, and H. Nicolai,
{\it The Supermembrane Is Unstable},
Nucl. Phys. {\bf B320} (1989), 135-159,
[\href{https://doi.org/10.1016/0550-3213(89)90214-9}{\tt doi:10.1016/0550-3213(89)90214-9}].


\vspace{-3mm}
\bibitem[Di97]{Diaconescu97}
D. Diaconescu,
{\it D-branes, Monopoles, and Nahm Equations},
Nucl. Phys. {\bf B503} (1997) 220-238, \newline
[\href{https://arxiv.org/abs/hep-th/9608163}{\tt arXiv:hep-th/9608163}].



\vspace{-3mm}
\bibitem[DFM03]{DFM03}
E. Diaconescu, D. S. Freed, and G. Moore,
{\it The M-theory 3-form and $E_8$ gauge theory},
 Elliptic Cohomology, 44-88,
%London Math. Soc. Lecture Note Series 342,
Cambridge University Press, 2007,
[\href{https://arxiv.org/abs/hep-th/0312069}{\tt arXiv:hep-th/0312069}].


\vspace{-3mm}
\bibitem[DMW00a]{DMW00a}
D. Diaconescu,  G. Moore, and E. Witten,
{\it $E_8$-gauge theory and a derivation of K-theory from M-theory},
Adv. Theor. Math. Phys {\bf 6} (2003), 1031--1134,
[\href{https://arxiv.org/abs/hep-th/0005090}{\tt arXiv:hep-th/0005090}].

\vspace{-3mm}
\bibitem[DMW00b]{DMW00b}
D. Diaconescu, G. Moore, and E. Witten,
{\it A Derivation of K-Theory from M-Theory}, \newline
[\href{https://arxiv.org/abs/hep-th/0005091}{\tt arXiv:hep-th/0005091}].





\vspace{-3mm}
\bibitem[DP17]{DP17}
G. Dibitetto and N. Petri,
{\it 6d surface defects from massive type IIA},
J. High Energy Phys. {\bf 01} (2018) 039, \newline
[\href{https://arxiv.org/abs/1707.06154}{\tt arXiv:1707.06154}].

\vspace{-3mm}
\bibitem[DP19]{DP19}
G. Dibitetto and N. Petri,
{\it $\mathrm{AdS}_3$ vacua and surface defects in massive IIA},
[\href{https://arxiv.org/abs/1904.02455}{\tt arXiv:1904.02455}].

\vspace{-3mm}
\bibitem[Di14]{Dimofte14}
T. Dimofte,
{\it 3d Superconformal Theories from Three-Manifolds},
In: J. Teschner (ed.),
{Exact Results on $\mathcal{N} = 2$ Supersymmetric Gauge Theories}, Springer, 2015,
[\href{https://arxiv.org/abs/1412.7129}{\tt arXiv:1412.7129}].

\vspace{-3mm}
\bibitem[Do84]{Donaldson84}
S. K. Donaldson,
{\it Nahm's equations and the classification of monopoles},
Comm. Math. Phys. {\bf 96} (1984), 387-407.

\vspace{-3mm}
\bibitem[DGKV10]{DGKV10}
A. Donos, J. Gauntlett, N. Kim, and O. Varelam,
{\it Wrapped M5-branes, consistent truncations and \newline AdS/CMT},
J. High Energy Phys. {\bf 1012} (2010), 003,
[\href{https://arxiv.org/abs/1009.3805}{\tt arXiv:1009.3805}].


\vspace{-3mm}
\bibitem[DKMTV97]{DKMTV97}
N. Dorey, V. V. Khoze, M. P. Mattis, D. Tong, and S. Vandoren,
{\it Instantons, Three-Dimensional Gauge Theory, and the Atiyah-Hitchin Manifold},
Nucl. Phys. {\bf B502} (1997), 59-93,  \newline
[\href{https://arxiv.org/abs/hep-th/9703228}{\tt arXiv:hep-th/9703228}].

\vspace{-3mm}
\bibitem[Du96]{Duff96}
M. Duff,
{\it M-Theory (the Theory Formerly Known as Strings)},
Int. J. Mod. Phys. {\bf A11} (1996), 5623-5642,
[\href{https://arxiv.org/abs/hep-th/9608117}{\tt arXiv:hep-th/9608117}].

\vspace{-3mm}
\bibitem[Du98]{Duff98}
M. Duff,
{\it A Layman's Guide to M-theory},
Abdus Salam Memorial Meeting, Trieste, Italy, 19 - 22 Nov 1997, pp.184-213,
[\href{https://arxiv.org/abs/hep-th/9805177}{\tt arXiv:hep-th/9805177}].

\vspace{-3mm}
\bibitem[Du99]{Duff99B}
M. Duff (ed.),
{\it The World in Eleven Dimensions: Supergravity, Supermembranes and M-theory},
Institute of Physics Publishing, Bristol, 1999.

\vspace{-3mm}
\bibitem[Du19]{Duff19}
M. Duff,
in: G. Farmelo,
{\it The Universe Speaks in numbers}, interview 14, 2019,
\newline
[\href{https://grahamfarmelo.com/the-universe-speaks-in-numbers-interview-14}{\tt grahamfarmelo.com/the-universe-speaks-in-numbers-interview-14}]
at 17:14.


\vspace{-3mm}
\bibitem[FH01]{FadellHusseini01}
E. Fadell and S. Husseini,
{\it Geometry and topology of configuration spaces},
Springer, New York, 2001, \newline
[\href{https://link.springer.com/book/10.1007/978-3-642-56446-8}
{\tt https://link.springer.com/book/10.1007/978-3-642-56446-8}].

\vspace{-3mm}
\bibitem[Fau73]{Faulkner73}
J. Faulkner, {\it On the geometry of inner ideals},
J. Algebra {\bf 26} (1973),  1-9, \newline
[\href{https://doi.org/10.1016/0021-8693(73)90032-X}{\tt doi:10.1016/0021-8693(73)90032-X}].


\vspace{-3mm}
\bibitem[Faz17]{Fazzi17}
M. Fazzi,
{\it Higher-dimensional field theories from type II supergravity},
[\href{https://arxiv.org/abs/1712.04447}{\tt arXiv:1712.04447}].


\vspace{-3mm}
\bibitem[FKV97]{FKV97}
J. Figueroa-O'Farrill, T. Kimura, and A. Vaintrob,
{\it The universal Vassiliev invariant for the Lie superalgebra $\mathfrak{gl}(1\vert 1)$},
Commun. Math. Phys. {\bf 185} (1997), 93-127,
[\href{https://arxiv.org/abs/q-alg/9602014}{\tt arXiv:q-alg/9602014}].

\vspace{-3mm}
\bibitem[Fi19]{Filippas19}
K. Filippas,
{\it Non-integrability on $\mathrm{AdS}_3$ supergravity},
[\href{https://arxiv.org/abs/1910.12981}{\tt arXiv:1910.12981}].

\vspace{-3mm}
\bibitem[FSS13]{FSS13}
 D. Fiorenza, H. Sati, and U. Schreiber,
\newblock {\it Super Lie $n$-algebra extensions, higher WZW models,
and super $p$-branes with tensor multiplet fields},
\href{http://www.worldscientific.com/doi/abs/10.1142/S0219887815500188}
{Intern. J. Geom. Meth. Mod. Phys. {\bf 12} (2015) 1550018}, \newline
[\href{http://arxiv.org/abs/1308.5264}{\tt arXiv:1308.5264}].


\vspace{-3mm}
\bibitem[FSS14a]{FSS14a}
D. Fiorenza, H. Sati, and U. Schreiber,
{\it The $E_8$ moduli 3-stack of the $C$-field},
Commun. Math. Phys. {\bf 333} (2015),  117-151,
[\href{https://arxiv.org/abs/1202.2455}{\tt arXiv:1202.2455}].


\vspace{-3mm}
\bibitem[FSS15]{FSS15}
D. Fiorenza, H. Sati and U. Schreiber,
{\it The WZW term of the M5-brane and differential cohomotopy},
J. Math. Phys. {\bf 56} (2015), 102301,
[\href{https://arxiv.org/abs/1506.07557}{\tt arXiv:1506.07557}].

\vspace{-3mm}
\bibitem[FSS16a]{FSS16a}
D. Fiorenza, H. Sati and U. Schreiber,
{\it Rational sphere valued supercocycles in M-theory and type IIA string theory},
J. Geom. Phys. {\bf 114} (2017) 91-108,
[\href{https://arxiv.org/abs/1606.03206}{\tt arXiv:1606.03206}].


\vspace{-3mm}
\bibitem[FSS16b]{FSS16b}
D. Fiorenza, H. Sati, and U. Schreiber,
{\it T-Duality from super Lie $n$-algebra cocycles for super p-branes},
Adv. Theor. Math. Phys. {\bf 22} (2018), 1209--1270,
[\href{https://arxiv.org/abs/1611.06536}{\tt arXiv:1611.06536}].


\vspace{-3mm}
\bibitem[FSS19a]{FSS19a}
D. Fiorenza, H. Sati,  and U. Schreiber,
\href{https://ncatlab.org/schreiber/show/The+rational+higher+structure+of+M-theory}
{\it The rational higher structure of M-theory},
Proc. LMS-EPSRC Durham Symposium
{\it Higher Structures in M-Theory}, Aug. 2018,
Fortsch. Phys., 2019, \newline
[\href{https://doi.org/10.1002/prop.201910017}{\tt doi:10.1002/prop.201910017}]
[\href{https://arxiv.org/abs/1903.02834}{\tt arXiv:1903.02834}].

\vspace{-3mm}
\bibitem[FSS19b]{FSS19b}
D. Fiorenza, H. Sati, and U. Schreiber,
{\it Twisted Cohomotopy implies M-Theory anomaly cancellation on 8-manifolds},
Comm. Math. Phys. 2020 (in print)
[\href{https://arxiv.org/abs/1904.10207}{\tt arXiv:1904.10207}].

\vspace{-3mm}
\bibitem[FSS19c]{FSS19c}
D. Fiorenza, H. Sati, and U. Schreiber,
{\it Twisted Cohomotopy implies level quantization of the full 6d Wess-Zumino-term of the M5-brane}
[\href{https://arxiv.org/abs/1906.07417}{\tt arXiv:1906.07417}].


\vspace{-3mm}
\bibitem[FSS19d]{FSS19d}
D. Fiorenza, H. Sati and U. Schreiber,
{\it Super-exceptional geometry: origin of heterotic M-theory and super-exceptional embedding construction of M5},
JHEP 2020 (in print)
[\href{https://arxiv.org/abs/1908.00042}{\tt arXiv:1908.00042}].


\vspace{-3mm}
\bibitem[Fr00]{Freed00}
D. Freed,
{\it Dirac charge quantization and generalized differential cohomology},
Surv. Diff. Geom. {\bf 7}, 129--194, Int. Press, Somerville, MA, 2000,
[\href{http://arxiv.org/abs/hep-th/0011220}{\tt arXiv:hep-th/0011220}].


\vspace{-3mm}
\bibitem[FKO17]{FKP17}
D. Futer, E. Kalfagianni, and J. S. Purcell,
{\it A survey of hyperbolic knot theory},
Springer Proceedings in Mathematics \& Statistics, vol. 284 (2019), 1-30,
[\href{https://arxiv.org/abs/1708.07201}{\tt arXiv:1708.07201}].


\vspace{-3mm}
\bibitem[GR03]{GaiottoRastelli03}
D. Gaiotto and L. Rastelli,
{\it A paradigm of open/closed duality: Liouville D-branes and the Kontsevich model},
J. High Energy Phys.  {\bf 0507} (2005), 053,
[\href{https://arxiv.org/abs/hep-th/0312196}{\tt arXiv:hep-th/0312196}].


\vspace{-3mm}
\bibitem[GT14]{GaiottoTomasiello14}
D. Gaiotto and A. Tomasiello,
{\it Holography for $(1,0)$ theories in six dimensions}
J. High Energy Phys. {\bf 12} (2014), 003,
[\href{https://arxiv.org/abs/1404.0711}{\tt arXiv:1404.0711}].


\vspace{-3mm}
\bibitem[GW08]{GaiottoWitten08}
D. Gaiotto and  E. Witten,
{\it Supersymmetric Boundary Conditions in $\mathcal{N}=4$ Super Yang-Mills Theory},
J. Stat. Phys. {\bf 135} (2009) 789-855,
[\href{https://arxiv.org/abs/0804.2902}{\tt arXiv:0804.2902}].


\vspace{-3mm}
\bibitem[GK19]{GK19}
D. Gang and N. Kim,
{\it Large $N$ twisted partition functions in 3d-3d correspondence and Holography},
Phys. Rev. {\bf D 99} (2019),  021901,
[\href{https://arxiv.org/abs/1808.02797}{\tt arXiv:1808.02797}].


\vspace{-3mm}
\bibitem[GKL14a]{GangKimLee14a}
D. Gang, N. Kim, and S. Lee,
{\it Holography of Wrapped M5-branes and Chern-Simons theory},
Phys. Lett. {\bf B 733} (2014), 316-319,
[\href{https://arxiv.org/abs/1401.3595}{\tt	arXiv:1401.3595}] [hep-th].

\vspace{-3mm}
\bibitem[GKL14b]{GangKimLee14b}
D. Gang, N. Kim and S. Lee,
{\it Holography of 3d-3d correspondence at Large $N$},
J. High Energy Phys. {\bf 04} (2015) 091,
[\href{https://arxiv.org/abs/1409.6206}{\tt arXiv:1409.6206}].

\vspace{-3mm}
\bibitem[GGJV18]{GGJV18}
A. M. Garc{\'i}a-Garc{\'i}a, Y Jia,  and J. J. M. Verbaarschot,
{\it Exact moments of the Sachdev-Ye-Kitaev model up to order $1/N^2$},
J. High Energy Phys. {\bf 04} (2018), 146,
[\href{https://arxiv.org/abs/1801.02696}{\tt arXiv:1801.02696}].

\vspace{-3mm}
\bibitem[GKW00]{GKW00}
J. Gauntlett, N. Kim, and D. Waldram,
{\it M-Fivebranes Wrapped on Supersymmetric Cycles},
Phys. Rev. {\bf D63} (2001) 126001,
[\href{https://arxiv.org/abs/hep-th/0012195}{\tt arXiv:hep-th/0012195}].

\vspace{-3mm}
\bibitem[GK99]{GiveonKutasov99}
A. Giveon and D. Kutasov,
{\it Brane Dynamics and Gauge Theory},
Rev. Mod. Phys. {\bf 71} (1999), 983-1084,
[\href{https://arxiv.org/abs/hep-th/9802067}{\tt arXiv:hep-th/9802067}].


\vspace{-3mm}
\bibitem[GB99]{GopakumarVafa99}
R. Gopakumar and C. Vafa,
{\it On the Gauge Theory/Geometry Correspondence},
Adv. Theor. Math. Phys. 3 (1999) 1415-1443
[\href{https://arxiv.org/abs/hep-th/9811131}{arXiv:hep-th/9811131}]

\vspace{-3mm}
\bibitem[GKSTY01]{GKSTY01}
E. Gorbatov, V.S. Kaplunovsky, J. Sonnenschein, S. Theisen, and S. Yankielowicz,
{\it On Heterotic Orbifolds, M Theory and Type I' Brane Engineering},
 J. High Energy Phys. {\bf 0205} (2002), 015, \newline
 [\href{https://arxiv.org/abs/hep-th/0108135}{\tt arXiv:hep-th/0108135}].


\vspace{-3mm}
\bibitem[GZZ09]{GZZ09}
A. Gorsky, V. Zakharov, and A. Zhitnitsky,
{\it On Classification of QCD defects via holography},
Phys. Rev. {\bf D79} (2009), 106003,
[\href{https://arxiv.org/abs/0902.1842}{\tt arXiv:0902.1842}].

\vspace{-3mm}
\bibitem[GS19]{GS}
D. Grady and H. Sati,
{\it Ramond-Ramond fields and twisted differential K-theory},
[\href{https://arxiv.org/abs/1903.08843}{\tt arXiv:1903.08843}].
% [hep-th].

\vspace{-3mm}
\bibitem[GMM89]{GMM89}
E. Guadagnini, E. Martellini, and M. Mintchev,
{\it Chern-Simons field theory and link invariants},
{ Knots, Topology and Quantum Field Theories},
%Proceedings of the Johns Hopkins Workshop on Current Problems in Particle Theory 13, Florence (1989),
World Scientific, Singapore, 1989, \newline
[\href{https://doi.org/10.1142/9789814540742}{\tt doi:10.1142/9789814540742}].



\vspace{-3mm}
\bibitem[GKP02]{GKP02}
S. Gubser, I. Klebanov, and A. Polyakov,
{\it A semi-classical limit of the gauge/string correspondence},
Nucl. Phys. {\bf B636} (2002), 99-114,
[\href{https://arxiv.org/abs/hep-th/0204051}{\tt arXiv:hep-th/0204051}].

\vspace{-3mm}
\bibitem[Gu05]{Gukov05}
S. Gukov,
{\it Three-Dimensional Quantum Gravity, Chern-Simons Theory, and the A-Polynomial},
Commun. Math. Phys. {\bf 255} (2005), 577-627,
[\href{https://arxiv.org/abs/hep-th/0306165}{\tt arXiv:hep-th/0306165}].

\vspace{-3mm}
\bibitem[GMMS04]{GMMS04}
S. Gukov, E. Martinec, G. Moore, and A. Strominger,
{\it Chern-Simons Gauge Theory and the $\mathrm{AdS}_3/\mathrm{CFT}_2$ Correspondence},
in: M. Shifman et al. (eds.),
{From fields to strings: Circumnavigating theoretical physics},
vol. 2, 1606-1647, 2004,
[\href{https://arxiv.org/abs/hep-th/0403225}{\tt arXiv:hep-th/0403225}].


\vspace{-3mm}
\bibitem[Ha92]{Halperin92}
S. Halperin,
{\it Universal enveloping algebras and loop space homology},
J. Pure Appl. Algebra {\bf 83} (1992),  237-282,
[\href{https://doi.org/10.1016/0022-4049(92)90046-I}{\tt doi:10.1016/0022-4049(92)90046-I}].

\vspace{-3mm}
\bibitem[HW97]{HananyWitten97}
A. Hanany and E. Witten,
{\it Type IIB Superstrings, BPS Monopoles, And Three-Dimensional Gauge Dynamics},
Nucl. Phys. {\bf B 492} (1997), 152-190,
[\href{https://arxiv.org/abs/hep-th/9611230}{\tt arXiv:hep-th/9611230}].

\vspace{-3mm}
\bibitem[HZ99]{HananyZaffaroni99}
A. Hanany and A. Zaffaroni,
{\it Monopoles in String Theory}, J. High Energy Phys. {\bf 9912} (1999) 014, \newline
[\href{https://arxiv.org/abs/hep-th/9911113}{\tt arXiv:hep-th/9911113}].

\vspace{-3mm}
\bibitem[HT97]{HashimotoTaylor97}
A. Hashimoto and W. Taylor,
{\it Fluctuation Spectra of Tilted and Intersecting D-branes from the Born-Infeld Action},
Nucl. Phys. {\bf B503} (1997),  193-219,
[\href{https://arxiv.org/abs/hep-th/9703217}{\tt arXiv:hep-th/9703217}].

\vspace{-3mm}
\bibitem[HKLY19]{HKLY}
H. Hayashi, S.-S. Kim, K. Lee, and F. Yagi,
{\it 6d SCFTs, 5d Dualities and Tao Web Diagrams},
J. High Energy Phys.
{\bf 05} (2019) 203,
[\href{https://arxiv.org/abs/1509.03300}{\tt  arXiv:1509.03300}] [hep-th].

\vspace{-3mm}
\bibitem[He18]{Heras18}
R. Heras,
{\it Dirac quantisation condition: a comprehensive review},
Contemp. Phys. {\bf 59} (2018), 331-355, \newline
[\href{https://arxiv.org/abs/1810.13403}{\tt arXiv:1810.13403}].

\vspace{-3mm}
\bibitem[HS99]{HitchinSawon99}
N. Hitchin and J. Sawon,
{\it Curvature and characteristic numbers of hyperk{\"a}hler manifolds},
Duke Math. J. {\bf 106} (2001), 599-615,
[\href{https://arxiv.org/abs/math/9908114}{\tt arXiv:math/9908114}].

\vspace{-3mm}
\bibitem[DHK19]{DHokerKraus19}
E. D'Hoker and P. Kraus,
{\it Gravitational Wilson lines in $\mathrm{AdS}_3$},
[\href{https://arxiv.org/abs/1912.02750}{\tt arXiv:1912.02750}].

\vspace{-3mm}
\bibitem[HLPY08]{HLPY08}
D. K. Hong, K.-M. Lee, C. Park, and H.-U. Yee,
{\it Holographic Monopole Catalysis of Baryon Decay},
J. High Energy Phys. {\bf 0808} (2008), 018,
[\href{https://arxiv.org/abs/0804.1326}{\tt arXiv:0804.1326}].


\vspace{-3mm}
\bibitem[HS05]{HopkinsSinger05}
M. J. Hopkins and I. M. Singer,
{\it Quadratic functions in geometry, topology, and M-theory},
J. Differential Geom. {\bf 70} (3) (2005), 329--452,
[\href{https://arxiv.org/abs/math/0211216}{\tt arXiv:math/0211216}].


\vspace{-3mm}
\bibitem[HOO98]{HoriOoguriOz98}
K. Hori, H. Ooguri, and Y .Oz,
{\it Strong Coupling Dynamics of Four-Dimensional $\mathcal{N}=1$
Gauge Theories from M Theory Fivebrane},
Adv. Theor. Math. Phys. {\bf 1} (1998), 1-52,
[\href{https://arxiv.org/abs/hep-th/9706082}{\tt arXiv:hep-th/9706082}].


\vspace{-3mm}
\bibitem[HLW98]{HoweLambertWest97}
P. S. Howe, N. D. Lambert, and P. C. West,
{\it The Self-Dual String Soliton},
Nucl. Phys. {\bf B515} (1998),  203-216,
[\href{https://arxiv.org/abs/hep-th/9706082}{\tt arXiv:hep-th/9709014}].


\vspace{-3mm}
\bibitem[HMS19]{HMS19}
X. Huang, C.-T. Ma, and H. Shu,
{\it Quantum Correction of the Wilson Line and Entanglement Entropy
in the $\mathrm{AdS}_3$ Chern-Simons Gravity Theory},
[\href{https://arxiv.org/abs/1911.03841}{\tt arXiv:1911.03841}].


\vspace{-3mm}
\bibitem[HSS19]{ADE}
J. Huerta, H. Sati,  and  U. Schreiber,
{\it Real ADE-equivariant (co)homotopy and Super M-branes},
Commun.  Math.  Phys.  {\bf 371} (2019), 425-524,
[\href{https://arxiv.org/abs/1805.05987}{\tt arXiv:1805.05987}].


\vspace{-3mm}
\bibitem[IU12]{IU12}
L. Ib{\'a}{\~n}ez and A. Uranga,
{\it String Theory and Particle Physics: An Introduction to String Phenomenology},
Cambridge University Press, 2012,
[\href{https://doi.org/10.1017/CBO9781139018951}{\tt doi:10.1017/CBO9781139018951}].

\vspace{-3mm}
\bibitem[In99]{Intriligator99}
K. Intriligator,
{\it Compactified Little String Theories and Compact Moduli Spaces of Vacua},
Phys. Rev. {\bf D61} (2000), 106005,
[\href{https://arxiv.org/abs/hep-th/9909219}{\tt arXiv:hep-th/9909219}].

\vspace{-3mm}
\bibitem[JM19]{JacksonMoffat19}
D. Jackson and I. Moffat,
{\it An Introduction to Quantum and Vassiliev Knot Invariants},
Springer, 2019,
[\href{https://link.springer.com/book/10.1007/978-3-030-05213-3}{\tt doi:10.1007/978-3-030-05213-3}].


\vspace{-3mm}
\bibitem[Je10]{Jensen10}
K. Jensen,
{\it Chiral anomalies and AdS/CMT in two dimensions},
J. High Energy Phys. {\bf 1101} (2011), 109,
[\href{https://arxiv.org/abs/1012.4831}{\tt arXiv:1012.4831}].

\vspace{-3mm}
\bibitem[JV18]{JiaVerbaarschot18}
Y. Jia and J. J. M. Verbaarschot,
{\it Large $N$ expansion of the moments and free energy of Sachdev-Ye-Kitaev
model, and the enumeration of intersection graphs},
J. High Energy Phys. {\bf 11} (2018) 031, \newline
[\href{https://arxiv.org/abs/1806.03271}{\tt arXiv:1806.03271}].

\vspace{-3mm}
\bibitem[Kac77]{Kac77}
V. Kac, {\it Lie superalgebras}, Adv. Math. {\bf 26} (1977),  8-96,
[\href{https://doi.org/10.1016/0001-8708(77)90017-2}{\tt doi:10.1016/0001-8708(77)90017-2}].

\vspace{-3mm}
\bibitem[Kal98]{Kallel98}
S. Kallel,
{\it Particle Spaces on Manifolds and Generalized Poincar{\'e} Dualities},
 Quarterly J. Math. {\bf 52} (2001), 45-70,
[\href{https://arxiv.org/abs/math/9810067}{\tt arXiv:math/9810067}].


\vspace{-3mm}
\bibitem[Ka99]{Kapranov99}
M. Kapranov,
{\it Rozansky-Witten invariants via Atiyah classes},
Compositio Math. {\bf 115} (1999),  71-113,
[\href{https://arxiv.org/abs/alg-geom/9704009}{\tt arXiv:alg-geom/9704009}].

\vspace{-3mm}
\bibitem[Kar98]{Karch98}
A. Karch,
{\it Field Theory Dynamics from Branes in String Theory},
PhD thesis,  Humboldt University, Berlin, 1998,
[\href{https://arxiv.org/abs/hep-th/9812072}{\tt arXiv:hep-th/9812072}].

\vspace{-3mm}
\bibitem[Ka96]{Kashaev96}
R. Kashaev,
{\it The Hyperbolic Volume Of Knots From The Quantum Dilogarithm},
Lett. Math. Phys. {\bf 39} (1997), 269-275,
[\href{https://arxiv.org/abs/q-alg/9601025}{\tt arXiv:q-alg/9601025}].


\vspace{-3mm}
\bibitem[KV97]{KatzVafa97}
S. Katz and C. Vafa,
{\it Geometric Engineering of $N=1$ Quantum Field Theories},
Nucl. Phys. {\bf B497} (1997), 196-204,
[\href{https://arxiv.org/abs/hep-th/9611090}{\tt arXiv:hep-th/9611090}].

\vspace{-3mm}
\bibitem[Ke14]{Keranen14}
V. Keranen,  {\it Chern-Simons interactions in $\mathrm{AdS}_3$ and the
current conformal block}, [\href{https://arxiv.org/abs/1403.6881}{\tt arXiv:1403.6881}].


\vspace{-3mm}
\bibitem[KKP03]{KKP03}
N. Kim, T. Klose, and J. Plefka,
{\it Plane-wave Matrix Theory from $\mathcal{N}=4$ Super Yang-Mills on $\mathbb{R} \times S^3$},
Nucl. Phys. {\bf B671} (2003), 359-382,
[\href{https://arxiv.org/abs/hep-th/0306054}{\tt arXiv:hep-th/0306054}].

\vspace{-3mm}
\bibitem[Kn97]{Kneissler97}
J. Kneissler,
{\it The number of primitive Vassiliev invariants up to degree 12},
[\href{https://arxiv.org/abs/q-alg/9706022}{\tt arXiv:q-alg/9706022}].

\vspace{-3mm}
\bibitem[Koh87]{Kohno87}
T. Kohno,
{\it Monodromy representations of braid groups and Yang-Baxter equations}, Ann. l'Institut Fourier
{\bf 37} (1987), 139-160,
[\href{https://doi.org/10.5802/aif.1114}{\tt doi:10.5802/aif.1114}].

\vspace{-3mm}
\bibitem[Koh02]{Kohno02}
T. Kohno,
{\it Loop spaces of configuration spaces and finite type invariants},
Geom. Topol. Monogr. {\bf 4} (2002), 143-160,
[\href{https://arxiv.org/abs/math/0211056}{\tt arXiv:math/0211056}].


\vspace{-3mm}
\bibitem[Kon92]{Kontsevich92}
M. Kontsevich,
{\it Feynman diagrams and low-dimensional topology},
First European Congress of Mathematics, 1992,
Paris, vol. II, Progress in Mathematics 120, Birkha\"user, 1994, 97-121.


\vspace{-3mm}
\bibitem[Kon93]{Kontsevich93}
M. Kontsevich,
{\it Vassiliev's knot invariants},
Adv. Soviet Math. {\bf 16} (1993), 137-150, \newline
[\href{http://pagesperso.ihes.fr/\~maxim/TEXTS/VassilievKnot.pdf}{\tt http://pagesperso.ihes.fr/$\sim$maxim/TEXTS/VassilievKnot.pdf}]

\vspace{-3mm}
\bibitem[Kr00]{Krasnov00}
K. Krasnov,
{\it Holography and Riemann Surfaces},
Adv. Theor. Math. Phys. {\bf 4} (2000), 929-979, \newline
[\href{https://arxiv.org/abs/hep-th/0005106}{\tt arXiv:hep-th/0005106}].

\vspace{-3mm}
\bibitem[Kr01]{Krasnov01}
K. Krasnov,
{\it Analytic Continuation for Asymptotically AdS 3D Gravity},
Class. Quant. Grav. {\bf 19} (2002), 2399-2424,
[\href{https://arxiv.org/abs/gr-qc/0111049}{\tt arXiv:gr-qc/0111049}].

\vspace{-3mm}
\bibitem[Kr06]{Kraus06}
P. Kraus,
{\it Lectures on black holes and the $\mathrm{AdS}_3/\mathrm{CFT}_2$ correspondence},
Lecture Notes Phys. {\bf 755} (2008), 193-247,
[\href{https://arxiv.org/abs/hep-th/0609074}{\tt arXiv:hep-th/0609074}].

\vspace{-3mm}
\bibitem[KL06]{KL06}
P. Kraus and F. Larsen,
{\it Partition functions and elliptic genera from supergravity},
J. High Energy Phys. {\bf 0701} (2007), 002,
[\href{https://arxiv.org/abs/hep-th/0607138}{\tt arXiv:hep-th/0607138}].


\vspace{-3mm}

\bibitem[La17]{Landsman17}
K. Landsman,
{\it Foundations of quantum theory – From classical concepts to Operator algebras},
Springer Open 2017,
[\href{https://link.springer.com/book/10.1007/978-3-319-51777-3}{\tt doi:10.1007/978-3-319-51777-3}].


\vspace{-3mm}
\bibitem[LLL18]{LLL18}
Y.-Z. Li, S.-L. Li, and H. Lu,
{\it Exact Embeddings of JT Gravity in Strings and M-theory},
Eur. Phys. J. {\bf C 78} (2018), 791,
[\href{https://arxiv.org/abs/1804.09742}{\tt arXiv:1804.09742}].


\vspace{-3mm}
\bibitem[Li04]{Lin04}
H. Lin,
{\it The Supergravity Dual of the BMN Matrix Model},
J. High Energy Phys. {\bf 0412} (2004), 001,
[\href{https://arxiv.org/abs/hep-th/0407250}{\tt arXiv:hep-th/0407250}].


\vspace{-3mm}
\bibitem[LMNR19a]{LMNR19a}
Y. Lozano, N. Macpherson, C. Nunez, and A. Ramirez,
{\it $1/4$ BPS $\mathrm{AdS}_3/\mathrm{CFT}_2$}, \newline
[\href{https://arxiv.org/abs/1909.09636}{\tt arXiv:1909.09636}].

\vspace{-3mm}
\bibitem[LMNR19b]{LMNR19b}
Y. Lozano, N. Macpherson, C. Nunez, and A. Ramirez,
{\it Two dimensional $\mathcal{N}=(0,4)$ quivers dual to $\mathrm{AdS}_3$ solutions in massive IIA},
[\href{https://arxiv.org/abs/1909.10510}{\tt arXiv:1909.10510}].

\vspace{-3mm}
\bibitem[LMNR19c]{LMNR19c}
Y. Lozano, N. Macpherson, C. Nunez, and A. Ramirez,
{\it $\mathrm{AdS}_3$ solutions in massive IIA, defect CFTs and T-duality},
[\href{https://arxiv.org/abs/1909.11669}{\tt arXiv:1909.11669}].


\vspace{-3mm}
\bibitem[Mad92]{Madore92}
J. Madore,
{\it The Fuzzy sphere},
Class. Quant. Grav. {\bf 9} (1992), 69-88,
[\href{http://inspirehep.net/record/314358}{\tt spire:314358}].


\vspace{-3mm}
\bibitem[MSJVR03]{MSJVR03}
J. Maldacena, M. Sheikh-Jabbari and M. Van Raamsdonk,
{\it Transverse Fivebranes in Matrix Theory},
J. High Energy Phys. {\bf 0301} (2003), 038,
[\href{https://arxiv.org/abs/hep-th/0211139}{\tt arXiv:hep-th/0211139}].


\vspace{-3mm}
\bibitem[Mar04]{Marino04}
M. Marino,
{\it Chern-Simons Theory and Topological Strings},
Rev. Mod. Phys. {\bf 77} (2005), 675-720,
[\href{https://arxiv.org/abs/hep-th/0406005}{\tt arXiv:hep-th/0406005}].


\vspace{-3mm}
\bibitem[May72]{May72}
P. May,
{\it The geometry of iterated loop spaces},
Springer, Berlin, 1972.

\vspace{-3mm}
\bibitem[Mc75]{McDuff75}
D. McDuff,
{\it Configuration spaces of positive and negative particles},
Topology {\bf 14} (1975), 91-107,
[\href{https://doi.org/10.1016/0040-9383(75)90038-5}{\tt doi:10.1016/0040-9383(75)90038-5}].

\vspace{-3mm}
\bibitem[MN06]{McNamara06}
S. McNamara,
{\it Twistor Inspired Methods in Perturbative Field Theory and Fuzzy Funnels}, 2006, \newline
[\href{http://inspirehep.net/record/1351861}{\tt spire:1351861}].

\vspace{-3mm}
\bibitem[MPRS06]{MPRS06}
S. McNamara, C. Papageorgakis, S. Ramgoolam, and B. Spence,
{\it Finite $N$ effects on the collapse of fuzzy spheres},
J. High Energy Phys. {\bf 0605} (2006), 060,
[\href{https://arxiv.org/abs/hep-th/0512145}{\tt arXiv:hep-th/0512145}].


\vspace{-3mm}
\bibitem[MR19]{MerbisRiegler19}
W. Merbis and M. Riegler,
{\it Geometric actions and flat space holography},
[\href{https://arxiv.org/abs/1912.08207}{\tt arXiv:1912.08207}].

\vspace{-3mm}
\bibitem[MZ02]{MinahanZarembo02}
J. A. Minahan and K. Zarembo,
{\it The Bethe-Ansatz for $\mathcal{N}=4$ Super Yang-Mills},
J. High Energy Phys. {\bf 0303} (2003), 013,
[\href{https://arxiv.org/abs/hep-th/0212208}{\tt arXiv:hep-th/0212208}].


\vspace{-3mm}
\bibitem[Mo14]{Moore14}
G. Moore, {\it Physical Mathematics and the Future},
talk at \href{http://physics.princeton.edu/strings2014/}{Strings 2014},
\newline
\url{http://www.physics.rutgers.edu/~gmoore/PhysicalMathematicsAndFuture.pdf}

\vspace{-3mm}
\bibitem[Mo68]{Mostow68}
G. Mostow,
{\it Quasi-conformal mappings in $n$-space and the rigidity of hyperbolic space forms},
Pub. Math. IH{\'E}S,
{\bf 34} (1968),  53-104,
[\href{http://www.numdam.org/item/PMIHES_1968__34__53_0}{\tt numdam:PMIHES\_1968\_\_34\_\_53\_0}].

\vspace{-3mm}
\bibitem[Mu10]{Murakami10}
H. Murakami,
{\it An Introduction to the Volume Conjecture},
 { Interactions Between Hyperbolic Geometry, Quantum Topology and Number Theory},
Contemp. Math. {\bf 541}, Amer. Math. Soc., Providence, RI, 2011, \newline
[\href{https://arxiv.org/abs/1002.0126}{\tt arXiv:1002.0126}].

\vspace{-3mm}
\bibitem[MM01]{MurakamiMurakami01}
H. Murakami and J. Murakami,
{\it The Colored Jones Polynomial And The Simplicial Volume Of A Knot},
Acta Math. {\bf 186} (2001) 85-104,
[\href{https://projecteuclid.org/euclid.acta/1485891370}{\tt euclid.acta/1485891370}].

\vspace{-3mm}
\bibitem[My99]{Myers99}
R. Myers,
{\it Dielectric-Branes},
J. High Energy Phys. {\bf 9912} (1999), 022,
[\href{https://arxiv.org/abs/hep-th/9910053}{\tt arXiv:hep-th/9910053}].


\vspace{-3mm}
\bibitem[My01]{Myers01}
R. Myers,
{\it Nonabelian D-branes and Noncommutative Geometry},
J. Math. Phys. {\bf 42} (2001),  2781-2797,
[\href{https://arxiv.org/abs/hep-th/0106178}{\tt arXiv:hep-th/0106178}].



\vspace{-3mm}
\bibitem[Na19]{Narovlansky19}
V. Narovlansky,
{\it Towards a Solution of Large $N$ Double-Scaled SYK},
2019, \newline
{\small[\href{http://phsites.technion.ac.il/talks/fifth-israeli-indian-conference-on-string-theory2019/Narvolansky.pdf}{\tt phsites.technion.ac.il/talks/fifth-israeli-indian-conference-on-string-theory2019/Narvolansky.pdf}]}


\vspace{-3mm}
\bibitem[NH98]{NicolaiHelling98}
H. Nicolai and R. Helling,
{\it Supermembranes and M(atrix) Theory},
In: M. Duff et. al. (eds.),
{Nonperturbative aspects of strings, branes and supersymmetry},
World Scientific, Singapore,  1999, \newline
[\href{https://arxiv.org/abs/hep-th/9809103}{\tt arXiv:hep-th/9809103}].


\vspace{-3mm}
\bibitem[O'G98]{OGrady98}
K. O'Grady,
{\it Desingularized moduli spaces of sheaves on a K3, I \& II},
J. Reine Angew. Math. {\bf 512} (1999), 49-117,
[\href{https://arxiv.org/abs/alg-geom/9708009}{\tt arXiv:alg-geom/9708009},
\href{https://arxiv.org/abs/math/9805099}{\tt arXiv:math/9805099}].


\vspace{-3mm}
\bibitem[O'G00]{OGrady00}
K. O'Grady,
{\it A new six dimensional irreducible symplectic variety},
J. Algebraic Geom. {\bf 12} (2003), 435-505,
[\href{https://arxiv.org/abs/math/0010187}{\tt arXiv:math/0010187}].

\vspace{-3mm}
\bibitem[Os04]{Ostrik04}
V. Ostrik,
{\it Tensor categories (after P. Deligne)},
[\href{https://arxiv.org/abs/math/0401347}{\tt arXiv:math/0401347}].


\vspace{-3mm}
\bibitem[Pa06]{Papageorgakis06}
C. Papageorgakis,
{\it On matrix D-brane dynamics and fuzzy spheres},
2006,
\newline
[\href{https://ncatlab.org/nlab/files/Papageorgakis06.pdf}{\tt ncatlab.org/nlab/files/Papageorgakis06.pdf}]

\vspace{-3mm}
\bibitem[PR84]{PenroseRindler84}
R. Penrose and W. Rindler,
{\it Spinors and space-time -- Volume 1},
Cambridge University Press, 1984, \newline
[\href{https://doi.org/10.1017/CBO9780511564048}{\tt doi:10.1017/CBO9780511564048}].

\vspace{-3mm}
\bibitem[Pe18]{Petri18}
N. Petri,
{\it Supersymmetric objects in gauged supergravities},
[\href{https://arxiv.org/abs/1802.04733}{\tt arXiv:1802.04733}].


\vspace{-3mm}
\bibitem[Po02]{Polyakov02}
A. Polyakov,
{\it Gauge Fields and Space-Time},
Int. J. Mod. Phys. {\bf A17 S1} (2002), 119-136, \newline
[\href{https://arxiv.org/abs/hep-th/0110196}{\tt arXiv:hep-th/0110196}].


\vspace{-3mm}
\bibitem[PT91]{PorstTholen91}
H.-E. Porst and W. Tholen,
{\it Concrete Dualities}, in:
H. Herrlich, Hans-E. Porst (eds.)
{\it Category Theory at Work},
Heldermann Verlag, 1991,
[\href{http://www.heldermann.de/R&E/RAE18/ctw07.pdf}{\tt www.heldermann.de/R\&E/RAE18/ctw07.pdf}]


\vspace{-3mm}
\bibitem[RST04]{RST04}
S. Ramgoolam, B. Spence, and S. Thomas,
{\it Resolving brane collapse with $1/N$ corrections in non-Abelian DBI},
Nucl. Phys. {\bf B703} (2004), 236-276,
[\href{https://arxiv.org/abs/hep-th/0405256}{\tt arXiv:hep-th/0405256}].

\vspace{-3mm}
\bibitem[Re14]{Rebhan14}
A. Rebhan,
{\it The Witten-Sakai-Sugimoto model: A brief review and some recent results},
3rd International Conference on New Frontiers in Physics, Kolymbari, Crete, 2014,
[\href{https://arxiv.org/abs/1410.8858}{\tt arXiv:1410.8858}].


\vspace{-3mm}
\bibitem[RW06]{RobertsWillerton06}
J. Roberts and S. Willerton,
{\it On the Rozansky-Witten weight systems},
Algebr. Geom. Topol. {\bf 10} (2010), 1455-1519,
[\href{https://arxiv.org/abs/math/0602653}{\tt arXiv:math/0602653}].


\vspace{-3mm}
\bibitem[Ro18]{Rosenhaus18}
V. Rosenhaus,
{\it An introduction to the SYK model},
J. Phys. A: Math. Theor. {\bf  52},
[\href{https://arxiv.org/abs/1807.03334}{\tt arXiv:1807.03334}].


\vspace{-3mm}
\bibitem[RzW97]{RozanskyWitten97}
L. Rozansky and E. Witten,
{\it Hyper-K{\"a}hler geometry and invariants of 3-manifolds},
Selecta Math., New Ser. {\bf 3} (1997), 401-458,
[\href{https://arxiv.org/abs/hep-th/9612216}{\tt arXiv:hep-th/9612216}].



\vspace{-3mm}
\bibitem[SSu04]{SakaiSugimoto04}
T. Sakai and S. Sugimoto,
{\it Low energy hadron physics in holographic QCD},
Prog. Theor. Phys. {\bf 113} (2005),
 843-882,
[\href{https://arxiv.org/abs/hep-th/0412141}{\tt arXiv:hep-th/0412141}].


\vspace{-3mm}
\bibitem[SSu05]{SakaiSugimoto05}
T. Sakai and S. Sugimoto,
{\it More on a holographic dual of QCD},
Prog. Theor. Phys. {\bf 114} (2005), 1083-1118,
[\href{https://arxiv.org/abs/hep-th/0507073}{\tt arXiv:hep-th/0507073}].

\vspace{-3mm}
\bibitem[Sa05a]{Sa1}
H. Sati,
{\it M-theory and characteristic classes},
J. High Energy Phys.
{\bf 0508} (2005) 020,
\newline
[\href{https://arxiv.org/abs/hep-th/0501245}{\tt arXiv:hep-th/0501245}].

\vspace{-3mm}
\bibitem[Sa05b]{Sa2}
H. Sati,
{\it Flux quantization and the M-theoretic characters},
Nucl. Phys. {\bf B727} (2005) 461,
\newline
[\href{https://arxiv.org/abs/hep-th/0507106}{\tt arXiv:hep-th/0507106}].

\vspace{-3mm}
\bibitem[Sa06]{Sa3}
H. Sati,
{\it Duality symmetry and the form-fields in M-theory},
J. High Energy Phys.
{\bf 0606} (2006) 062,
\newline
[\href{https://arxiv.org/abs/hep-th/0509046}{\tt arXiv:hep-th/0509046}].

\vspace{-3mm}
\bibitem[Sa10]{tcu}
H. Sati,
{\it Geometric and topological structures related to M-branes},
Proc. Symp. Pure Math. {\bf 81} (2010) 181--236,
[\href{https://arxiv.org/abs/1001.5020}{\tt arXiv:1001.5020}] [math.DG].


\vspace{-3mm}
\bibitem[Sa13]{Sati13}
H. Sati,
{\it Framed M-branes, corners, and topological invariants},
J. Math. Phys. {\bf 59} (2018), 062304, \newline
[\href{https://arxiv.org/abs/1310.1060}{\tt arXiv:1310.1060}] [hep-th].

\vspace{-3mm}
\bibitem[SS17]{SS17}
H. Sati and U. Schreiber,
{\it Lie $n$-algebras of BPS charges},
J. High Energy Phys. {\bf 2017} (2017), 87, \newline
[\href{https://arxiv.org/abs/1507.08692}{\tt  arXiv:1507.08692}] [math-ph].

\vspace{-3mm}
\bibitem[SS19a]{SS19a}
H. Sati and U. Schreiber,
\href{https://ncatlab.org/schreiber/show/Equivariant+Cohomotopy+implies+orientifold+tadpole+cancellation}
{\it Equivariant Cohomotopy implies orientifold tadpole cancellation}, \newline
[\href{https://arxiv.org/abs/1909.12277}{\tt arXiv:1909.12277}].

\vspace{-3mm}
\bibitem[SS19b]{SS19b}
H. Sati and U. Schreiber.
\href{https://ncatlab.org/schreiber/show/Lift+of+fractional+D-brane+charge+to+equivariant+Cohomotopy+theory}
{\it Lift of fractional D-brane charge to equivariant Cohomotopy theory}, \newline
[\href{https://arxiv.org/abs/1812.09679}{\tt arXiv:1812.09679}].


\vspace{-3mm}
\bibitem[Saw04]{Sawon04}
J. Sawon,
{\it Rozansky-Witten invariants of hyperk{\"a}hler manifolds},
PhD thesis, University of Cambridge, 2000,
[\href{https://arxiv.org/abs/math/0404360}{\tt arXiv:math/0404360}].

\vspace{-3mm}
\bibitem[Sc14]{Schreiber14}
U. Schreiber
{\it Quantization via Linear homotopy types},
talk at ESI Vienna, 2014,
[\href{https://arxiv.org/abs/1402.7041}{\tt arXiv:1402.7041}].


\vspace{-3mm}
\bibitem[Sc20]{Schreiber20}
U. Schreiber, 
{\it Microscopic brane physics from Cohomotopy theory},
talk at: 
H. Sati (org.), \href{https://hisham-sati.github.io/M-theory-and-Mathematics/}{\it M-Theory and Mathematics}. 
NYU AD Research Institute, January 27-30, 2020
\newline
[\href{https://ncatlab.org/schreiber/files/Schreiber-MTheoryMathematics2020-v200126.pdf}{ncatlab.org/schreiber/files/Schreiber-MTheoryMathematics2020-v200126.pdf}]



\vspace{-3mm}
\bibitem[Sc01]{Schwarz01}
J. Schwarz,
{\it Comments on Born-Infeld Theory},
\href{http://inspirehep.net/record/944370}{Proceedings of Strings 2001},
[\href{https://arxiv.org/abs/hep-th/0103165}{\tt arXiv:hep-th/0103165}].



\vspace{-3mm}
\bibitem[Seg73]{Segal73}
G. Segal,
{\it Configuration-spaces and iterated loop-spaces},
Invent. Math. {\bf 21} (1973), 213-221.

%\vspace{-3mm}
%\bibitem[Se97]{Segal97}
%G. Segal, {\it Topology of the space of ${\rm SU}(2)$-monopoles in $R^3$},
%Geometry and physics (Aarhus, 1995), 141-147, Lecture Notes in Pure and Appl. Math.,
%184, Dekker, New York, 1997.
%
%\vspace{-3mm}
%\bibitem[SeS96]{SegalSelby}
%G. Segal and A. Selby,
%{\it The cohomology of the space of magnetic monopoles},
%Comm. Math. Phys. {\bf 177} (1996), no. 3, 775-787.

\vspace{-3mm}
\bibitem[SW96]{SeibergWitten96}
N. Seiberg and E. Witten,
{\it Gauge Dynamics And Compactification To Three Dimensions},
In: J. M. Drouffe, J. B. Zuber (eds.),
{The mathematical beauty of physics: A memorial volume for Claude Itzykson},
Proceedings, Conference, Saclay, France, June 5-7, 1996,
[\href{https://arxiv.org/abs/hep-th/9607163}{\tt arXiv:hep-th/9607163}].


\vspace{-3mm}
\bibitem[Sel09]{Selinger09}
P. Selinger,
{\it A survey of graphical languages for monoidal categories},
in: B. Coecke (ed.),
{New Structures for Physics},
Lecture Notes in Physics, vol 813, Springer, Berlin, Heidelberg, 2010,
[\href{https://arxiv.org/abs/0908.3347}{\tt arXiv:0908.3347}].

\vspace{-3mm}
\bibitem[Sen07]{Sen07}
A. Sen,
{\it Black Hole Entropy Function, Attractors and Precision Counting of Microstates},
Gen. Rel. Grav. {\bf 40} (2008),  2249-2431,
[\href{https://arxiv.org/abs/0708.1270}{\tt arXiv:0708.1270}].





\vspace{-3mm}
\bibitem[SV96]{SV96}
A. Strominger and C. Vafa,
{\it Microscopic Origin of the Bekenstein-Hawking Entropy},
Phys. Lett. {\bf B379} (1996),  99-104,
[\href{https://arxiv.org/abs/hep-th/9601029}{\tt arXiv:hep-th/9601029}].

\vspace{-3mm}
\bibitem[Su16]{Sugimoto16}
S. Sugimoto,
{\it Skyrmion and String theory},
chapter 15 in: M. Rho, Ismail Zahed (eds.),
{The Multifaceted Skyrmion}, World Scientific, Singapore,  2016,
[\href{https://doi.org/10.1142/9710}{\tt doi:10.1142/9710}].

\vspace{-3mm}
\bibitem[Sw17]{Swart17}
J. Swart,
{\it Introduction to Quantum Probability}, 2017
[\href{http://staff.utia.cas.cz/swart/dict.pdf}{staff.utia.cas.cz/swart/dict.pdf}]


\vspace{-3mm}
\bibitem[Ta01]{Taylor01}
W. Taylor,
{\it M(atrix) Theory: Matrix Quantum Mechanics as a Fundamental Theory},
Rev. Mod. Phys. {\bf 73} (2001), 419-462,
[\href{https://arxiv.org/abs/hep-th/0101126}{\tt arXiv:hep-th/0101126}].

\vspace{-3mm}
\bibitem[TvR99]{TaylorRaamsdonk00}
W. Taylor and M. Van Raamsdonk,
{\it Multiple D$p$-branes in Weak Background Fields},
Nucl. Phys. {\bf B573} (2000), 703-734,
[\href{https://arxiv.org/abs/hep-th/9910052}{\tt arXiv:hep-th/9910052}].


\vspace{-3mm}
\bibitem[TW06]{ThomasWard06}
S. Thomas and J. Ward,
{\it Electrified Fuzzy Spheres and Funnels in Curved Backgrounds},
JHEP 0611:019, 2006
[\href{https://arxiv.org/abs/hep-th/0602071}{arXiv:hep-th/0602071}]


\vspace{-3mm}
\bibitem['tH74]{tHooft74}
G. 't Hooft,
{\it A Planar Diagram Theory for Strong Interactions},
Nucl. Phys. {\bf B72} (1974) 461-473, \newline
[\href{http://inspirehep.net/record/80491}{\tt spire:80491}].

\vspace{-3mm}
\bibitem[To99]{Tong99}
D. Tong,
{\it Three-Dimensional Gauge Theories and ADE Monopoles},
Phys. Lett. {\bf B448} (1999), 33-36,
[\href{https://arxiv.org/abs/hep-th/9803148}{\tt arXiv:hep-th/9803148}].


\vspace{-3mm}
\bibitem[TZ06]{TradlerZeinalian04}
T. Tradler and M. Zeinalian,
{\it On the Cyclic Deligne Conjecture},
J. Pure Appl. Alg. {\bf 204} (2006), 280-299,
[\href{https://arxiv.org/abs/math/0404218}{\tt arXiv:math/0404218}].


\vspace{-3mm}
\bibitem[TZ07]{TradlerZeinalian06}
T. Tradler and M. Zeinalian,
{\it Algebraic String Operations}, $K$-Theory {\bf 38} (2007), 59-82, \newline
[\href{https://arxiv.org/abs/math/0605770}{\tt arXiv:math/0605770}].


\vspace{-3mm}
\bibitem[Ts97]{Tseytlin97}
A. Tseytlin,
{\it On non-Abelian generalization of Born-Infeld action in string theory}, Nucl. Phys.
{\bf B501} (1997), 41-52,
[\href{https://arxiv.org/abs/hep-th/9701125}{\tt arXiv:hep-th/9701125}].


\vspace{-3mm}
\bibitem[Va94]{Vaintrob94}
A. Vaintrob,
{\it Vassiliev knot invariants and Lie S-algebras},
Math. Res. Lett. {\bf 1} (1994), 579-595,
\newline
[\href{https://pdfs.semanticscholar.org/bdc3/ac1d8da476245e2408e481a70b115b3e9aab.pdf}{\tt pdfs.semanticscholar.org/bdc3/ac1d8da476245e2408e481a70b115b3e9aab.pdf}]


\vspace{-3mm}
\bibitem[Va96]{Vafa96}
C. Vafa,
{\it Instantons on D-branes},
Nucl. Phys. {\bf B463} (1996), 435-442,
[\href{https://arxiv.org/abs/hep-th/9512078}{\tt arXiv:hep-th/9512078}].


\vspace{-3mm}
\bibitem[VW94]{VafaWitten94}
C. Vafa and E. Witten,
{\it A Strong Coupling Test of S-Duality},
Nucl. Phys. {\bf B431} (1994), 3-77,  \newline
[\href{https://arxiv.org/abs/hep-th/9408074}{\tt arXiv:hep-th/9408074}].


\vspace{-3mm}
\bibitem[Va04]{Varadarajan04}
V. Varadarajan,
{\it Supersymmetry for mathematicians: An introduction},
Courant Lecture Notes in Mathematics,
AMS, 2004,
[\href{http://dx.doi.org/10.1090/cln/011}{\tt doi:10.1090/cln/011}].


\vspace{-3mm}
\bibitem[Va92]{Vassiliev92}
V. Vassiliev,
{\it Complements of discriminants of smooth maps: topology and applications},
Amer. Math. Soc.. Providence, RI,  1992.


\vspace{-3mm}
\bibitem[Vo11]{Vogel11}
P. Vogel,
{\it Algebraic structures on modules of diagrams},
J. Pure Appl. Alg. {\bf 215} (2011),  1292-1339, \newline
[\href{https://doi.org/10.1016/j.jpaa.2010.08.013}{\tt doi:10.1016/j.jpaa.2010.08.013}].


\vspace{-3mm}
\bibitem[Vo13]{Volic13}
I. Voli{\'c},
{\it Configuration space integrals and the topology of knot and link spaces},
Morfismos {\bf 17} (2013), 1-56,
[\href{https://arxiv.org/abs/1310.7224}{\tt arXiv:1310.7224}].


\vspace{-3mm}
\bibitem[Wi92]{Witten92}
E. Witten,
{\it Chern-Simons Gauge Theory As A String Theory},
Progr. Math. {\bf 133} (1995),  637-678,  \newline
[\href{https://arxiv.org/abs/hep-th/9207094}{\tt arXiv:hep-th/9207094}].


\vspace{-3mm}
\bibitem[Wi98]{Witten98}
E. Witten,
{\it Anti-de Sitter Space, Thermal Phase Transition, And Confinement In Gauge Theories},
Adv. Theor. Math. Phys. {\bf 2} (1998), 505-532,
[\href{https://arxiv.org/abs/hep-th/9803131}{\tt arXiv:hep-th/9803131}].

\vspace{-3mm}
\bibitem[Wi01]{Witten01}
E. Witten,
{\it Multi-Trace Operators, Boundary Conditions, And AdS/CFT Correspondence}, \newline
[\href{https://arxiv.org/abs/hep-th/0112258}{\tt arXiv:hep-th/0112258}].

\vspace{-3mm}
\bibitem[Wi19]{Witten19}
E. Witten,
in: G. Farmelo,
{\it The Universe Speaks in numbers}, interview 5, 2019,
\newline
[\href{https://grahamfarmelo.com/the-universe-speaks-in-numbers-interview-5}{\tt grahamfarmelo.com/the-universe-speaks-in-numbers-interview-5}]
at 21:15.


\end{thebibliography}
\end{document}